\documentclass[12pt]{article}

\usepackage[top=30mm, bottom=30mm, left=30mm, right=30mm]{geometry}
\usepackage{url}
\usepackage[numbers]{natbib}
\usepackage{authblk}
\usepackage{amsthm,mathtools}
\usepackage{bm,color}
\usepackage{hyperref}

\usepackage{graphicx}%
\usepackage{multirow}%
\usepackage{amsmath,amssymb,amsfonts}%
\usepackage{amsthm}%
\usepackage{mathrsfs}%
\usepackage[title]{appendix}%
\usepackage{xcolor}%
\usepackage{textcomp}%
\usepackage{manyfoot}%
\usepackage{booktabs}%
\usepackage{algorithm}%
\usepackage{algorithmicx}%
\usepackage{algpseudocode}%
\usepackage{listings}%
\usepackage{enumitem}%

\newcommand{\de}{\mathrm{d}}

\newcommand{\bs}{\bm{\sigma}}
\newcommand{\bet}{\bm{\eta}}
\newcommand{\bmu}{\bm{\mu}}

\newcommand{\bq}{\bm{q}}
\newcommand{\bS}{\bm{\Sigma}}
\newcommand{\bg}{\bm{\gamma}}

\newcommand{\E}{\mathbb{E}}

\newcommand{\cL}{\mathcal{L}}

\newcommand{\cS}{\mathcal{S}}

\newcommand{\N}{\mathbb{N}}

\newcommand{\R}{\mathbb{R}}
\newcommand{\Z}{\mathbb{Z}}
\newcommand{\Tr}{\operatorname{Tr}}

\newcommand{\eps}{\varepsilon}

\DeclareMathOperator{\pP}{\mathbb{P}}

\newcommand{\bx}{\bm{x}}
\newcommand{\by}{\bm{y}}

\newcommand{\bR}{\bm{R}}

\newcommand{\bQ}{\bm{Q}}

\newcommand{\bA}{\bm{A}}
\newcommand{\bO}{\bm{O}}

\newcommand{\bU}{\bm{U}}
\newcommand{\bC}{\bm{C}}

\newcommand{\bz}{\bm{z}}

\newcommand{\bX}{\bm{X}}
\newcommand{\bY}{\bm{Y}}
\newcommand{\bZ}{\bm{Z}}
\newcommand{\bB}{\bm{B}}
\newcommand{\bT}{\bm{T}}

\newcommand{\bxi}{\bm{\xi}}

\newcommand{\diag}{\mathrm{diag}}

\DeclareMathOperator{\supp}{\mathrm{supp}}

\DeclareMathOperator*{\argmax}{arg\,max}

\numberwithin{equation}{section}

\theoremstyle{plain}
\newtheorem{theorem}{Theorem}[section]
\newtheorem{lemma}[theorem]{Lemma}
\newtheorem{corollary}[theorem]{Corollary}
\newtheorem{proposition}[theorem]{Proposition}

\theoremstyle{definition}

\theoremstyle{remark}
\newtheorem{remark}[theorem]{Remark}

\begin{document}

\title{A multiscale cavity method for sublinear-rank symmetric matrix factorization}

\author[1]{Jean Barbier}

\author[2,3]{Justin Ko}

\author[1]{Anas A. Rahman*}

\affil[1]{The Abdus Salam ICTP}

\affil[2]{University of Waterloo}

\affil[3]{ENS Lyon}

\maketitle

\abstract{We consider a statistical model for symmetric matrix factorization with additive Gaussian noise in the high-dimensional regime, where the rank of the signal matrix to infer $M$ scales with its size $N$ as $M=\mathrm{o}(\sqrt{\ln N})$. Allowing for an $N$-dependent rank offers new challenges and requires new methods. Working in the Bayes-optimal setting, we show that whenever the signal has i.i.d. entries, the limiting mutual information between signal and data is given by a variational formula involving a rank-one replica symmetric potential. In other words, from the information-theoretic perspective, the case of a (slowly) growing rank is the same as when $M=1$ (namely, the standard spiked Wigner model). The proof is primarily based on a novel multiscale cavity method allowing for growing rank along with some information-theoretic identities on worst noise for the vector Gaussian channel. We believe that the cavity method developed here will play a role in the analysis of a broader class of inference and spin models where the degrees of freedom are large arrays instead of vectors.}

\section{Introduction}\label{s1}

Reconstructing a low-rank signal matrix observed through a noisy channel with highest possible accuracy is of fundamental importance across a range of disciplines including statistical inference, machine learning, signal processing, and information theory. This task encompasses a variety of models that have attracted considerable attention over the last decade, such as sparse principal component analysis and $\Z_2$ synchronization \cite{sPCA1,sPCA2}, the stochastic block model for community detection \cite{SBM1,SBM2}, and submatrix localization \cite{localization1,localization2}, just to cite a few. A common theme unifying these models is that they are either examples of, or are otherwise closely related to, the general notion of spiked matrix models \cite{johnstone,BBP,peche} of the form ``signal plus noise'' that were introduced at the turn of this century.

In this work, we consider the archetypal \textit{spiked Wigner model}, where one observes a symmetric deformation of a standard Wigner matrix. The data is generated as
\begin{equation} \label{spikedwignermodel}
	\bY = \sqrt{\frac{\lambda}{N}} \bX_0 \bX_0^\intercal+\bZ,
\end{equation}
where $\bX_0\in\R^{N\times M}$ is the signal matrix distributed according to some prior distribution $\pP_{X,N,M}$, $\bZ\in\R^{N\times N}$ is a standard Wigner matrix with $\bZ_{ii}\sim\mathcal{N}(0,2)$ ($1\leq i\leq N)$ and $\bZ_{ij}=\bZ_{ji}\sim\mathcal{N}(0,1)$ ($1\leq i<j\leq N$) independently and identically distributed (i.i.d.), $\lambda\geq0$ is the signal-to-noise ratio (SNR), and the scaling by $\sqrt{N}$ is such that the signal and noise are of comparable magnitudes. From another perspective, we may say that we observe the spike $\bX_0$ corrupted by additive Gaussian noise -- note that results concerning this model extend to a wide range of noisy observation channels due to recently proven universality principles \cite{universality1,universality2,universality3,rank1universality}.

We study the \textit{free entropy} (equivalently, up to a simple additive term, the mutual information $I(\bX_0;\bY)$, see, e.g., \cite{LelargeMiolane})
\begin{align}
	F_N(\lambda) := F_{N,M}(\lambda) = \frac{1}{NM} \E_{\bZ,\bX_0} \ln Z_{N,M} \label{FrenEnt}    
\end{align}
of the model \eqref{spikedwignermodel} in the thermodynamic limit $N,M\to\infty$ with the rank $M=M_N$ growing sufficiently slowly with $N$. Here,
\begin{equation}
Z_{N,M}:=\int_{\R^{N\times M}} e^{H_N(\bX)}\,\de\!\pP_{X,N,M}(\bX) \label{ZNM}
\end{equation}
denotes the \textit{partition function} associated with the \textit{Hamiltonian}
\begin{equation}
H_N(\bX):=\frac{1}{2} \Tr\Big( \sqrt{\frac{ \lambda }{N}}  \bZ \bX \bX^\intercal  + \frac{\lambda}{N}  \bX_0 \bX_0^\intercal \bX \bX^\intercal- \frac{\lambda}{2N}   \bX \bX^\intercal \bX \bX^\intercal  \Big); \label{Hamiltonian}
\end{equation}
in the language of Bayesian inference rather than statistical physics, the partition function is simply the normalization of the posterior distribution $\pP_{X\mid Y,N,M}(\bX\mid\bY(\bX_0,\bZ))$ and the Hamiltonian is the log-likelihood. In order to take advantage of replica symmetry (i.e., concentration of measure) \cite{LelargeMiolane,overlapconcentration,barbier2022strong}, we moreover work in the \textit{Bayes-optimal setting}, which means that in addition to the observed data $\bY$, the statistician charged with reconstructing $\bX_0$ from $\bY$ also has knowledge of the form \eqref{spikedwignermodel} of $\bY$, the prior distribution $\pP_{X,N,M}$ of the signal $\bX_0$, the SNR $\lambda$, and the rank $M$ of $\bX_0\bX_0^\intercal$.

Focusing on the Bayesian viewpoint, various applications of (the adaptive) interpolation methods, the cavity method, and the Approximate Message Passing (AMP) algorithm have been used to study (various generalizations of) the spiked Wigner model. Studies commenced in the rank $M=1$ case \cite{korada2009exact,universality2,SBM1,replicaproof,spikedtensor} before progressing to the low-rank (i.e., finite-rank) case \cite{LelargeMiolane,adaptiveinterpolation}, with recent studies venturing into the regime of growing rank \cite{reeves, sublinear_sphericalint,Maillard_2022,bodin2023gradient}. In particular, an $M$-dimensional variational formula for the limiting free entropy $\lim_{N\to\infty}F_N(\lambda)$ can be surmised from \cite{reeves} when $M=\mathrm{o}(N^{1/20})$. Our present goal is to continue this line of research by proving a scalar equivalent of the aforementioned formula (thereby circumventing the intractability of the $M$-dependence seen therein) with regime of veracity $M=\mathrm{o}(\sqrt{\ln N})$ -- this is a significant milestone towards understanding the recent conjecture based on the replica method of \cite{farzad} that the rank-$M$ spiked Wigner model should behave as its rank-one counterpart, so long as the rank is sublinear, i.e., when $M=\mathrm{o}(N)$.

It has long been known and utilized that there is a close interplay between statistical inference and the theory of spin glasses. In particular, the free entropy of the finite-rank spiked Wigner model is closely related to the free energies of vector spin glass models, with the main technical challenge in studies of the latter being the absence of replica symmetry. Variational formulae for free energies of various vector spin glass models have been proven using synchronization of overlap matrices \cite{PVS,PPotts,JSphere, kocs, GGVec, vecSphereTAP} and Hamilton--Jacobi equations \cite{mourratchen_noncvx,MourratVec,chen2023selfoverlap,chen2023cvx}; the relevant order parameters are increasing matrix paths, which are more challenging to deal with than the matrix-valued order parameters typical of Bayes-optimal inference problems. Said variational formulae have moreover been recently analyzed and simplified in \cite{bates2023parisi,chen2023potts,tucazhou_vecspin} through exploitation of inherent symmetries within the studied models in order to reduce dimensionality of the order parameters. Our main technical contribution in this line of work is a generalization of the \textit{Aizenman--Sims--Starr scheme} \cite{ASS} to multiscale mean field models, which in this context are vector spin glass models with dimension-dependent vector spin coordinates. Using the cavity method instead of, say, the adaptive interpolation method \cite{adaptiveinterpolation,reeves} means that we require only posterior (thermal) concentration of our order parameter, the overlap matrix, rather than full (quenched) concentration. This in turn enables us to employ a scalar perturbation parameter (as opposed to a matrix one) in the proof of concentration, leading to more tractable computations and better rates of convergence. Our secondary contribution is a proof of the fact that our variational formula for the limiting free entropy reduces to that of the rank-one spiked Wigner model, owing to a set of simple information-theoretic inequalities. Said identities are, to the best of our knowledge, not in the literature and we expect Corollary~\ref{Cor1}, say, to be a valuable addition to the literature on worst noises \cite{diggavicover97,diggavicover01,worstnoise}.

The reduction to rank one mentioned above follows from the fact that the entries of the signal $\bX_0$ are i.i.d., which induces exchangeability in the entries of the overlap matrix. This is similar to recent simplifications of vector spin Parisi formulae obtained in \cite{bates2023parisi,chen2023potts} through the utilization of underlying symmetries in the Potts spin glass. Combining these simplifications with our multiscale cavity method has the potential to extend the corresponding finite-rank vector spin models to sublinear-rank multiscale models, which is a natural step towards understanding the challenging extensive-rank regime $M=\Theta(N)$, see \cite{tarmoun_gradientoverparam,Maillard_2022,bodin2023gradient,farzad,DL2,camilli2023matrix,camilli2024decimation}. In this vein, we foresee that the relative tractability of our computations will also enable the theory of asymmetric matrix and tensor factorization \cite{tensorPCA,Miolane17,spikedtensor,adaptiveinterpolation, donoho2023optimal,feldman2023spiked,montanari2024fundamental} to be pushed into the growing rank regime. Indeed, we find some potential extensions of our methods to particular related settings are readily observable: In the case of asymmetric matrix factorization \cite{Miolane17,Miolane19,donoho2023optimal, feldman2023spiked,montanari2024fundamental}, we expect the finite-rank cavity method of \cite{Miolane17} to extend to the growing rank regime via our multiscale formalism. Likewise for the cavity method employed in the study of symmetric tensor factorization \cite{spikedtensor}. The missing ingredient in both of these cases is our result reducing the rank-$M$ model to its rank-one counterpart. In the present setting, this is done by showing equivalence between the suprema of a rank-$M$ potential and its rank-one counterpart, introduced below. In \cite{Miolane17}, one is instead interested in the supremum of the sum of two such potentials coupled through a quadratic term, while in \cite{spikedtensor}, the relevant potential involves Hadamard powers of the matrix argument $\bQ$, so that the eigenvalue-based arguments of Section \ref{s3} cannot be used. Further investigation is needed to ascertain if the ideas of Section \ref{s3} can be adapted to overcome these challenges.

\section{Setting and main results} \label{s2}
\subsection{The sublinear-rank spiked Wigner model and the replica symmetric potential} \label{s2.1}
We study the $N,M\to\infty$ limit of the free entropy $F_N(\lambda)$ \eqref{FrenEnt} of the spiked Wigner model \eqref{spikedwignermodel} with rank $M=\mathrm{o}(\sqrt{\ln N})$ and signal of i.i.d.~entries, i.e., the prior distribution $\pP_{X,N,M}$ on $\bX_0=(\bx_{0,i}\in\R^M)_{i\le N}=(\bx_{0,ij}\in\R)_{i\le N,j\le M}$ fully factorizes:
\begin{equation}
\pP_{X,N,M}(\bX_0) = \prod_{i=1}^N \pP_{X,M}(\bx_{0,i}) = \prod_{i = 1}^N \prod_{j = 1}^M  \pP_{X}(\bx_{0,ij}). \label{iidprior}
\end{equation}
In addition, we take $\pP_X$ to be a centered distribution with $D$\textit{-bounded support}: There exists $0\le D<\infty$ such that $\supp\pP_X\subseteq[-D,D]$.

An application of the (non-rigorous) replica method \cite{spinglass,talagrand} suggests that the large $N$ limit of $F_N(\lambda)$ is given by a variational formula written in terms of the rank-$M$ \textit{replica symmetric potential} $F^{\mathrm{RS}}_M:\cS_M\times \R_{+}\mapsto \R$ defined by (with $\cS_M$ the set of symmetric positive semidefinite $M\times M$ matrices)
\begin{align}
	F^{\mathrm{RS}}_M(\bQ,\lambda)&:=  \frac{1}{M} \E_{\bz,\bx_0} \ln  Z_M^{\mathrm{RS}} - \frac{\lambda}{4M}  \Tr\bQ^2, \label{RSpot}
	\\ Z_M^{\mathrm{RS}}&:= \int_{\R^M} e^{\sqrt{\lambda}  \bx^\intercal \sqrt{\bQ} \bz   + \lambda  \bx_0^\intercal \bQ \bx  - \frac{\lambda}{2}  \bx^\intercal \bQ \bx}\,  \de\!\pP_{X,M}(\bx), \nonumber   
\end{align}
where $\bz \in \R^{M}$ is a standard Gaussian vector and $\bx_0\sim \pP_{X,M}$ \cite{spikedtensor, adaptiveinterpolation}. However, we find that it is instead given by a formula in terms of the simpler rank-one analog
\begin{align}
	F_1^{\mathrm{RS}}(q,\lambda)&:=   \E_{z,x_0} \ln Z_1^{\mathrm{RS}}- \frac{\lambda}{4}  q^2, \label{F1RSpot} \\
	Z_1^{\mathrm{RS}}&:=\int_{-\infty}^{\infty} e^{ \sqrt{\lambda q} z x + \lambda  qx_0 x - \frac{1}{2} \lambda q x^2 } \,  \de\!\pP_{X}(x), \nonumber
\end{align}
with $z\sim\mathcal{N}(0,1)$ and $x_0\sim \pP_{X}$. Thus, we have as our main result that, under our hypotheses above in addition to two further technical hypotheses concerning vector and scalar analogs of our channel \eqref{spikedwignermodel}, the $N\to\infty$ limit of the free entropy $F_N(\lambda)$ is given by the rank-one replica symmetric potential $F_1^{\mathrm{RS}}(q,\lambda)$ at its supremum over $q$.

\begin{theorem}[Rank-one replica formula for the growing-rank spiked Wigner model] \label{thrm1}
Assume the following hypotheses:
\begin{enumerate}[label=\upshape(H\arabic*),ref=(H\arabic*)]
\item\label{H1} The rank $M=M_N\in\N$ is such that $M=\mathrm{o}(\sqrt{\ln N})$ and $M_{N+1}-M_N \le 1$ for all $N\in\N$.
\item\label{H2} The prior distribution $\pP_{X,N,M}$ on $\bX_0$ fully factorizes as in equation \eqref{iidprior}.
\item\label{H3} The distribution $\pP_X$ is centered with $D$-bounded support.
\item\label{H4} The distribution $\pP_X$ is such that for $\pP_{X,M}=\pP_X^{\otimes M}$ and all constant $m\le M$,
\begin{equation} \label{phiM}
\phi_m(\lambda):=\sup_{\bQ\in\cS_m}F_m^{\mathrm{RS}}(\bQ,\lambda)
\end{equation}
is real analytic in $\lambda$ on $[0,\infty)$ except for possibly one critical point (phase transition) $0<\lambda_c<\infty$, independent of $m$.
\item\label{H5} Letting $x_0\sim\pP_X$ and $z\sim\mathcal{N}(0,1)$, the minimum mean-square error (MMSE) of the scalar Gaussian channel $y=\sqrt{\lambda}x_0+z$,
\begin{equation} \label{scalarMMSE}
\mathrm{mmse}(\lambda):=\E_{z,x_0}[(x_0-\E[x_0\mid y])^2],
\end{equation}
is such that $\lambda^2\mathrm{mmse}(\lambda)$ is monotone on $[\lambda_L'',\infty)$ for some $\lambda_L''>0$.
\end{enumerate}

Setting $\rho:=\E_{\pP_X}X^2$, the limiting free entropy \eqref{FrenEnt} of the spiked Wigner model~\eqref{spikedwignermodel} is then given in terms of the replica symmetric potential \eqref{F1RSpot} by
\begin{equation} \label{mainresult}
\lim_{N\to\infty}F_N(\lambda)=\lim_{N\to\infty}F_{N,M_N}(\lambda)=\sup_{q\in[0,\rho]}F_1^{\mathrm{RS}}(q,\lambda).
\end{equation}

As a consequence, we have the following formula for the limiting minimum mean-square error of the spiked Wigner model:
\begin{align}
 \lim_{N\to\infty}\mathrm{MMSE}_{N,M_N}(\lambda)  \!&\;= \rho^2-q^*(\lambda)^2,\quad 0\le\lambda\ne\lambda_c,\nonumber
 \\\mathrm{MMSE}_{N,M}(\lambda)&:=\frac{1}{N^2M}\E_{\bZ,\bX_0}\lVert\bX_0 \bX_0^\intercal-\E[\bX_0 \bX_0^\intercal \mid \bY]\rVert_{\mathrm{F}}^2, \label{MMSEspikedwigner}
\end{align}
where $q^*(\lambda):=\argmax_{q\in[0,\rho]}F_1^{\mathrm{RS}}(q,\lambda)$ is the (unique \cite[Prop.~17]{LelargeMiolane}) maximizer of $F_1^{\mathrm{RS}}(q,\lambda)$ and $\lVert\cdot\rVert_{\mathrm{F}}$ denotes the Frobenius norm.
\end{theorem}
The statement on the minimum mean-square error follows from the I-MMSE formula \cite{guoshamaiverdu} and the arguments of \cite{LelargeMiolane}.

We comment on our hypotheses:
\begin{enumerate}
\item We require $M=\mathrm{o}(N^{1/4})$ to achieve thermal concentration of the overlap matrix
\begin{equation} \label{overlap}
\bR_{10}:=\frac{1}{N}\bX^\intercal\bX_0=\frac{1}{N}\sum_{i=1}^N\bx_i\bx_{0,i}^\intercal
\end{equation}
with respect to the Gibbs average corresponding to a particular perturbation of the Hamiltonian $H_N(\bX)$ \eqref{Hamiltonian} and $M=\mathrm{o}(\sqrt{\ln N})$ for the error terms arising from the cavity computation to vanish in the large $N$ limit. We believe that this growth rate for the rank is not fundamental and is rather a limitation of the proof methods. Indeed, we conjecture that Theorem \ref{thrm1} holds for $M=\mathrm{o}(N)$, in line with \cite{sublinear_sphericalint}. The requirement that $M_{N+1}-M_N\le 1$ is in keeping with the natural notion that the sequence of matrices $\bX_0$ has its dimensions $M,N$ growing by increments of $1$, with $N$ growing steadily.

\item Symmetry properties of $\pP_{X,N,M}$ owing to its factorized nature encourage the thermal concentration of $\bR_{10}$ onto a scalar multiple of the identity, which in turn induces an equivalence $\sup_{q\in[0,\rho]}F_M^{\mathrm{RS}}(q I_M,\lambda)\equiv\sup_{q\in[0,\rho]}F_1^{\mathrm{RS}}(q,\lambda)$. The factorization of the prior is thus key to reducing the dimensionality of the variational formula \eqref{mainresult}.

\item Requiring $\pP_X$ to have $D$-bounded support is not very restrictive and is only a technical constraint; prior distributions $\pP_X$ with unbounded support can be treated by truncating them to have $D$-bounded support, working in the finite $D$ regime as in this paper, and then taking $D\to\infty$, see, e.g., \cite{LelargeMiolane,adaptiveinterpolation}.

\item It is believed that for sufficiently slowly growing rank, $\lim_{N\to\infty}F_{N}(\lambda)$ exhibits a single rank-independent phase transition for many natural choices of prior $\pP_X$ (uniform, Rademacher, Gaussian, simple mixtures, etc.), see \cite{Miolane19} and references therein; for rank proportional to $N$, one still observes a unique phase transition, but it has a rank-dependent location \cite{BCKO2025}. We assume there is at most one rank-independent phase transition and translate this assumption on the model into one for $\phi_m(\lambda)$ ($m\le M$). It is difficult to prove this hypothesis in general and it may not hold for certain (exotic) priors. Nonetheless, a convincing check for its validity is to numerically evaluate $\phi_2(\lambda)$, which is easy, and look for a unique first or second order phase transition (a discontinuity in its first or second order derivative, which are the standard transitions in high-dimensional inference problems \cite{zdeborova2016statistical}). The same behavior is then expected for any fixed $m$ -- additional phase transitions can only appear in a large system limit. This is sufficient for our argument (in particular, Theorem~\ref{thrm2}). To improve clarity, we will often refer to hypothesis~\ref{H4} when it suffices to take a weaker hypothesis, such as when we require real analyticity of $\phi_m(\lambda)$ only at low or high SNR, or at exactly $m=M$.

\item As scalar channels are far more tractable than higher dimensional ones, it is relatively easy to check monotonicity of $\lambda^2\mathrm{mmse}(\lambda)$ in the high SNR regime. In particular, hypothesis \ref{H5} holds when $\mathrm{mmse}(\lambda)$ is real analytic on $[\lambda_L'',\infty)$ since $\mathrm{mmse}(\lambda)=\mathrm{o}(1/\lambda)$ as $\lambda\to\infty$ \cite{wuverdu}. This includes exponential, finite-alphabet, and Gaussian distributed signals \cite{analyticMMSE}. 
\end{enumerate}

Theorem \ref{thrm1} manifests equivalence with its rank-one analog, rigorously analyzed in \cite{replicaproof,LelargeMiolane,adaptiveinterpolation}, due to the $M$-independence of the right-hand side of equation~\eqref{mainresult}. Similar results concerning sublinear-rank models behaving as their finite-rank counterparts exist for mesoscopic spiked perturbations of random matrices \cite{huang2018mesoscopic} and for spherical integrals \cite{sublinear_sphericalint}. Our results complement these in the context of Bayesian inference. If hypothesis \ref{H4} is weakened, then Theorem \ref{thrm1} holds for all SNR $\lambda$ greater (smaller) than the largest (smallest) non-analytic value of $\lambda$. This follows from straightforward changes to the proof of Theorem \ref{thrm2} given at the end of Section \ref{s3}.

The proof of Theorem \ref{thrm1} comprises a lower bound on $F_N(\lambda)$ (valid for $M,N\in\N$), hence $\liminf_{N\to\infty}F_N(\lambda)$, obtained via a standard application of Guerra's interpolation method \cite{Guerra}, and a matching upper bound on $\limsup_{N\to\infty} F_N(\lambda)$, obtained through a multiscale application of the cavity method prescribed by the Aizenman--Sims--Starr scheme \cite{ASS}. The novelty of the argument lies in two key ingredients: Firstly, we show in Section~\ref{s3}, and summarize in \S\ref{s2.2}, how the supremum of the rank-$M$ replica symmetric potential $\phi_M(\lambda)$ \eqref{phiM} reduces to its rank-one analog $\phi_1(\lambda)$ -- this is used to obtain agreement between the aforementioned upper and lower bounds and moreover improves the tractability of the right-hand side of equation \eqref{mainresult}. Secondly, we develop in Section~\ref{s5} a cavity method that is able to treat sequences $F_{N,M}(\lambda)$ indexed by \emph{two} growing dimensions instead of one, as is usually the case. In addition to the rank-one reduction of Section~\ref{s3}, said cavity computations also use the lower bound on $F_N(\lambda)$ that we obtain via Guerra's interpolation method and thermal concentration of the overlap matrix $\bR_{10}$. We present these two prerequisite results, which are proven through relatively standard methods, in Section \ref{s4} -- the contents of sections~\ref{s4} and~\ref{s5} are summarized in \S\ref{s2.3}. We emphasize that all results but one in this paper hold for $M$ growing at certain polynomial rates in $N$.

For convenience of later analysis, it is helpful to re-express the replica symmetric potential in terms of the mutual information of the $M$-dimensional vector Gaussian channel (see Appendix \ref{appA}):
\begin{equation}
	F^{\mathrm{RS}}_M(\bQ,\lambda)=-\frac1M I(\bx_0;\sqrt{\lambda\bQ}\bx_0+\bz ) - \frac{\lambda}{4M}  \lVert\bQ-\rho I_M\rVert_{\mathrm{F}}^2 +\frac{\lambda\rho^2 }{4},\label{RSpotMI}
\end{equation}
where $\bz \sim \mathcal{N}(0,I_M)$ is standard Gaussian noise and $\bx_0\sim \pP_{X,M}$.

\begin{remark}[Properties of the replica symmetric potential]\label{rmk1}
The squared Frobenius norm $\lVert\bQ-\rho I_M\rVert^2_{\mathrm{F}}$ is a convex paraboloid with minimum at $\bQ=\rho I_M$, while the mutual information $I(\bx_0;\sqrt{\lambda\bQ}\bx_0+\bz)$ is a concave, non-decreasing monotone function of $\bQ$ \cite{guoshamaiverdu,reeves2018}: for $\bQ_1,\bQ_2\in\mathcal{S}_M$, $0\le t\le 1$, and $\bz,\bz',\bz''\sim\mathcal{N}(0,I_M)$ independent,
\begin{multline*}
tI(\bx_0;\sqrt{\lambda\bQ_1}\bx_0+\bz)+(1-t)I(\bx_0;\sqrt{\lambda\bQ_2}\bx_0+\bz')
\\\le I(\bx_0;\sqrt{\lambda(t\bQ_1+(1-t)\bQ_2)}\bx_0+\bz'').
\end{multline*}
The problem of maximizing $F_M^{\mathrm{RS}}(\bQ,\lambda)$ over $\bQ$ therefore translates to one of comparing these two competing terms. For each $m\in\N$ and $\bQ\in\cS_m$, $F_m^{\mathrm{RS}}(\bQ,\lambda)-\lambda\rho^2/4$ is a non-increasing, monotone, continuous, convex function of $\lambda$. Thus, the same holds true for $\phi_m(\lambda)-\lambda\rho^2/4=\sup_{\bQ\in\cS_m}F_m^{\mathrm{RS}}(\bQ,\lambda)-\lambda\rho^2/4$.
\end{remark}

\subsection{Equivalence of the suprema of the rank-\texorpdfstring{$M$}{M} and rank-one replica symmetric potentials} \label{s2.2}

We link the supremum of the rank-$M$ replica symmetric potential $F_M^{\mathrm{RS}}(\bQ,\lambda)$ to its rank-one equivalent through a series of intermediary results. The first two are of particular interest due to their being a statement on the worst possible additive Gaussian noise in a vector channel with i.i.d.~inputs.

\begin{lemma}[Off-diagonal entries of the noise covariance bolster information] \label{Lemma1}
Assume hypothesis \ref{H2} and let $\bx_0\sim\pP_{X,M}$, $x_0\sim\pP_X$, $\bz\sim\mathcal{N}(0,I_M)$, and $z\sim\mathcal{N}(0,1)$. Then, for $\bS\in\cS_M$, the mutual information between $\bx_0$ and the output of the vector channel with additive Gaussian noise $\bS^{1/2}\bz$ of covariance $\bS$ satisfies
\begin{equation}
I(\bx_0;\bx_0+\bS^{1/2}\bz)\ge \sum_{i=1}^M I(x_0;x_0+\sqrt{\bS_{ii}}z).\label{Lem1_1}
\end{equation}
This inequality is an equality if $\bS$ is diagonal.
\end{lemma}
We prove this lemma in \S\ref{s3.1} and use it in \S\ref{s3.3}, with $\bS=(\lambda\bQ)^{-1}$, to give a partial proof of Theorem \ref{thrm2} below. Indeed, due to convexity properties of the squared Frobenius norm in equation \eqref{RSpotMI}, Lemma \ref{Lemma1} only supplies a proof of Theorem \ref{thrm2} when the eigenvalues of $\bQ$ are restricted to be sufficiently large. However, we are able to utilize a convexity result of \cite{wibisono2018convexity} on signals drawn from distributions with $D$-bounded supports in order to formulate a parallel proof for when $\bQ$ has small eigenvalues and obtain the following corollary of Lemma \ref{Lemma1}.

\begin{corollary}[Worst Gaussian noise with covariance of fixed trace] \label{Cor1}
Let $\bS\in\cS_M$ be such that $\bS_{ii}\ge D^2$ for each $1\le i\le M$ and let $\sigma:=\Tr\bS/M$ be its normalized trace. Then, in the setting of Lemma \ref{Lemma1} with $\pP_X$ having $D$-bounded support, we have that
    \begin{equation}
		I(\bx_0;\bx_0+\bS^{1/2}\bz)\ge MI(x_0;x_0+\sqrt{\sigma}z).\label{Cor1_1}
	\end{equation}
 This inequality is an equality if $\bS=\sigma I_M$.
\end{corollary}

The partial proofs of Theorem \ref{thrm2} emerging from Lemma \ref{Lemma1} and Corollary \ref{Cor1} are interestingly dual to each other in the sense that the first proof focuses on the diagonal entries of $\bS$, while the second focuses instead on its eigenvalues -- note that these parameters coincide when $\bS$ is diagonal. This shift in focus means that the roles of the mutual information and the squared Frobenius norm in the expression \eqref{RSpotMI} are reversed in the two partial proofs. Importantly, said partial proofs correspond respectively to regimes of high and low SNR due to the following properties, proven in \S\ref{s3.2}, of maximizers $\bQ^*(\lambda)$ of $F_M^{\mathrm{RS}}(\bQ,\lambda)$ linking the regime of $\bQ$ having large (small) eigenvalues to that of high (low) SNR.
\begin{proposition}[Properties of $\bQ^*(\lambda)$] \label{prop1}
Assume hypotheses \ref{H2} and \ref{H4}. Let $\pP_X$ be centered with bounded second moment $\rho=\E_{\pP_X}X^2$ and let $\bQ^*(\lambda)\in\cS_M$ be such that
\begin{equation*}
F_M^{\mathrm{RS}}(\bQ^*(\lambda),\lambda)=\phi_M(\lambda).
\end{equation*}
Then, we have that $\bQ^*(\lambda)$ and its eigenvalues $q_1^*,\ldots,q_M^*$ have the following properties:
\begin{enumerate}
\item $0\le q_1^*,\ldots,q_M^*\le\rho$ for all $\lambda\ge0$.
\item $\lVert\bQ^*(\lambda)\rVert_{\mathrm{F}}$ is continuous for all $\lambda\ge0$ except possibly at $\lambda=\lambda_c$.
\item Given $\rho_S'\in(0,\rho)$, there exists $0<\lambda_S'<\lambda_c$ such that $0\le q_1^*,\ldots,q_M^*<\rho_S'$ for all $0\le\lambda<\lambda_S'$.
\item Given $\rho_L'\in(0,\rho)$, there exists $\lambda_L'>\lambda_c$ such that $\rho_L'<q_1^*,\ldots,q_M^*\le\rho$ for all $\lambda>\lambda_L'$.
\end{enumerate}
\end{proposition}

Thus, in \S\ref{s3.3}, we first prove the following theorem in the high and low SNR regimes before finally using analytic continuation in $\lambda$ of $\phi_m(\lambda)$ ($m\le M$), assumed to be real analytic for all $\lambda\ne\lambda_c$ by hypothesis, to extend the regime of veracity to the entirety of the positive real line.
\begin{theorem}[Rank-one reduction] \label{thrm2}
Assume hypotheses \ref{H2}--\ref{H5}. We then have for all SNR $\lambda\ge0$ that
\begin{equation*}
\sup_{\bQ\in\cS_M}F_M^{\mathrm{RS}}(\bQ,\lambda)=\sup_{q\in[0,\rho]}F_1^{\mathrm{RS}}(q,\lambda).
\end{equation*}
\end{theorem}

\subsection{Bounding the limiting free entropy via the interpolation and multiscale cavity methods} \label{s2.3}

The proof of Theorem \ref{thrm1} is a combination of the following two results (see \S\ref{s5.4}).
\begin{proposition}[Free entropy lower bound] \label{prop2}
Under hypotheses \ref{H2}--\ref{H5}, we have for any $M,N\ge1$ that
\begin{equation*}
F_N(\lambda)\ge\sup_{q\in[0,\rho]}F_1^{\mathrm{RS}}(q,\lambda).
\end{equation*}
\end{proposition}
\begin{proposition}[Free entropy upper bound] \label{prop3}
Under the assumptions of Theorem \ref{thrm1}, we have that
\begin{equation} \label{FrenUpperBound}
\limsup_{N\to\infty}F_N(\lambda)\le\sup_{q\in[0,\rho]}F_1^{\mathrm{RS}}(q,\lambda).
\end{equation}
\end{proposition}

Proposition \ref{prop2} is proven in \S\ref{s4.1} using a standard application (apart from the use of Theorem~\ref{thrm2}) of Guerra's interpolation method \cite{Guerra,universality2,LelargeMiolane}, while Proposition \ref{prop3} requires a non-trivial adaptation of the cavity method. In the standard cavity computation for finite rank-models \cite{LelargeMiolane}, one begins by Ces\`aro summing a telescoping series to obtain the upper bound
\begin{align}
\limsup_{N\to\infty}\widetilde{F}_N(\lambda)&=\limsup_{N\to\infty}\frac{1}{NM}\sum_{n=0}^{N-1}\left(\E_{\bZ,\bX_0,\eps}\ln \widetilde{Z}_{n+1,M}-\E_{\bZ,\bX_0,\eps}\ln \widetilde{Z}_{n,M}\right) \nonumber
\\&\le\limsup_{N\to\infty}\frac{1}{M}\left(\E_{\bZ,\bX_0,\eps}\ln \widetilde{Z}_{N+1,M}-\E_{\bZ,\bX_0,\eps}\ln \widetilde{Z}_{N,M}\right), \label{FtildeCesaro}
\end{align}
where $\widetilde{F}_N(\lambda)$ and $\widetilde{Z}_{N,M}$ are perturbations of the free entropy and partition function defined by equations \eqref{FrenEnt} and \eqref{ZNM}, but with $H_N(\bX)$ replaced by a \textit{perturbed Hamiltonian} 
\begin{equation} \label{Htilde}
\widetilde{H}_{N,\eps}(\bX):=H_N(\bX)+H_{N,\eps}(\bX)
\end{equation}
and the expectation now containing an average over the perturbation parameter $\eps$ and any additional randomness due to the \textit{perturbation Hamiltonian} $H_{N,\eps}(\bX)$. This latter Hamiltonian is carefully chosen to induce concentration of the overlap matrix $\bR_{10}$~\eqref{overlap} in the small $\eps$ regime -- we show in \S\ref{s4.3} that for $H_{N,\eps}(\bX)$ defined by equation~\eqref{Hpert} forthcoming, $\limsup_{N\to\infty}\widetilde{F}_N(\lambda)$ serves as an upper bound for the left-hand side of inequality \eqref{FrenUpperBound}, so bounding it by the right-hand side of said inequality completes the proof of Proposition~\ref{prop3}. In the present setting, the rank $M=M_N$ grows as a function of $N$, so inequality~\eqref{FtildeCesaro} becomes
\begin{equation*}
\limsup_{N\to\infty}\widetilde{F}_N(\lambda)\le\limsup_{N\to\infty}\frac{1}{M_N}\left(\E_{\bZ,\bX_0,\eps}\ln \widetilde{Z}_{N+1,M_{N+1}}-\E_{\bZ,\bX_0,\eps}\ln \widetilde{Z}_{N,M_N}\right),
\end{equation*}
which is no longer amenable to the standard Aizenman--Sims--Starr approach \cite{ASS,LelargeMiolane}: there are two growing, interacting scales that need to be considered jointly. This challenge can be overcome by putting the $N$ and $M$ parameters on equal footing and splitting the $n$- and $m_n$-dependencies within our telescoping sum (now letting $0\le m_n\le M_N$ be a rank parameter scaling with $0\le n\le N$ as $M_N$ does with $N$) across two separate sums. Thus, we proceed by defining
\begin{align*}
\widetilde{\Delta}_N(n)&:=\E_{\bZ,\bX_0,\eps}\ln\widetilde{Z}_{n+1,m_{n+1}}-\E_{\bZ,\bX_0,\eps}\ln\widetilde{Z}_{n,m_{n+1}},
\\ \widetilde{\Delta}_M(n)&:=\E_{\bZ,\bX_0,\eps}\ln\widetilde{Z}_{n,m_{n+1}}-\E_{\bZ,\bX_0,\eps}\ln\widetilde{Z}_{n,m_n}
\end{align*}
and observing that
\begin{equation} \label{eq19}
\widetilde{F}_N(\lambda)=\frac{1}{NM_N}\left(\sum_{n=0}^{N-1}\widetilde{\Delta}_N(n)+\sum_{n=0}^{N-1}\widetilde{\Delta}_M(n)\right);
\end{equation}
the upshot here is that each of the summands in the above is a finite difference of two terms that differ only in one index. Simplifying the right-hand side of equation \eqref{eq19}, as detailed in \S\ref{s5.1}, then leads to the following identity bounding the limit supremum of the right-hand side of equation \eqref{eq19} by a tractable convex combination.

\begin{theorem}[Multiscale Aizenman--Sims--Starr scheme] \label{thrm3}
Letting $M=M_N\in\N$ be an increasing sequence, there exist some $0\le\alpha,\beta\le1$ such that
 \begin{equation} \label{DeltaN+M}
 \limsup_{N\to\infty}\widetilde F_N(\lambda)\leq \alpha\limsup_{N\to\infty}\frac{\Delta_N}{M}+(1-\alpha)\limsup_{N\to\infty}\frac{\Delta_M}{N}
 \end{equation}
 and
  \begin{equation} \label{DeltaN+M_inf}
 \liminf_{N\to\infty}\widetilde F_N(\lambda)\geq\beta\liminf_{N\to\infty}\frac{\Delta_N}{M}+(1-\beta)\liminf_{N\to\infty}\frac{\Delta_M}{N},
 \end{equation}
 where
 \begin{align}
 \Delta_N&:=\E_{\bZ,\bX_0,\eps}\ln\widetilde{Z}_{N+1,M_{N+1}}-\E_{\bZ,\bX_0,\eps}\ln\widetilde{Z}_{N,M_{N+1}}, \label{DeltaN}
 \\ \Delta_M&:=\E_{\bZ,\bX_0,\eps}\ln\widetilde{Z}_{N,M_{N+1}}-\E_{\bZ,\bX_0,\eps}\ln\widetilde{Z}_{N,M_N}. \label{DeltaM}
 \end{align}
Furthermore, if $M=\lfloor f(N)\rfloor$ for some differentiable, strictly increasing function (e.g., a monotonic cubic interpolation of the strictly increasing subsequence of the $M_N$ such that $M_N=M_{N-1}+1$) $f:\R_+\to\R_+$, then we have
\begin{equation} \label{eqalpha}
\liminf_{N\to\infty}\frac{f(N - 1)}{f(N)+Nf'(N)} \leq \beta \leq \alpha \leq \limsup_{N\to\infty}\frac{f(N)}{f(N)+Nf'(N)}.
\end{equation}
\end{theorem} 
\begin{remark} \label{rmk2}
1.~For $\gamma\in(0,1]$, the bounds of Theorem \ref{thrm3} hold with $\alpha=\beta=(1+\gamma)^{-1}$ for sublinear polynomial growth $f(N)=N^\gamma$, while if $f(N)=(\ln N)^\gamma$, we instead have $\alpha=\beta=1$. Taking $\gamma\to0$ recovers the usual Aizenman--Sims--Starr scheme.
\\2.~Our interest is primarily in the bound \eqref{DeltaN+M}, but inequality \eqref{DeltaN+M_inf} shows that these results are sharp in the sense that if the limits on the right-hand sides exist and $\alpha=\beta$, then we have the equality
\begin{equation*}
\lim_{N\to\infty}\widetilde{F}_N(\lambda)=\alpha\lim_{N\to\infty}\frac{\Delta_N}{M}+(1-\alpha)\lim_{N\to\infty}\frac{\Delta_M}{N}.
\end{equation*}
3.~In contrast to the Bayes-optimal setting, in the usual mixed $p$-spin models (see for example \cite{PBook}) or mismatched inference problems \cite{rank1universality}, the upper bound is typically proved by interpolation and the lower bound is proved using the cavity method. In particular, inequality \eqref{DeltaN+M_inf} will be the relevant bound for the cavity computations when studying growing rank versions of these models.
\end{remark}

We have stated here the multiscale Aizenman--Sims--Starr scheme in terms of the (perturbed) free entropy for convenience, but it actually holds for \textit{any} generic sequence with two growing indices. Indeed, it is a general fact about multiscale sequences and uses no properties of the log partition function, as can be seen from the proof presented in \S\ref{s5.1}. The multiscale cavity method is a powerful tool because the limit of a two-indices sequence with nonlinear asymptotic growth rate can be computed by finding the limits of differences where one of the coordinates is \emph{fixed}, instead of having to deal with increments in both coordinates simultaneously. In other words, the $\Delta_N$ and $\Delta_M$ can be computed independently of one another to isolate the effects of adding one spin or one rank coordinate. In statistical physics parlance, this reduces the problem to those of computing a cavity in the spins for fixed rank dimension and a cavity in the rank for fixed spin dimension in the thermodynamic limit. The $\Delta_{N}$ is the standard cavity studied in classical models, while the $\Delta_M$ term only appears when examining multiscale models.

With Theorem \ref{thrm3} in hand, it remains to bound the limit suprema in the right-hand side of inequality~\eqref{DeltaN+M} and compute their convex combination. In \S\ref{s5.2}, we confirm via the standard fixed-rank cavity computations therein that $\limsup_{N\to\infty}\Delta_N/M$ is bounded above by the right-hand side of inequality~\eqref{FrenUpperBound}. Remarkably, thanks to the rank-one reduction described in Theorem \ref{thrm2}, it turns out that the novel $\Delta_M$ cavity term also obeys the same bound, hence so too does the right-hand side of inequality~\eqref{DeltaN+M}, as desired -- we expect this observation to extend to the full sublinear-rank regime $M=\mathrm{o}(N)$. This drastic simplification is however not expected to hold in the regime of extensive rank $M=\Theta(N)$. In that setting, Theorem~\ref{thrm2} can no longer be used, so the right-hand sides of propositions~\ref{prop2} and~\ref{prop3} become more complicated and cease to match. Thus, more refined tools are needed to derive extensive-rank counterparts of said propositions.

In the present sublinear-rank setting, the rank-one equivalence of Theorem \ref{thrm2} and a bootstrapping argument suggests to write the heuristic
\begin{align}
\frac{\Delta_M}{N}&=\frac{1}{N}\E_{\bZ,\bX_0,\eps}\ln\widetilde{Z}_{N,M_{N+1}}-\frac{1}{N}\E_{\bZ,\bX_0,\eps}\ln\widetilde{Z}_{N,M_N} \nonumber
\\&\approx(M_N+1)\sup_{q\in[0,\rho]}F_1^{\mathrm{RS}}(q,\lambda)-M_N\sup_{q\in[0,\rho]}F_1^{\mathrm{RS}}(q,\lambda)=\sup_{q\in[0,\rho]}F_1^{\mathrm{RS}}(q,\lambda) \label{DeltaMapprox}
\end{align}
for $N$ such that $M_{N+1}=M_N+1$. In \S\ref{s5.3}, we present some manipulations involving Proposition \ref{prop2} and the fixed-rank $\Delta_N$ cavity bound to prove that the intuition behind approximation~\eqref{DeltaMapprox} is the correct one. As mentioned above, said $\Delta_N$ cavity bound is obtained in \S\ref{s5.2} using standard fixed-rank arguments \cite{ASS}, \cite[\S4.3 \& \S6.2.5]{LelargeMiolane}. The idea, then, is to compare the rank $M$, perturbed, spiked Wigner model of size $N+1$ to that of size $N$, where we now fix our choice of perturbation to be
\begin{equation} \label{Hpert}
H_{N,\eps}(\bX):=\Tr\left(\sqrt{\eps}\widetilde{\bZ}\bX^\intercal+\eps\bX_0\bX^\intercal-\frac{\eps}{2}\bX\bX^\intercal\right),\quad\eps\in[s_N,2s_N],
\end{equation}
with $\widetilde{\bZ}\in\R^{N\times M}$ a standard Ginibre matrix consisting of i.i.d.~standard Gaussian entries independent of everything else and $(s_N)_{N\ge1}$ a sequence of numbers in $(0,1)$ that decrease slowly to zero as $N\to\infty$. For the size $N+1$ model, denoting the $(N+1)\textsuperscript{th}$ row of the signal by $\bet_0\sim\pP_{X,M}=\pP_X^{\otimes M}$ and that of the noise by $(\bxi,\xi)$ with $\bxi\sim\mathcal{N}(0,I_N)$ and $\xi\sim\mathcal{N}(0,1)$ (the first $N$ rows being $\bX_0$ and $(\bZ,\bxi^\intercal)$, respectively) shows that
\begin{align}
\widetilde{H}_{N+1,\eps_{N+1}}(\bX,\bet)\!&\;=H_N'(\bX)+H_{N,\eps_{N+1}}(\bX) \nonumber
\\&\qquad+\bet z(\bX)+\bet s(\bX)\bet^\intercal+m(\bet)+r_{\eps_{N+1}}(\bet), \label{HtildeN+1}
\\ H_N'(\bX)&:=\frac{1}{2}\Tr\Bigg( \sqrt{\frac{ \lambda }{N+1}}  \bZ \bX \bX^\intercal  + \frac{\lambda}{N+1}  \bX_0 \bX_0^\intercal \bX \bX^\intercal\nonumber
\\&\hspace{4.1cm} -\frac{\lambda}{2(N+1)}   \bX \bX^\intercal \bX \bX^\intercal  \Bigg),\nonumber
\\ z(\bX)&:=\sqrt{\frac{\lambda}{N+1}}\bX^\intercal\bxi^\intercal+\frac{\lambda}{N+1}\bX^\intercal\bX_0\bet_0^\intercal,\nonumber
\\ s(\bX)&:=-\frac{\lambda}{2(N+1)}\bX^\intercal\bX,\nonumber
\end{align}
\begin{align}
m(\bet)&:=\frac{1}{2}\left(\sqrt{\frac{\lambda}{N+1}}\xi+\frac{\lambda}{N+1}\bet_0\bet_0^\intercal-\frac{\lambda}{2(N+1)}\bet\bet^\intercal\right)\bet\bet^\intercal,\nonumber
\\ r_{\eps_{N+1}}(\bet)&:=\sqrt{\eps_{N+1}}\widetilde{\bxi}\bet^\intercal+\eps_{N+1}\bet_0\bet^\intercal-\frac{\eps_{N+1}}{2}\bet\bet^\intercal,\nonumber
\end{align}
where we have made explicit the $N$-dependence of $\eps=\eps_N$ and denoted the $(N+1)\textsuperscript{th}$ row of the perturbation noise in $H_{N+1,\eps_{N+1}}(\bX,\bet)$ by $\widetilde{\bxi}$. The point of decomposing $\widetilde{H}_{N+1,\eps_{N+1}}(\bX,\eta)$ into the form \eqref{HtildeN+1} is that one may similarly observe for
\begin{equation*}
\widehat{y}(\bX):=\frac{1}{2}\Tr\Bigg(\frac{\sqrt{\lambda}}{N}\widehat{\bZ}\bX\bX^\intercal+\frac{\lambda}{N^2}\bX_0\bX_0^\intercal\bX\bX^\intercal-\frac{\lambda}{2N^2}\bX\bX^\intercal\bX\bX^\intercal\Bigg),
\end{equation*}
with $\widehat{\bZ}$ an independent copy of $\bZ$, that $\widetilde{H}_{N,\eps_N}(\bX)$ is statistically equivalent (up to some lower order terms) to 
\begin{equation*}
H_N'(\bX)+H_{N,\eps_N}(\bX)+\widehat{y}(\bX)
\end{equation*}
in the sense that the associated partition functions are equal in distribution. Thus, the difference $\Delta_N$ \eqref{DeltaN} can be expressed in terms of averages of exponentials of the cavity fields $z(\bX),s(\bX),m(\bet),r_{\eps_{N+1}}(\bet),\widehat{y}(\bX)$, and $H_{N,\eps_{N+1}}(\bX)-H_{N,\eps_N}(\bX)$ with respect to the Gibbs (posterior) average (writing $\eps_N$ as $\eps$)
\begin{align}
\langle\,\cdot\,\rangle_{N,\eps}'&:=\frac{1}{Z_{N,\eps}'}\int_{\R^{N\times M}}\,\cdot\,e^{H_N'(\bX)+H_{N,\eps}(\bX)}\,\de\!\pP_{X,N,M}(\bX), \label{Gibbspert}
\\ Z_{N,\eps}'&:=\int_{\R^{N\times M}}e^{H_N'(\bX)+H_{N,\eps}(\bX)}\,\de\!\pP_{X,N,M}(\bX). \label{ZNeps'}
\end{align}
In \S\ref{s5.2}, we exploit the relative tractability of our cavity fields and the fact that the noise within them is independent of the noise in $H_N'(\bX)+H_{N,\eps_N}(\bX)$ to control certain rates of convergence and ultimately deduce that
\begin{equation*}
\limsup_{N\to\infty}\frac{\Delta_N}{M_N}\le\sup_{q\in[0,\rho]}F_1^{\mathrm{RS}}(q,\lambda), 
\end{equation*}
as desired. We emphasize here that although the techniques used in \S\ref{s5.2} are already in the literature, our contribution lies in our tracking the $M$-dependence of rates of convergence that arise throughout the computations. These become consequential in the thermodynamic limit when working in the sublinear-rank regime.

The final key ingredient for carrying out the cavity computation is the thermal concentration of the overlap matrix $\bR_{10}$ \eqref{overlap}. This result allows us to replace $\bR_{10}$ by the better behaved quantity $\langle\bR_{10}\rangle_{N,\eps}'$ throughout the computations of \S\ref{s5.2}. While it is sufficient for our purposes to prove thermal concentration of the overlap matrix with respect to the perturbed Hamiltonian $H_N'(\bX)+H_{N,\eps}(\bX)$, we are able to do so for a larger class of Hamiltonians and so believe our result to be a valuable tool for future studies of Bayes-optimal inference problems with factorized priors and growing rank.

\begin{theorem}[Thermal concentration of the overlap matrix] \label{thrm4}
Let $\pP_X$ be a centered distribution with bounded fourth moment, assume hypothesis \ref{H2}, let $H_N''(\bX)$ denote the log-likelihood of a fully factorized channel
\begin{align}
\bY''&\sim\pP_{\mathrm{out}}(\,\cdot\mid\bX_0\bX_0^\intercal), \nonumber
\\ \pP_{\mathrm{out}}(\bY''\mid\bX_0\bX_0^\intercal)&=\prod_{i=1}^N\pP_{\mathrm{out}}^{(i=j)}(\bY_{ii}''\mid\bx_{0,i}^\intercal\bx_{0,i})\prod_{1\le i<j\le N}\pP_{\mathrm{out}}^{(i<j)}(\bY_{ij}''\mid\bx_{0,i}^\intercal\bx_{0,j}), \label{factorizedchannel}
\end{align}
with general conditional probability distributions $\pP_{\mathrm{out}}^{(i=j)},\pP_{\mathrm{out}}^{(i<j)}$ sufficiently integrable, and define $\langle\,\cdot\,\rangle_{N,\eps}''$ to be the Gibbs average corresponding to the perturbed Hamiltonian
\begin{equation} \label{Htildegeneral}
\widetilde{H}_{N,\eps}''(\bX):=H_N''(\bX)+H_{N,\eps}(\bX),
\end{equation}
a la equation~\eqref{Gibbspert}. Then, there exists a finite positive constant $C$ independent of $M,N$ and depending only on properties of $\pP_X$ such that
\begin{equation*}
\frac{1}{s_N}\int_{s_N}^{2s_N}\E\langle\lVert\bR_{10}-\langle\bR_{10}\rangle_{N,\eps}''\rVert_{\mathrm{F}}^2\rangle_{N,\eps}''\,\de\eps\le\Gamma(N,M):=\frac{CM^2}{\sqrt{Ns_N}},
\end{equation*}
where the expectation $\E[\,\cdot\,]$ is taken over $\widetilde{\bZ},\bX_0$ and the randomness in $\pP_{\mathrm{out}}(\,\cdot\mid\bX_0\bX_0^\intercal)$.
\end{theorem}

The proof of this theorem, given in \S\ref{s4.2}, relies on the presence of the perturbation Hamiltonian $H_{N,\eps}(\bX)$ \eqref{Hpert} in the definition of $\langle\,\cdot\,\rangle_{N,\eps}''$. This perturbation corresponds to side information coming from the matrix Gaussian channel $\sqrt{\eps}\bX_0+\widetilde{\bZ}$ which serves to temper singularities of $\E\langle\lVert\bR_{10}-\langle\bR_{10}\rangle_{N,\eps}''\rVert^2_{\mathrm{F}}\rangle_{N,\eps}''$ that otherwise prevent concentration. The necessity of this side information and the novelty of the proof of Theorem~\ref{thrm4} is discussed in further detail in \S\ref{s4.2}. A tool used throughout said proof and indeed this entire paper is the \emph{Nishimori identity} or the tower rule for conditional expectation (see, e.g., \cite{LelargeMiolane}), which in our setting states that for continuous, bounded functions $g$,
\begin{equation*}
\E_{\bZ,\bX_0}\langle g(\bY,\bX_1,\ldots,\bX_k)\rangle=\E_{\bZ,\bX_0}\langle g(\bY,\bX_0,\bX_2,\ldots,\bX_k)\rangle,\quad k\in\N,
\end{equation*}
where $\bX_1,\ldots,\bX_k$ are i.i.d.~samples from the posterior $\pP_{X\mid Y,N,M}(\bX\mid\bY(\bX_0,\bZ))$ and $\langle\,\cdot\,\rangle$ denotes the Gibbs average with respect to the product conditional distribution $\prod_{i=1}^k\pP_{X\mid Y,N,M}(\bX_i\mid\bY(\bX_0,\bZ))$. Henceforth, when dealing with multiple conditionally independent samples from a Bayes-posterior distribution, we will write $\bX=\bX_1$ for the first sample and the Gibbs average will refer to an average over all samples.

\section{Relating the rank-\texorpdfstring{$M$}{M} variational formula to its rank-one analog} \label{s3}

In this section, we focus entirely on the replica symmetric potential $F_M^{\mathrm{RS}}(\bQ,\lambda)$ \eqref{RSpot}, so there is no $N$ parameter nor large $N$ limit to consider. Hence, $M$ is a fixed, finite parameter throughout this section and, consequently, we have access to the literature on finite-$M$ analysis. In particular, we are able to benefit from characterizations of maximizers $\bQ^*(\lambda)$ of $F_M^{\mathrm{RS}}(\bQ,\lambda)$ given in \cite{LelargeMiolane}. It turns out that the growing rank considered in this paper only comes into play when we link the free entropy $F_N(\lambda)$ \eqref{FrenEnt} of the sublinear-rank spiked Wigner model \eqref{spikedwignermodel} to the replica symmetric potential. The upshot here is that since the supremum of the rank-$M$ replica symmetric potential is equivalent to its rank-one analog, many complications arising from the growing rank can eventually be circumvented.

\subsection{Information-theoretic inequalities on worst Gaussian noise} \label{s3.1}

We now present the proofs of Lemma \ref{Lemma1} and Corollary \ref{Cor1} given in \S\ref{s2.2}. The first follows from the definitions and standard properties of mutual information and conditional (differential) entropy (see, e.g., \cite{coverthomas}), the chain rule for mutual information, the fact that the signal $\bx_0$ and noise $\bz$ have i.i.d.~entries, and from the rotational invariance of standard Gaussian vectors. The corollary follows from a convexity property \cite{wibisono2018convexity} of the mutual information whose use is enabled by our prior having $D$-bounded support.

\begin{proof}[\textbf{Proof of Lemma \ref{Lemma1}}]
We recall that we take $\pP_{X,M}$ to have factorized form $\pP_X^{\otimes M}$, $\bx_0\sim\pP_{X,M}$, $x_0\sim\pP_X$, $\bz\sim\mathcal{N}(0,I_M)$, $z\sim\mathcal{N}(0,1)$, and $\bS\in\cS_M$. Applying the chain rule for mutual information twice shows that
\begin{align*}
I(\bx_0;\bx_0+\bS^{1/2}\bz)&=\sum_{i=1}^MI(\bx_{0,i};\bx_0+\bS^{1/2}\bz\mid\{\bx_{0,k}\}_{k<i})
\\&=\sum_{i,j=1}^MI(\bx_{0,i};\bx_{0,j}+(\bS^{1/2}\bz)_j\mid\{\bx_{0,k}\}_{k<i},\{\bx_{0,k'}+(\bS^{1/2}\bz)_{k'}\}_{k'<j}).
\end{align*}
Thus, by the non-negativity of mutual information, we can discard all terms for which $i\ne j$ to see that
\begin{align}
I(\bx_0;\bx_0+\bS^{1/2}\bz)&\ge\sum_{i=1}^MI(\bx_{0,i};\bx_{0,i}+(\bS^{1/2}\bz)_i\mid\{\bx_{0,k}\bx_{0,k}+(\bS^{1/2}\bz)_{k}\}_{k<i}) \nonumber
\\&=\sum_{i=1}^MI(\bx_{0,i};\bx_{0,i}+(\bS^{1/2}\bz)_i\mid\{(\bS^{1/2}\bz)_k\}_{k<i}). \label{MIchainrule}
\end{align}
Here, the second line follows from the fact that knowledge of $\bx_{0,k}$ and $\bx_{0,k}+(\bS^{1/2}\bz)_k$ is equivalent to that of $\bx_{0,k}$ and $(\bS^{1/2}\bz)_k$, but also that the mutual information between $\bx_{0,i}$ and $\bx_{0,i}+(\bS^{1/2}\bz)_i$ is independent of $\{\bx_{0,k}\}_{k<i}$ due to $\bx_0$ having i.i.d.~entries.

Now, notice that for a generic random variable $X$ independent of $Z_1,Z_2$, but with $Z_1$ and $Z_2$ possibly depending on each other in an arbitrary way,
\begin{align*}
I(X;X+Z_1\mid Z_2)&=H(X\mid Z_2)-H(X\mid X+Z_1,Z_2)
\\ &\ge H(X)-H(X\mid X+Z_1)
\\ &=I(X;X+Z_1).
\end{align*}
Here, $H(X)$ denotes the (differential) entropy of $X$ and we have used the facts that $H(X\mid Z_2)=H(X)$ and $H(X\mid X+Z_1,Z_2)\le H(X\mid X+Z_1)$. Applying this inequality with $X=\bx_{0,i}$, $Z_1=(\bS^{1/2}\bz)_i$, and $Z_2=\{(\bS^{1/2}\bz)_k\}_{k<i}$ to the summand in the right-hand side of inequality \eqref{MIchainrule} then shows that
\begin{equation} \label{MIchainrulefinal}
I(\bx_0;\bx_0+\bS^{1/2}\bz)\ge\sum_{i=1}^MI(\bx_{0,i};\bx_{0,i}+(\bS^{1/2}\bz)_i).
\end{equation}

Since $\bS\in\cS_M$, we may diagonalize it as $\bS=\bO\bs\bO^\intercal$ with $\bs=\diag(\sigma_1,\ldots,\sigma_M)$ being the diagonal matrix of eigenvalues of $\bS$ and $\bO$ being the corresponding real orthogonal matrix of eigenvectors. Then, due to properties of standard Gaussian vectors, we have in law that
\begin{equation*}
(\bS^{1/2}\bz)_i=(\bO\bs^{1/2}\bO^\intercal\bz)_i\overset{d}{=}(\bO\bs^{1/2}\bz)_i\overset{d}{=}\Big(\sum_{j=1}^M\bO_{ij}^2\sigma_j\Big)^{1/2}z=\sqrt{\bS_{ii}}z.
\end{equation*}
Observing also that $\bx_{0,i}\overset{d}{=}x_0$, inequality \eqref{MIchainrulefinal} reduces to the sought form~\eqref{Lem1_1}.

As $\bz,\bx_0$ have i.i.d.~entries, it is immediate from properties of mutual information that inequality \eqref{Lem1_1} turns into an equality upon setting $\bS=\diag(\bS_{11},\ldots,\bS_{MM})$.
\end{proof}

\begin{proof}[\textbf{Proof of Corollary \ref{Cor1}}]
We remain in the setting of the above proof, but further recall that $\pP_X$ has $D$-bounded support and that $\bS\in\cS_M$ is such that $\bS_{ii}\ge D^2$ for each $1\le i \le M$. Then, setting $t=\bS_{ii}$ in \cite[Thrm.~2]{wibisono2018convexity} shows that $I(x_0;x_0+\sqrt{\bS_{ii}}z)$ is a convex function of $\bS_{ii}$ in the $\bS_{ii}\ge D^2$ regime. Thus, applying Jensen's inequality to the right-hand side of inequality \eqref{Lem1_1} yields
\begin{align*}
I(\bx_0;\bx_0+\bS^{1/2}\bz)&\ge \sum_{i=1}^MI(x_0;x_0+\sqrt{\bS_{ii}}z)
\\&\ge MI\left(x_0;x_0+\Big(\frac{1}{M}\sum_{i=1}^M\bS_{ii}\Big)^{1/2}z\right),
\end{align*}
which is seen to be equivalent to inequality \eqref{Cor1_1} upon recalling that $\sigma$ denotes the normalized trace $\Tr\bS/M$.

It is again immediate that inequality \eqref{Cor1_1} is an equality when $\bS=\sigma I_M$.
\end{proof}

\subsection{Properties of maximizers of the replica symmetric potential} \label{s3.2}

We now give a series of lemmas that amount to Proposition \ref{prop1}. We recall that this proposition lists properties of maximizers $\bQ^*(\lambda)$ of $F_M^{\mathrm{RS}}(\bQ,\lambda)$ \eqref{RSpot} that will be useful in the following subsection. To begin, it is essential to observe that we are allowed to restrict our search for maximizers to solutions of a certain criticality condition; note that although it is necessary for $\bQ^*(\lambda)$ to satisfy  the following condition, said condition is not sufficient for completely characterizing $\bQ^*(\lambda)$.
\begin{lemma}[A criticality condition] \label{Lemma2}
Diagonalize $\bQ\in\cS_M$ as $\bQ=\bO\bq\bO^\intercal$ with $\bq=\diag(q_1,\ldots,q_M)$ being the diagonal matrix of eigenvalues of $\bQ$ and $\bO$ being the corresponding real orthogonal matrix of eigenvectors. Then, any maximizer $\bQ$ of $F_M^{\mathrm{RS}}(\bQ,\lambda)$ must be such that the following system of equations, solved in particular by all critical points of $F_M^{\mathrm{RS}}(\bQ,\lambda)$, is satisfied:
\begin{equation} \label{fixedpoint}
\sqrt{q_i}\left(\bO_i^\intercal\E_{\bz,\bx_0}\langle\bx\bx_0^\intercal\rangle_{\mathrm{RS}}\bO_i-q_i\right)=0,\quad1\le i\le M,
\end{equation}
where $\bO_i$ is the $i\textsuperscript{th}$ column of $\bO$ and we define the Gibbs average associated with the replica symmetric potential \eqref{RSpot} by
\begin{equation} \label{RSgibbs}
\langle\,\cdot\,\rangle_{\mathrm{RS}}:=\frac{1}{Z_M^{\mathrm{RS}}}\int\,\cdot\,e^{\sqrt{\lambda}  \bx^\intercal \sqrt{\bQ} \bz   + \lambda  \bx_0^\intercal \bQ \bx  - \frac{\lambda}{2}  \bx^\intercal \bQ \bx}\,  d \pP_{X,M}(\bx).
\end{equation}
In particular, if $q_1,\ldots,q_M>0$, this system simplifies to the fixed-point equation
\begin{equation*}
\bQ=\E_{\bz,\bx_0}\langle\bx\bx_0^\intercal\rangle_{\mathrm{RS}}.
\end{equation*}
\end{lemma}
\begin{proof}
Let $\bg=(\gamma_1,\ldots,\gamma_M)$ and interpret $F_M^{\mathrm{RS}}(\bQ,\lambda)$ as a function of $\sqrt{\bq}$ at fixed eigenvectors. Then, defining $\bQ_{\bg}$ such that $\sqrt{\bQ_{\bg}}=\bO(\sqrt{\bq}+\diag\,\bg)\bO^\intercal$, we have that the $\bg$-directional $\sqrt{\bq}$-derivative of the replica symmetric potential is given by
\begin{align}
\nabla_{\bg}F_M^{\mathrm{RS}}(\bQ,\lambda)&=\sum_{i=1}^M\gamma_i\lim_{\bg\to\mathbf{0}}\frac{\partial}{\partial\gamma_i}F_M^{\mathrm{RS}}(\bQ_{\bg},\lambda) \nonumber
\\&=\frac{1}{M}\sum_{i=1}^M\gamma_i\E_{\bz,\bx_0}\langle2\lambda\sqrt{q_i}\bx_0^\intercal\bO_i\bO_i^\intercal\bx-\lambda\sqrt{q_i}\bx^\intercal\bO_i\bO_i^\intercal\bx\rangle_{\mathrm{RS}} \nonumber
\\&\quad+\frac{1}{M}\sum_{i=1}^M\gamma_i\left(\E_{\bz,\bx_0}\langle\sqrt{\lambda}\bx^\intercal\bO_i\bO_i^\intercal\bz\rangle_{\mathrm{RS}}-\lambda q_i^{3/2}\right). \label{gradFM_pre}
\end{align}
Gaussian integration by parts shows that
\begin{align*}
\E_{\bz,\bx_0}\langle\sqrt{\lambda}\bx^\intercal\bO_i\bO_i^\intercal\bz\rangle_{\mathrm{RS}}&=\sum_{j=1}^M\E_{\bz,\bx_0}\frac{\partial}{\partial\bz_j}\langle\sqrt{\lambda}(\bO_i\bO_i^\intercal\bx)_j\rangle_{\mathrm{RS}}
\\&=\E_{\bz,\bx_0}\left(\langle\lambda\bx^\intercal\bO_i\bO_i^\intercal\sqrt{\bQ}\bx\rangle_{\mathrm{RS}}-\langle\lambda\bx^\intercal\bO_i\bO_i^\intercal\sqrt{\bQ}\rangle_{\mathrm{RS}}\langle\bx\rangle_{\mathrm{RS}}\right)
\\&=\E_{\bz,\bx_0}\langle\lambda\bx^\intercal\bO_i\bO_i^\intercal\sqrt{\bQ}\bx-\lambda\bx^\intercal\bO_i\bO_i^\intercal\sqrt{\bQ}\bx_2\rangle_{\mathrm{RS}}
\\&=\E_{\bz,\bx_0}\langle\lambda\bx^\intercal\bO_i\bO_i^\intercal\sqrt{\bQ}\bx-\lambda\bx^\intercal\bO_i\bO_i^\intercal\sqrt{\bQ}\bx_0\rangle_{\mathrm{RS}}
\\&=\E_{\bz,\bx_0}\langle\lambda\sqrt{q_i}\bx^\intercal\bO_i\bO_i^\intercal\bx-\lambda\sqrt{q_i}\bx^\intercal\bO_i\bO_i^\intercal\bx_0\rangle_{\mathrm{RS}},
\end{align*}
where in the third line we have written $\bx_2$ for a sample, conditionally independent of $\bx$, from the posterior distribution of the Gibbs average \eqref{RSgibbs} and redefined said Gibbs average to be over both samples, while in the fourth line we have used the Nishimori identity to replace $\bx_2$ with $\bx_0$; the fifth line follows from recalling that $\bO_i$ is an eigenvector of $\sqrt{\bQ}$ with eigenvalue $\sqrt{q_i}$. Subsituting this into the right-hand side of equation \eqref{gradFM_pre} then shows that
\begin{align}
\nabla_{\bg}F_M^{\mathrm{RS}}(\bQ,\lambda)&=\frac{1}{M}\sum_{i=1}^M\gamma_i\left(\E_{\bz,\bx_0}\langle \lambda\sqrt{q_i}\bx_0^\intercal\bO_i\bO_i^\intercal\bx\rangle_{\mathrm{RS}}-\lambda q_i^{3/2}\right) \nonumber
\\&=\frac{\lambda}{M}\sum_{i=1}^M\gamma_i\sqrt{q_i}\left(\bO_i^\intercal\E_{\bz,\bx_0}\langle\bx\bx_0^\intercal\rangle_{\mathrm{RS}}\bO_i-q_i\right), \label{gradFM}
\end{align}
where we have used the fact that $\bx_0^\intercal\bO_i\bO_i^\intercal\bx=\Tr\bx_0^\intercal\bO_i\bO_i^\intercal\bx$ along with the cyclic property of the trace. Equation \eqref{fixedpoint} follows from noting that all critical points $\bQ$ of $F_M^{\mathrm{RS}}(\bQ,\lambda)$ must be such that $\nabla_{\bg}F_M^{\mathrm{RS}}(\bQ,\lambda)=0$ for all choices of $\bg$, including those for which all but one $\gamma_i$ are zero.

To discard the possibility of maximizers on the boundary of $\cS_M$ that do not satisfy the system of equations \eqref{fixedpoint}, note that being a boundary point means that at least one of the $q_i$, say $q_1$, must vanish or tend to infinity. In the latter case, one may check from the form \eqref{RSpotMI} of the replica symmetric potential that $\lim_{q_1\to\infty}F_M^{\mathrm{RS}}(\bQ,\lambda)=-\infty$. In the former case, we assume for the sake of contradiction that there is at least one $q_i\ne 0$ with $i>1$, say $q_2$, that does not satisfy equation \eqref{fixedpoint}. Then, we note from equation \eqref{gradFM} that $\nabla_{(1,1,0,\ldots,0)}F_M^{\mathrm{RS}}(\bQ,\lambda)$ and $\nabla_{(1,-1,0\ldots,0)}F_M^{\mathrm{RS}}(\bQ,\lambda)$ are non-zero and disagree in sign and, therefore, such a $\bQ$ is at best a saddle point and cannot be a maximizer. Thus, for $\bQ$ to be a maximizer of $F_M^{\mathrm{RS}}(\bQ,\lambda)$ with some eigenvalues vanishing, the other eigenvalues must satisfy equation \eqref{fixedpoint}.
\end{proof}

With this criticality condition in hand, we may now readily address the properties of the maximizers $\bQ^*(\lambda)$ listed in Proposition \ref{prop1}.
\begin{lemma}[Eigenvalues of $\bQ^*(\lambda)$ are bounded] \label{Lemma3}
Assume hypothesis \ref{H2}, let $\pP_X$ have bounded second moment $\rho=\E_{\pP_X}X^2$, and let $q_1^*,\ldots,q_M^*$ be the eigenvalues of a maximizer $\bQ^*(\lambda)$ of $F_M^{\mathrm{RS}}(\bQ,\lambda)$. Then, for all $\lambda\ge0$,
\begin{equation*}
0\le q_1^*,\ldots,q_M^*\le\rho.
\end{equation*}
\end{lemma}
\begin{proof}
By Lemma \ref{Lemma2}, each eigenvalue $q_i$ ($1\le i\le M$) of a maximizer $\bQ$ of $F_M^{\mathrm{RS}}(\bQ,\lambda)$ must either be zero or equal $\bO_i^\intercal\E_{\bz,\bx_0}\langle\bx\bx_0^\intercal\rangle_{\mathrm{RS}}\bO_i$, where $\bO_i$ is the eigenvector of $\bQ$ corresponding to $q_i$. In the latter case, we see that
\begin{align*}
q_i&=\bO_i^\intercal\E_{\bz,\bx_0}\langle\bx\bx_0^\intercal\rangle_{\mathrm{RS}}\bO_i
\\&=\E_{\bz,\bx_0}\langle\bO_i^\intercal\bx(\bO_i^\intercal\bx_0)\rangle_{\mathrm{RS}}
\\&=\E_{\bz,\bx_0}\langle\bO_i^\intercal\bx\rangle_{\mathrm{RS}}^2
\\&\le \E_{\bz,\bx_0}\langle(\bO_i^\intercal\bx)^2\rangle_{\mathrm{RS}}
\\&=\E_{\bz,\bx_0}(\bO_i^\intercal\bx_0)^2
\\&=\bO_i^\intercal\E_{\bz,\bx_0}[\bx_0\bx_0^\intercal]\bO_i
\\&=\rho.
\end{align*}
Here, the second and sixth lines follow from the fact that $\bx_0^\intercal\bO_i=\bO_i^\intercal\bx_0$, the third and fifth lines are due to the Nishimori identity, the fourth line follows from applying Jensen's inequality, and the final line follows from noting that since $\bx_0$ has i.i.d.~entries of second moment $\rho$, we must have that $\E_{\bz,\bx_0}[\bx_0\bx_0^\intercal]=\rho I_M$. Since $\bQ$ is positive semidefinite, its eigenvalues are bounded below by zero, so we are done.
\end{proof}

\begin{lemma}[Continuity of $\lVert\bQ^*(\lambda)\rVert_{\mathrm{F}}$] \label{Lemma4}
Assume hypothesis \ref{H4}. Then, any maximizer $\bQ^*(\lambda)$ of $F_M^{\mathrm{RS}}(\bQ,\lambda)$ is such that $\lVert\bQ^*(\lambda)\rVert_{\mathrm{F}}$ is continuous on $[0,\infty)\setminus\{\lambda_c\}$.
\end{lemma}
\begin{proof}
It is immediate that the proof of the rank-one result \cite[Prop.~17]{LelargeMiolane} extends to the rank-$M$ case to yield
\begin{equation*}
\lVert\bQ^*(\lambda)\rVert_{\mathrm{F}}^2=4\phi_M'(\lambda),
\end{equation*}
where we recall that $\phi_M(\lambda)=\sup_{\bQ\in\cS_M}F_M^{\mathrm{RS}}(\bQ,\lambda)$. By hypothesis, $\phi_M(\lambda)$ is real analytic on $[0,\infty)\setminus\{\lambda_c\}$, so $\phi_M'(\lambda)$ and $\lVert\bQ^*(\lambda)\rVert_{\mathrm{F}}=\sqrt{4\phi_M'(\lambda)}$ are likewise.
\end{proof}

\begin{lemma}[Eigenvalues of $\bQ^*(\lambda)$ at low SNR] \label{Lemma5}
Assume hypotheses \ref{H2} and \ref{H4}, fix $\rho_S'\in(0,\rho)$, let $\pP_X$ be centered, and let $q_1^*,\ldots,q_M^*$ be the eigenvalues of a maximizer $\bQ^*(\lambda)$ of $F_M^{\mathrm{RS}}(\bQ,\lambda)$. Then, there exists $0<\lambda_S'<\lambda_c$ such that whenever $0\le\lambda<\lambda_S'$,
\begin{equation*}
0\le q_1^*,\ldots,q_M^*<\rho_S'.
\end{equation*}
\end{lemma}
\begin{proof}
At $\lambda=0$, there is no coupling between $\bx$ and $\bx_0$ in the definition \eqref{RSgibbs} of $\langle\,\cdot\,\rangle_{\mathrm{RS}}$, so one sees that
\begin{equation*}
\E_{\bz,\bx_0}\langle\bx\bx_0^\intercal\rangle_{\mathrm{RS}}=\langle\bx\rangle_{\mathrm{RS}}\Big\vert_{\lambda=0\,}\E_{\bx_0}\bx_0^\intercal=\mathbf{0},
\end{equation*}
where we have used the fact that $\pP_X$ is centered. Thus, by the criticality condition of Lemma \ref{Lemma2}, we have that $\bQ^*(0)$, hence $\lVert\bQ^*(0)\rVert_{\mathrm{F}}$, must be zero. Then, the continuity of $\lVert\bQ^*(\lambda)\rVert_{\mathrm{F}}$ due to Lemma \ref{Lemma4} may be used to establish existence of $0<\lambda_S'<\lambda_c$ such that $0\le\lVert\bQ^*(\lambda)\rVert_{\mathrm{F}}<\rho_S'$ for all $0\le\lambda<\lambda_S'$. In terms of eigenvalues, this shows that
\begin{equation*}
0\le(q_1^*)^2+\cdots+(q_M^*)^2<(\rho_S')^2,
\end{equation*}
so the largest value an eigenvalue can attain occurs when the remaining eigenvalues are zero and this value is bounded above by $\rho_S'$.
\end{proof}

\begin{lemma}[Eigenvalues of $\bQ^*(\lambda)$ at high SNR] \label{Lemma6}
Let $\pP_{X,M}$ be as in Proposition \ref{prop1}, fix $\rho_L'\in(0,\rho)$, and let $q_1^*,\ldots,q_M^*$ be the eigenvalues of a maximizer $\bQ^*(\lambda)$ of $F_M^{\mathrm{RS}}(\bQ,\lambda)$. Then, there exists $\lambda_L'>\lambda_c$ such that whenever $\lambda>\lambda_L'$,
\begin{equation*}
\rho_L'< q_1^*,\ldots,q_M^*\le\rho.
\end{equation*}
\end{lemma}
\begin{proof}
We read from \cite[Prop.~43]{LelargeMiolane} that for all $\lambda\ne\lambda_c$, the limiting minimum mean-square error \eqref{MMSEspikedwigner} of the spiked Wigner model of fixed rank $M$ is given by
\begin{equation*}
\lim_{N\to\infty}\mathrm{MMSE}_{N,M}(\lambda)=M\rho^2-\lVert\bQ^*(\lambda)\rVert_{\mathrm{F}}^2.
\end{equation*}
It is well known that the left-hand side of this equation vanishes in the $\lambda\to\infty$ limit, so we have that $\lim_{\lambda\to\infty}\lVert\bQ^*(\lambda)\rVert_{\mathrm{F}}^2=M\rho^2$. Then, the continuity of $\lVert\bQ^*(\lambda)\rVert_{\mathrm{F}}$ due to Lemma \ref{Lemma4} says that there must exist $\lambda_L'>\lambda_c$ such that $\lVert\bQ^*(\lambda)\rVert_{\mathrm{F}}^2>(M-1)\rho^2+(\rho_L')^2$ whenever $\lambda>\lambda_L'$. In terms of eigenvalues, this translates to
\begin{equation*}
(q_1^*)^2+\cdots+(q_M^*)^2>(M-1)\rho^2+(\rho_L')^2.
\end{equation*}
From Lemma \ref{Lemma3}, we see that the largest value an eigenvalue can take is $\rho$. When $M-1$ of the eigenvalues take this value, the above inequality shows that the remaining eigenvalue must still be bounded below by $\rho_L'$.
\end{proof}

Combining lemmas \ref{Lemma3}--\ref{Lemma6} produces Proposition \ref{prop1}, which enables us to constrain the eigenvalues $q_1^*,\ldots,q_M^*$ of suitable maximizers $\bQ^*(\lambda)$ of $F_M^{\mathrm{RS}}(\bQ,\lambda)$ depending on whether we are working in the high or low SNR regime. Thus, we may now proceed to the proof of Theorem \ref{thrm2}.

\subsection{Rank-one reduction for high, low, and all SNR} \label{s3.3}
Since $\bQ\in\cS_M$, so too is $(\lambda\bQ)^{-1}$, which we denote as $\bS$ from here on out. Rewriting equation \eqref{RSpotMI} in terms of $\bS$ shows that
\begin{equation} \label{FMRSSigma}
F_M^{\mathrm{RS}}(\bQ,\lambda)=-\frac{1}{M}I(\bx_0;\bx_0+\bS^{1/2}\bz)-\frac{\lambda}{4M}\lVert\bS^{-1}/\lambda-\rho I_M\rVert_{\mathrm{F}}^2+\frac{\lambda\rho^2}{4}.
\end{equation}
Maximizing this expression over $\bQ\in\cS_M$ is equivalent to the problem of minimizing
\begin{equation} \label{FtildeMRS}
\widetilde{F}_M^{\mathrm{RS}}(\bS,\lambda):=I(\bx_0;\bx_0+\bS^{1/2}\bz)+\frac{\lambda}{4}\lVert\bS^{-1}/\lambda-\rho I_M\rVert_{\mathrm{F}}^2
\end{equation}
over $\bS\in\cS_M$. It is immediate that Lemma \ref{Lemma1} and Corollary \ref{Cor1} induce decoupling in this expression when $\bS$ is diagonal or, more specifically, has the form $\sigma I_M$. Our goal now is to show that this is indeed the form of $\bS$ for which $F_M^{\mathrm{RS}}(\bQ,\lambda)$ is maximized. We first do so in the low and high SNR regimes separately. The separation of our task into these two regimes corresponds, via Proposition \ref{prop1}, to the tasks of maximizing over $\bQ$ with eigenvalues in, respectively, $[0,\rho']$ or $[\rho',\rho]$, where we now set $\rho':=2\rho/3$. This choice of $\rho'$ is of crucial interest because the squared Frobenius norm in expression~\eqref{FtildeMRS} has first and second order partial derivatives with respect to the eigenvalues $\sigma_i$ ($1\le i\le M$) of $\bS$ given by
\begin{align}
\frac{\partial}{\partial\sigma_i}\lVert\bS^{-1}/\lambda-\rho I_M\rVert_{\mathrm{F}}^2&=\frac{2(\lambda\rho\sigma_i-1)}{\lambda^2\sigma_i^3}, \label{gradFrobenius}
\\ \frac{\partial^2}{\partial\sigma_i^2}\lVert\bS^{-1}/\lambda-\rho I_M\rVert_{\mathrm{F}}^2&=\frac{2(3-2\lambda\rho\sigma_i)}{\lambda^2\sigma_i^4}, \label{LaplacianFrobenius}
\end{align}
so one can see that the second order partial derivatives vanish exactly when the eigenvalues $q_i=(\lambda\sigma_i)^{-1}$ of $\bQ$ equal $\rho'=2\rho/3$. Thus, the gradient $\nabla\lVert\bS^{-1}/\lambda-\rho I_M\rVert_{\mathrm{F}}^2$ is a bijective function of $\sigma_1,\ldots,\sigma_M$ on the domain $\{\sigma_1,\ldots,\sigma_M>(\lambda\rho')^{-1}\}$, while on $\{(\lambda\rho)^{-1}\le\sigma_1,\ldots,\sigma_M<(\lambda\rho')^{-1}\}$, it is convex -- the significance of these properties will soon become apparent.

Let us now give the partial statement of Theorem \ref{thrm2} valid for the low SNR regime.
\begin{lemma}[Low SNR rank-one reduction] \label{Lemma7}
Assume hypotheses \ref{H2}--\ref{H4} and fix $\rho'=2\rho/3$. Then, there exists $0<\lambda_S<\lambda_c$ such that whenever $0\le\lambda<\lambda_S$, 
\begin{equation} \label{FM=F1low}
\sup_{\bQ\in\cS_M}F_M^{\mathrm{RS}}(\bQ,\lambda)=\sup_{q\in[0,\rho]}F_1^{\mathrm{RS}}(q,\lambda)=\sup_{q\in[0,\rho']}F_1^{\mathrm{RS}}(q,\lambda).
\end{equation}
\end{lemma}
\begin{proof}
Let $\bQ^*(\lambda)\in\cS_M$ be a maximizer of $F_M^{\mathrm{RS}}(\bQ,\lambda)$ with eigenvalues $q_1^*,\ldots,q_M^*$. Then, by Proposition \ref{prop1}, there exists $0<\lambda_S'<\lambda_c$ such that for all $0\le\lambda<\lambda_S'$,
\begin{equation*}
0\le q_1^*,\ldots,q_M^*<\rho'.
\end{equation*}
Thus, the eigenvalues $\sigma_i=(\lambda q_i^*)^{-1}$ ($1\le i\le M$) of the minimizer $\bS=(\lambda\bQ^*(\lambda))^{-1}$ of $\widetilde{F}_M^{\mathrm{RS}}(\bS,\lambda)$ \eqref{FtildeMRS} are such that
\begin{equation} \label{sigmailowerbound}
\sigma_1,\ldots,\sigma_M>\frac{1}{\lambda\rho'}
\end{equation}
whenever $0\le\lambda<\lambda_S'$. Letting $\bO_i$ be the eigenvector of $\bS$ corresponding to $\sigma_i$ and setting $\lambda_S:=\min\{\lambda_S',D^{-2}/\rho\}$, we then see that taking $0\le\lambda<\lambda_S\le D^{-2}/\rho$ enforces the following lower bound on the diagonal entries of $\bS$:
\begin{equation*}
\bS_{ii}=\sum_{j=1}^M\bO_{ij}^2\sigma_j>\frac{1}{\lambda\rho'}\sum_{j=1}^M\bO_{ij}^2=\frac{1}{\lambda\rho'}>\frac{D^2\rho}{\rho'}>D^2,\quad1\le i\le M.
\end{equation*}
Hence, by Corollary \ref{Cor1}, the mutual information term in equation \eqref{FtildeMRS} satisfies
\begin{equation*}
I(\bx_0;\bx_0+\bS^{1/2}\bz)\ge MI(x_0;x_0+\sqrt{\sigma}z),
\end{equation*}
where we recall that $x_0\sim\pP_X$, $z\sim\mathcal{N}(0,1)$, and $\sigma=\Tr\bS/M$. Therefore, we have that
\begin{align}
\inf_{\bS\in\cS_M}\widetilde{F}_M^{\mathrm{RS}}(\bS,\lambda)\!&\;\ge\inf_{\sigma_1,\ldots,\sigma_M\ge0}\widetilde{F}_{M,\mathrm{low}}^{\mathrm{RS}}(\bs,\lambda), \label{tildeFMtildeFMlow}
\\\widetilde{F}_{M,\mathrm{low}}^{\mathrm{RS}}(\bs,\lambda)&:=MI(x_0;x_0+\sqrt{\sigma}z)+\frac{\lambda}{4}\sum_{i=1}^M\left(\frac{1}{\lambda\sigma_i}-\rho\right)^2, \nonumber
\end{align}
where $\bs=(\sigma_1,\ldots,\sigma_M)$. Further writing $\bg=(\gamma_1,\ldots,\gamma_M)$ and using equation \eqref{gradFrobenius} with the chain rule shows that the $\bg$-directional $\bs$-derivative of $\widetilde{F}_{M,\mathrm{low}}^{\mathrm{RS}}(\bs,\lambda)$ is
\begin{equation*}
\nabla_{\bg}\widetilde{F}_{M,\mathrm{low}}^{\mathrm{RS}}(\bs,\lambda)=\sum_{i=1}^M\gamma_i\left(\iota(\sigma)+\nu(\sigma_i)\right),
\end{equation*}
where we define
\begin{align}
\iota(\sigma)&:=\frac{\de}{\de\sigma}I(x_0;x_0+\sqrt{\sigma}z), \label{iotasigma}
\\ \nu(\sigma_i)&:=\frac{\lambda\rho\sigma_i-1}{2\lambda\sigma_i^3}. \label{nusigma}
\end{align}
Hence, the critical points of $\widetilde{F}_{M,\mathrm{low}}^{\mathrm{RS}}(\bs,\lambda)$ occur at
\begin{equation} \label{systemlow}
\sigma_i=\nu^{-1}\circ\iota(\sigma),\quad 1\le i\le M,
\end{equation}
where we note that $\nu(\sigma_i)$ is a monotone, bijective function when $0\le\lambda<\lambda_S$, since then we are constrained to the regime \eqref{sigmailowerbound} on which $\nu'(\sigma_i)$ is strictly negative due to equation \eqref{LaplacianFrobenius}. The system of equations \eqref{systemlow} is thus satisfied only if $\sigma_1=\cdots=\sigma_M$, i.e., when $\bS$ has the form $\bS=\sigma I_M$.

Now, a minimizer $\bs$ of $\widetilde{F}_{M,\mathrm{low}}^{\mathrm{RS}}(\bs,\lambda)$ must either be a critical point or lie on the boundary of $[0,\infty)^M$. For $\bs$ to be a boundary point, it must be such that at least one of the $\sigma_i$, say $\sigma_1$, vanishes or tends to infinity. If $\sigma_1=0$, $\bs$ cannot be a minimizer of $\widetilde{F}_{M,\mathrm{low}}^{\mathrm{RS}}(\bs,\lambda)$ in the $0\le\lambda<\lambda_S$ regime due to the bound \eqref{sigmailowerbound}. On the other hand, if $\sigma_1\to\infty$, we have by the I-MMSE relation for scalar Gaussian channels \cite{guoshamaiverdu} that
\begin{equation} \label{scalarIMMSE}
\iota(\sigma)=-\frac{1}{2\sigma^2}\mathrm{mmse}(1/\sigma)
\end{equation}
(recall the definition \eqref{scalarMMSE} of the MMSE) and thus the bounds $0\le\mathrm{mmse}(t)\le\rho^2$ valid for all $t\ge0$ (see, e.g., \cite{LelargeMiolane}) imply that
\begin{equation*}
\lim_{\sigma_1\to\infty}\nabla_{\bg}\widetilde{F}_{M,\mathrm{low}}^{\mathrm{RS}}(\bs,\lambda)=\sum_{i=2}^M\gamma_i\frac{\lambda\rho\sigma_i-1}{2\lambda\sigma_i^3}.
\end{equation*}
Thus, for $\bs$ to be a minimizer of $\widetilde{F}_{M,\mathrm{low}}^{\mathrm{RS}}(\bs,\lambda)$ when $\sigma_1\to\infty$, we require that all of the remaining $\sigma_i$ also tend to infinity or equal $(\lambda\rho)^{-1}$ so that the above gradient vanishes and we avoid saddle points of the type described in the proof of Lemma \ref{Lemma2}. Since $\sigma_i=(\lambda\rho)^{-1}<(\lambda\rho')^{-1}$ is in contradiction with the bound \eqref{sigmailowerbound}, we must have that $\sigma_2,\ldots,\sigma_M\to\infty$ and we again see that this corresponds to the form $\bS=\sigma I_M$.

We have established that a minimimizer $\bs$ of $\widetilde{F}_{M,\mathrm{low}}^{\mathrm{RS}}(\bs,\lambda)$, regardless of whether it is a critical or boundary point, must have equal components when we are constrained to the $0\le\lambda<\lambda_S$ regime. Our infimum of interest then decouples over the $\sigma_1,\ldots,\sigma_M$ according to
\begin{equation*}
\inf_{\sigma_1,\ldots,\sigma_M\ge0}\widetilde{F}_{M,\mathrm{low}}^{\mathrm{RS}}(\bs,\lambda)=M\inf_{\sigma\ge0}\left\{I(x_0;x_0+\sqrt{\sigma}z)+\frac{\lambda}{4}\left(\frac{1}{\lambda\sigma}-\rho\right)^2\right\}.
\end{equation*}
Since $\widetilde{F}_M^{\mathrm{RS}}(\bS,\lambda)$ decouples in a similar fashion upon setting $\bS=\sigma I_M$, we see that inequality \eqref{tildeFMtildeFMlow} is an equality with the right-hand side given by the above expression. Comparing equations \eqref{FMRSSigma} and \eqref{FtildeMRS} then shows that $F_M^{\mathrm{RS}}(\bQ,\lambda)$ is maximized when $\bQ$ has the form $\bQ=qI_M$. Substituting this into the left-hand side of equation \eqref{FM=F1low} finally yields
\begin{equation*}
\sup_{q\ge0}F_M^{\mathrm{RS}}(qI_M,\lambda)=\sup_{q\geq0}F_1^{\mathrm{RS}}(q,\lambda),
\end{equation*}
which reduces to the middle and right-hand side of equation \eqref{FM=F1low} upon realising that the eigenvalues of $\bQ$ are bounded due to Lemma \ref{Lemma2} and our choice of $\lambda_S'$.
\end{proof}

We comment on a few features of the above proof. First, we used the fact that the mutual information $I(x_0;x_0+\sqrt{\bS_{ii}}z)$ is a convex function of $\bS_{ii}$ in the low SNR regime, which allowed us to use Jensen's inequality to replace $\sum_{i=1}^MI(x_0;x_0+\sqrt{\bS_{ii}}z)$ by $MI(x_0;x_0+\sqrt{\sigma}z)$ through Corollary \ref{Cor1}. Thus, the problem of minimizing $\widetilde{F}_M^{\mathrm{RS}}(\bS,\lambda)$ over $\bS$ became a problem of minimizing $\widetilde{F}_{M,\mathrm{low}}^{\mathrm{RS}}(\bS,\lambda)$ over the eigenvalues of $\bS$. We then showed that the critical points and indeed all minimizers of $\widetilde{F}_{M,\mathrm{low}}^{\mathrm{RS}}(\bS,\lambda)$ are of the form $\bS=\sigma I_M$, which relied crucially on the monotonicity of the gradient of the squared Frobenius norm $\lVert\bS^{-1}/\lambda-\rho I_M\rVert^2_{\mathrm{F}}$.

In contrast, the upcoming proof of the high SNR counterpart of Lemma \ref{Lemma7} is in some sense dual to the above, with us first needing to use the convexity of $\lVert\bS^{-1}/\lambda-\rho I_M\rVert^2_{\mathrm{F}}$ in terms of the eigenvalues of $\bS$ to replace it by a function of $\sigma=\Tr\bS/M$ via Jensen's inequality. This will then supplant the problem of minimizing $\widetilde{F}_M^{\mathrm{RS}}(\bS,\lambda)$ over $\bS$ by a minimization problem in terms of the diagonal entries of $\bS$. We will again show that said minimization problem is solved when $\bS$ is proportional to the identity -- recall that the eigenvalues and diagonal entries of $\bS$ coincide when $\bS$ is diagonal, thus explaining how the low and high SNR regimes will eventually be reconciled. Showing that our minimizers have the form $\bS=\sigma I_M$ will require us to use monotonicity of the derivative of the mutual information $I(x_0;x_0+\sqrt{\bS_{ii}}z)$ with respect to $\bS_{ii}$, which is guaranteed by hypothesis \ref{H5} for small $\bS_{ii}$ (i.e., high SNR, due to Proposition \ref{prop1}).

\begin{lemma}[High SNR rank-one reduction] \label{Lemma8}
Assume hypotheses \ref{H2}, \ref{H4}, and \ref{H5}. Let $\pP_X$ be centered with bounded second moment $\rho=\E_{\pP_X}X^2$ and fix $\rho'=2\rho/3$. Then, there exists $\lambda_L>\lambda_c$ such that whenever $\lambda>\lambda_L$, 
\begin{equation*}
\sup_{\bQ\in\cS_M}F_M^{\mathrm{RS}}(\bQ,\lambda)=\sup_{q\in[0,\rho]}F_1^{\mathrm{RS}}(q,\lambda)=\sup_{q\in[\rho',\rho]}F_1^{\mathrm{RS}}(q,\lambda).
\end{equation*}
\end{lemma}
\begin{proof}
We are once again tasked with minimizing $\widetilde{F}_M^{\mathrm{RS}}(\bS,\lambda)$ \eqref{FtildeMRS} over $\bS\in\cS_M$. In the high SNR setting, we lose the bound \eqref{sigmailowerbound} on the eigenvalues $\sigma_1,\ldots,\sigma_M$ of minimizers $\bS$ and hence Corollary \ref{Cor1} no longer applies. Instead, we begin by observing that due to Propositon \ref{prop1} and the relation $\bS=(\lambda\bQ^*(\lambda))^{-1}$, there exists $\lambda_L'>\lambda_c$ such that
\begin{equation} \label{sigmaiupperbound}
\frac{1}{\lambda\rho}<\sigma_1,\ldots,\sigma_M<\frac{1}{\lambda\rho'}
\end{equation}
whenever $\lambda>\lambda_L'$. Then, since the Frobenius norm term in equation \eqref{FtildeMRS} is convex on the domain \eqref{sigmaiupperbound}, we may use Jensen's inequality to see that
\begin{equation*}
\lVert\bS^{-1}/\lambda-\rho I_M\rVert^2_{\mathrm{F}}\ge M\left(\frac{1}{\lambda\sigma}-\rho\right)^2,
\end{equation*}
where we recall that $\sigma=\Tr\bS/M$ denotes the normalized trace of $\bS$. Inserting this into equation \eqref{FtildeMRS} and using Lemma \ref{Lemma1} then shows that
\begin{align}
\inf_{\bS\in\cS_M}\widetilde{F}_M^{\mathrm{RS}}(\bS,\lambda)\!&\;\ge\inf_{\bS_{11},\ldots,\bS_{MM}\ge0}\widetilde{F}_{M,\mathrm{high}}^{\mathrm{RS}}(\mathrm{diag}\,\bS,\lambda), \label{tildeFMtildeFMhigh}
\\\widetilde{F}_{M,\mathrm{high}}^{\mathrm{RS}}(\mathrm{diag}\,\bS,\lambda)&:=\sum_{i=1}^MI(x_0;x_0+\sqrt{\bS_{ii}}z)+\frac{\lambda M}{4}\left(\frac{1}{\lambda\sigma}-\rho\right)^2, \nonumber
\end{align}
where we further recall that $x_0\sim\pP_X,z\sim\mathcal{N}(0,1)$ and write $\mathrm{diag}\,\bS$ for the diagonal of $\bS$. Writing $\bg=(\gamma_1,\ldots,\gamma_M)$, the $\bg$-directional derivative of $\widetilde{F}_{M,\mathrm{high}}^{\mathrm{RS}}(\mathrm{diag}\,\bS,\lambda)$ is
\begin{equation*}
\nabla_{\bg}\widetilde{F}_{M,\mathrm{high}}^{\mathrm{RS}}(\mathrm{diag}\,\bS,\lambda)=\sum_{i=1}^M\gamma_i\left(\iota(\bS_{ii})+\nu(\sigma)\right),
\end{equation*}
where $\iota(\bS_{ii})$ and $\nu(\sigma)$ are specified by equations \eqref{iotasigma} and \eqref{nusigma}, respectively. Now, by hypothesis \ref{H5}, there exists $\lambda_L''>0$ such that $\lambda^2\mathrm{mmse}(\lambda)$ is monotone on $[\lambda_L'',\infty)$. Then, by equation \eqref{scalarIMMSE}, $\iota(\bS_{ii})$ is monotone for $\bS_{ii}<(\lambda_L'')^{-1}$. This condition is satisfied when $\lambda>\lambda_L',\lambda_L''/\rho'$, since then
\begin{equation*}
\bS_{ii}=\sum_{j=1}^M\bO_{ij}^2\sigma_j<\frac{1}{\lambda\rho'}\sum_{j=1}^M\bO_{ij}^2=\frac{1}{\lambda\rho'}<\frac{1}{\lambda_L''},\quad 1\le i\le M,
\end{equation*}
where $\bO_i$ is the eigenvector of $\bS$ corresponding to $\sigma_i$, the first inequality follows from the bound \eqref{sigmaiupperbound} valid for $\lambda>\lambda_L'$, and the second inequality is due to the constraint $\lambda>\lambda_L''/\rho'$. Hence, setting $\lambda_L:=\max\{\lambda_L',\lambda_L''/\rho'\}$, $\iota(\bS_{ii})$ is a monotone, bijective function for $\lambda>\lambda_L$ and the critical points of $\widetilde{F}_{M,\mathrm{high}}^{\mathrm{RS}}(\mathrm{diag}\,\bS,\lambda)$ occur at
\begin{equation*}
\bS_{ii}=\iota^{-1}\circ\nu(\sigma),\quad1\le i\le M,
\end{equation*}
which is satisfied only if $\bS_{11}=\cdots=\bS_{MM}$.

For $\widetilde{F}_{M,\mathrm{high}}^{\mathrm{RS}}(\mathrm{diag}\,\bS,\lambda)$ to be minimized on the boundary of $[0,\infty)^M$, the proposed minimizer $\mathrm{diag}\,\bS$ must be such that at least one of the $\bS_{ii}$ vanishes or tends to infinity. This is not possible for finite $\lambda>\lambda_L$ since the bound \eqref{sigmaiupperbound} and the relation $\bS=\sum_{j=1}^M\bO_{ij}^2\sigma_j$ implies the constraint $(\lambda\rho)^{-1}<\bS_{ii}<(\lambda\rho')^{-1}$. When $\lambda\to\infty$, this constraint suggests that $\bS_{11}=\cdots=\bS_{MM}\to0$. As in the proof of Lemma \ref{Lemma7}, we thus conclude that the right-hand side of inequality \eqref{tildeFMtildeFMhigh} reduces to
\begin{equation*}
\inf_{\bS_{11},\ldots,\bS_{MM}\ge0}\widetilde{F}_{M,\mathrm{high}}^{\mathrm{RS}}(\mathrm{diag}\,\bS,\lambda)=M\inf_{\sigma\ge0}\left\{I(x_0;x_0+\sqrt{\sigma}z)+\frac{\lambda}{4}\left(\frac{1}{\lambda\sigma}-\rho\right)^2\right\}.
\end{equation*}
Since substituting $\bS=\sigma I_M$ into the left-hand side of inequality \eqref{tildeFMtildeFMhigh} yields the above, the remainder of this proof follows in the same manner as in the proof of Lemma \ref{Lemma7}.
\end{proof}

In order to complete the proof of Theorem \ref{thrm2}, it remains only to use analytic continuation to combine lemmas \ref{Lemma7} and \ref{Lemma8}.
\begin{proof}[\textbf{Proof of Theorem \ref{thrm2}}]
Using the notation of equation \eqref{phiM}, we have by lemmas \ref{Lemma7} and~\ref{Lemma8} that there exist $0<\lambda_S<\lambda_c$ and $\lambda_L>\lambda_c$ such that $\phi_1(\lambda)=\sup_{q\in[0,\rho]}F_1^{\mathrm{RS}}(q,\lambda)$ and $\phi_M(\lambda)=\sup_{\bQ\in\cS_M}F_M^{\mathrm{RS}}(\bQ,\lambda)$ agree on $[0,\lambda_S)\cup(\lambda_L,\infty)$. By hypothesis \ref{H4}, $\phi_1(\lambda)$ and $\phi_M(\lambda)$ are real analytic on $[0,\infty)\setminus\{\lambda_c\}$, so by the identity theorem for real analytic functions, the aforementioned agreement between said functions extends to the domain $[0,\infty)\setminus\{\lambda_c\}$. Since $\phi_1(\lambda)$ and $\phi_M(\lambda)$ are also continuous functions of $\lambda$ (recall Remark \ref{rmk1}), they must then agree at $\lambda=\lambda_c$, as well.
\end{proof}

\section{Standard prerequisites: Interpolation and overlap concentration} \label{s4}
With Theorem \ref{thrm2} established, we would now like to move on to the next stage of the proof of Theorem~\ref{thrm1}. We recall that said proof amounts to combining the lower and upper bounds on the limiting free entropy given by propositions \ref{prop2} and \ref{prop3}. The brunt of the argument henceforth lies within the proof of Proposition \ref{prop3}. We present the key idea of the proof, which is a multiscale application of the cavity method, in Section~\ref{s5} and use this section to lay out some standard prerequisite results that will be needed therein. In \S\ref{s4.1}, we use Guerra's interpolation method \cite{Guerra} to prove Proposition \ref{prop2}, which is both a counterpart to Proposition \ref{prop3} and an ingredient of its proof. In \S\ref{s4.2}, we prove thermal concentration of the overlap matrix $\bR_{10}$ \eqref{overlap} with respect to the perturbed Hamiltonian $\widetilde{H}_{N,\eps}(\bX)$~\eqref{Htilde}, as described in Theorem \ref{thrm4}, along with some straightforward corollaries. Finally, in \S\ref{s4.3}, we prove negligibility of the difference between the limiting free entropies $\lim_{N\to\infty}F_N(\lambda)$ and $\lim_{N\to\infty}\widetilde{F}_N(\lambda)$ corresponding respectively to the original and perturbed Hamiltonians $H_N(\bX)$ \eqref{Hamiltonian} and $\widetilde{H}_{N,\eps}(\bX)$.

\subsection{Free entropy lower bound: Guerra interpolation} \label{s4.1}
Proposition \ref{prop2} follows from the standard (finite-rank) interpolation argument seen in, e.g., \cite{LelargeMiolane} in combination with the rank-one reduction of Theorem \ref{thrm2}. The main idea is to introduce a Hamiltonian that interpolates between $H_N(\bX)$ and one corresponding to $N$ copies of the Hamiltonian associated with the Gibbs average $\langle\,\cdot\,\rangle_{\mathrm{RS}}$ \eqref{RSgibbs}. Then, it is simply a matter of bounding the derivative of the associated interpolating free entropy. We present the details for the sake of completeness.

\begin{proof}[\textbf{Proof of Proposition \ref{prop2}}]
Introduce the interpolating Hamiltonian
\begin{align*}
H_t (\bX) &:=\frac12\Tr\Big( \sqrt{\frac{\lambda t}{N}}  \bZ \bX \bX^\intercal  + \frac{\lambda t}{N}  \bX_0 \bX_0^\intercal \bX \bX^\intercal  - \frac{\lambda t}{2N}   \bX \bX^\intercal \bX \bX^\intercal \Big)
\\& \;\quad+ \sum_{i = 1}^N \Big(\sqrt{\lambda (1-t)}  \bx_i^\intercal \sqrt{\bQ} \widetilde{\bz}_i   + (1-t)\lambda \bx_{0,i}^\intercal \bQ \bx_i  -  \frac{(1-t)\lambda}{2}   \bx_i^\intercal \bQ \bx_i \Big),
\end{align*}
with $\bZ,\bX_0$ as in equation \eqref{spikedwignermodel}, each $\bx_i \in \R^M$ the $i\textsuperscript{th}$ row of $\bX$, each $\bx_{0,i}\in\R^M$ the $i\textsuperscript{th}$ row of $\bX_0$, and $\widetilde{\bz}_1,\ldots,\widetilde{\bz}_N\sim\mathcal{N}(0,I_M)$ i.i.d.~standard Gaussian vectors independent of $\bZ$. Define the associated interpolating free entropy as
\begin{align*}
\varphi(t) := \frac{1}{NM} \E \ln \int_{\R^{N\times M}} e^{H_t(\bX) } \, \de\!\pP_{X,N,M}(\bX),
\end{align*}
with the expectation taken over $\bZ,\bX_0,\widetilde{\bz}_1,\ldots,\widetilde{\bz}_N$ here and throughout the remainder of this proof. It follows that
\begin{align}
\varphi'(t) &= \frac{1}{NM} \E \bigg\langle \frac{\de}{\de t} H_t(\bX) \bigg\rangle_t \nonumber
\\&= \frac{1}{2NM} \Tr\E \bigg\langle \frac12\sqrt{\frac{\lambda}{t N}}  \bZ \bX \bX^\intercal  + \frac{\lambda}{N}  \bX_0 \bX_0^\intercal \bX \bX^\intercal  - \frac{\lambda}{2N}  \bX \bX^\intercal  \bX \bX^\intercal  \bigg\rangle_t \nonumber
\\&\quad - \frac{1}{NM} \E \bigg\langle \sum_{i = 1}^N \frac12\sqrt{\frac{\lambda}{{1 - t}}}  \bx_i^\intercal \sqrt{\bQ} \widetilde{\bz}_i   + \lambda  \bx_{0,i}^\intercal \bQ \bx_i  -  \frac{\lambda}{2}  \bx_i^\intercal \bQ \bx_i   \bigg\rangle_t, \label{phiprime}
\end{align}
where we write
\begin{equation*}
\langle\,\cdot\,\rangle_t:=\frac{1}{\widehat{Z}_{N,t}}\int_{\R^{N\times M}}\cdot\,\,e^{H_t(\bX)}\,\de\!\pP_{X,N,M},\qquad \widehat{Z}_{N,t}:=\int_{\R^{N\times M}} e^{H_t(\bX)}\,\de\!\pP_{X,N,M}
\end{equation*}
for the Gibbs average with respect to the interpolating Hamiltonian $H_t(\bX)$ and the corresponding partition function. By Gaussian integration by parts with respect to the (symmetric) Wigner matrix $\bZ$, we have that
\begin{align*}
\Tr\E\left\langle\bZ\bX\bX^\intercal\right\rangle_t&=2\sum_{1\le i\le j\le N}\E\left[\frac{\partial}{\partial\bZ_{ij}}\left\langle(\bX\bX^\intercal)_{ji}\right\rangle_t\right]
\\&=2\sum_{1\le i\le j\le N}\E\left[\left\langle(\bX\bX^\intercal)_{ji}\frac{\partial H_t(\bX)}{\partial\bZ_{ij}}\right\rangle_t-\frac{1}{\widehat{Z}_{N,t}}\left\langle(\bX\bX^\intercal)_{ji}\right\rangle_t\frac{\partial \widehat{Z}_{N,t}}{\partial\bZ_{ij}}\right]
\\&=2\sum_{1\le i\le j\le N}\E\left[\left\langle(\bX\bX^\intercal)_{ji}\frac{\partial H_t(\bX)}{\partial\bZ_{ij}}\right\rangle_t-\left\langle(\bX\bX^\intercal)_{ji}\right\rangle_t\left\langle\frac{\partial H_t(\bX)}{\partial\bZ_{ij}}\right\rangle_t\right]
\\&=\sqrt{\frac{\lambda t}{N}}\Tr\E\left\langle\bX_1\bX_1^\intercal\bX_1\bX_1^\intercal-\bX_1\bX_1^\intercal\bX_2\bX_2^\intercal\right\rangle_t,
\end{align*}
where we have written $\bX_1$ for $\bX$, defined $\bX_2$ to be a sample, conditionally independent of $\bX_1$, from the interpolating posterior distribution, and have redefined $\langle\,\cdot\,\rangle_t$ to be the Gibbs average over both samples $\bX_1$ and $\bX_2$, as per the convention described at the end of \S\ref{s2.3}. In a similar fashion, Gaussian integration by parts with respect to $\widetilde{\bz}_i$ ($1\leq i\leq N$) shows that
\begin{equation*}
\E\left\langle\bx_i^\intercal\sqrt{\bQ}\widetilde{\bz}_i\right\rangle_t=\sqrt{\lambda(1-t)}\E\left\langle\bx_{1,i}^\intercal\bQ\bx_{1,i}-\bx_{1,i}^\intercal\bQ\bx_{2,i}\right\rangle_t.
\end{equation*}
Thus, using the cyclic property of the trace and defining
\begin{equation}
\bR_{12} := \frac{1}{N}\bX_1^\intercal \bX_2 =  \frac{1}{N} \sum_{i = 1}^N \bx_{1,i} \bx_{2,i}^\intercal \label{overlap12}
\end{equation}
in analogy with equation \eqref{overlap}, we may rewrite equation \eqref{phiprime} as
\begin{align*}
\varphi'(t)&=\frac{\lambda }{2M} \Tr\E \left\langle  \bR_{10}^\intercal \bR_{10}  - \frac{1}{2}   \bR_{12}^\intercal \bR_{12} \right\rangle_t - \frac{\lambda }{NM} \sum_{i = 1}^N \Tr\E \left\langle  \bQ \bx_{1,i} \bx_{0,i}^\intercal  - \frac{1}{2}   \bQ \bx_{1,i} \bx_{2,i}^\intercal  \right\rangle_t
\\&= \frac{\lambda }{2M} \Tr\E \left\langle  \bR_{10}^\intercal \bR_{10}  - \frac{1}{2}  \bR_{12}^\intercal \bR_{12}  \right\rangle_t - \frac{\lambda }{M} \Tr \E \left\langle  \bQ \bR_{10}  - \frac{1}{2}   \bQ \bR_{12}  \right\rangle_t.
\end{align*}
Using the Nishimori identity to replace $\bX_2$ by $\bX_0$, this further simplifies to
\begin{align*}
\varphi'(t)&=\frac{\lambda}{4M} \Tr\E \left\langle\bR_{10}^\intercal \bR_{10}  - 2   \bQ \bR_{10} \right\rangle_t 
\\&\geq -\frac{\lambda}{4M} \Tr\bQ^2,
\end{align*}
with the last line following from completing the square. Combining this with the definitions \eqref{FrenEnt} and \eqref{RSpot} of the free entropy and replica symmetric potential shows that 
\begin{align*}
F_N(\lambda)-F_M^{\mathrm{RS}}(\bQ,\lambda)&=\varphi(1)-\varphi(0)+\frac{\lambda}{4M}\Tr\bQ^2
\\&=\int_0^1\varphi'(t)\,\de t+\frac{\lambda}{4M}\Tr\bQ^2
\\&\geq0.
\end{align*}
As this holds for all $\bQ\in\mathcal{S}_M$, we then have that
\begin{equation*}
F_N(\lambda)\ge\sup_{\bQ\in\cS_M}F_M^{\mathrm{RS}}(\bQ,\lambda).
\end{equation*}
Using Theorem \ref{thrm2} to replace the right-hand side of this identity completes the proof.	
\end{proof}

\subsection{Thermal concentration of the overlap matrix} \label{s4.2}

The cavity computations of \S\ref{s5.2} rely on a form of thermal concentration of the overlap matrix $\bR_{10}$ \eqref{overlap}. Taken at face value, one might expect this to mean that 
\begin{equation*}
\E_{\bZ,\widetilde{\bZ},\bX_0}\langle\lVert\bR_{10}-\langle\bR_{10}\rangle_{N,0}'\rVert_{\mathrm{F}}^2\rangle_{N,0}'\overset{N\to\infty}{\longrightarrow}0,
\end{equation*}
with $\langle\,\cdot\,\rangle_{N,0}'$ specified by equation \eqref{Gibbspert}, so that one may replace $\bR_{10}$ by the more amenable quantity $\langle\bR_{10}\rangle_{N,0}'$ within the calculations of \S\ref{s5.2}. However, it is known that the concentration result proposed above does not hold in generality, with ``non-physical'' singularities arising at particular model-dependent values of SNR $\lambda$ (see \cite{barbier2022strong} and references therein). Nonetheless, due to the symmetries of our model, we are able to add infinitesimal side information (without affecting the limiting free entropy, as discussed in \S\ref{s4.3}) to induce thermal concentration of the overlap matrix up to an average over a perturbation parameter. This is achieved by replacing $\langle\,\cdot\,\rangle_{N,0}'$ in the above with the Gibbs average $\langle\,\cdot\,\rangle_{N,\eps}'$ \eqref{Gibbspert} and then integrating the left-hand side over $\eps$ as described in Theorem \ref{thrm4}, with the integration step serving to smooth out the aforementioned singularities. Similar perturbations have been commonly used in other spin glass models, where regularizing perturbations are introduced in order to enforce stability properties of the high dimensional probability measures \cite{aizenmancontuccistability,GG} that in turn induce asymptotically nice structural phenomena such as ultrametricity \cite{PUltra} and synchronization \cite{PPotts,PVS}. 

It turns out that the driving mechanisms behind our proof of thermal concentration of the overlap matrix are the presence of the perturbation Hamiltonian $H_{N,\eps}(\bX)$~\eqref{Hpert} and the Nishimori identity, which applies due to our working in the Bayes-optimal setting. Indeed, our proof holds in more general settings, so long as the model under consideration exhibits the same symmetry properties as the spiked Wigner model~\eqref{spikedwignermodel}. Thus, we now consider the fully factorized channel \eqref{factorizedchannel}
\begin{equation*}
\pP_{\mathrm{out}}(\bY''\mid\bX_0\bX_0^\intercal)=\prod_{i=1}^N\pP_{\mathrm{out}}^{(i=j)}(\bY_{ii}''\mid\bx_{0,i}^\intercal\bx_{0,i})\prod_{1\le i<j\le N}\pP_{\mathrm{out}}^{(i<j)}(\bY_{ij}''\mid\bx_{0,i}^\intercal\bx_{0,j})
\end{equation*}
and study the inference problem
\begin{equation*}
\begin{cases}\bY''\sim\pP_{\mathrm{out}}(\,\cdot\mid\bX_0\bX_0^\intercal),\\ \bY^{(\eps)}=\sqrt{\eps}\bX_0+\widetilde{\bZ},\end{cases}
\end{equation*}
with $\widetilde{\bZ}\in\R^{N\times M}$ a standard Ginibre matrix of i.i.d.~standard Gaussian entries independent of everything else, $\eps\in[s_N,2s_N]$ a perturbation parameter, and $(s_N)_{N\ge1}$ a sequence of numbers in $(0,1)$ that decrease slowly to zero as $N\to\infty$. Denoting the log-likelihood of the first channel by $H_N''(\bX)$ and recalling that the analogous term for the second channel is $H_{N,\eps}(\bX)$ \eqref{Hpert}, we see that the Hamiltonian corresponding to the full inference problem is $\widetilde{H}_{N,\eps}''(\bX)=H_N''(\bX)+H_{N,\eps}(\bX)$~\eqref{Htildegeneral}. As stated above, we observe that this Hamiltonian corresponds to the log-likelihood of an inference problem in the Bayes-optimal setting and the Nishimori identity holds, which induces concentration-of-measure \cite{overlapconcentration,barbier2022strong}.

We now prove thermal concentration of $\bR_{10}$ \eqref{overlap} with respect to the general perturbed Hamiltonian $\widetilde{H}_{N,\eps}''(\bX)$.
 
\begin{proof}[\textbf{Proof of Theorem \ref{thrm4}}]
To simplify notation throughout this proof, we write $x$ for $\bX$, $X^*$ for $\bX_0$, $\bR$ for $\bR_{10}$, $C$ a generic positive constant independent of $M,N$ and depending only on the properties of $\pP_X$, $\E[\,\cdot\,]$ for expectation over $\widetilde{\bZ},X^*$ and the randomness in $\pP_{\mathrm{out}}(\,\cdot\mid X^*X^{*\intercal})$, $h_N$ for $\widetilde{H}_{N,\eps}''(x)$, and $\langle\,\cdot\,\rangle$ for the Gibbs average
\begin{align*}
\langle\,\cdot\,\rangle_{N,\eps}''&:=\frac{1}{Z_{N,\eps}''}\int_{\R^{N\times M}}\,\cdot\,e^{h_N}\,\de\!\pP_{X,N,M}(x),
\\ Z_{N,\eps}''&:=\int_{\R^{N\times M}}e^{h_N}\,\de\!\pP_{X,N,M}(x).
\end{align*}

Following \cite{overlapconcentration,adaptiveinterpolation}, introduce the auxiliary function
\begin{equation*}
\cL:=-\frac{1}{N}\frac{\partial h_N}{\partial\eps}=-\frac{1}{N}\Tr\left(\frac{1}{2\sqrt{\eps}}\widetilde{\bZ}x^\intercal+X^*x^\intercal-\frac{1}{2}xx^\intercal\right).
\end{equation*}
Then, one sees from Gaussian integration by parts with respect to $\widetilde{\bZ}$ and the Nishimori identity that
\begin{equation} \label{RLrelation}
\frac{1}{N}\E\langle\cL\rangle=-\frac{1}{2N}\Tr\E\langle\bR\rangle .
\end{equation}
Differentiating both sides of this equation with respect to $\eps$ and simplifying the left-hand side shows that
\begin{equation} \label{LRrelation}
\frac{1}{N}\E\left\langle\frac{\partial\cL}{\partial\eps}\right\rangle-\E\langle(\cL-\langle\cL\rangle)^2\rangle=-\frac{1}{2N}\frac{\partial}{\partial\eps}\Tr\E\langle\bR\rangle .
\end{equation}
Another application of Gaussian integration by parts with respect to $\widetilde{\bZ}$ and the Nishimori identity shows that the first term in the above simplifies as
\begin{equation} \label{dLde}
\frac{1}{N}\E\left\langle\frac{\partial\cL}{\partial\eps}\right\rangle=\frac{1}{N}\E\left\langle\frac{1}{4N\eps^{3/2}}\Tr\widetilde{Z}x^\intercal\right\rangle=\frac{1}{4\eps N}\Tr\E\left(\frac{1}{N}X^*X^{*\intercal}-\langle\bR\rangle\right).
\end{equation}
By the positive semidefiniteness of $\langle x\rangle^\intercal\langle x\rangle$, the Nishimori identity, Jensen's inequality, hypothesis \ref{H2}, and the boundedness of the second moment of $\pP_X$, we have that
\begin{equation} \label{Rbound}
0\le\Tr\E\langle\bR\rangle=\frac{1}{N}\Tr\E\langle x\rangle^\intercal\langle x\rangle\le \frac{1}{N}\Tr\E\langle x^\intercal x\rangle=\frac{1}{N}\Tr\E X^{*\intercal}X^*\le CM.
\end{equation}
Hence, the quantity \eqref{dLde} is such that
\begin{equation} \label{dLdebounds}
0\le \frac{1}{N}\E\left\langle\frac{\partial\cL}{\partial\eps}\right\rangle\le \frac{CM}{4\eps N}
\end{equation}
and we read from equation \eqref{LRrelation} that
\begin{equation} \label{RLineq}
\frac{1}{2N}\frac{\partial}{\partial\eps}\Tr\E\langle\bR\rangle\le\E\langle(\cL-\langle\cL\rangle)^2\rangle .
\end{equation}
It remains now to show how the left-hand side of this inequality relates to the difference $\E\langle(\bR_{\ell\ell'}-\langle\bR_{\ell\ell'}\rangle)^2\rangle$, which sums over $1\le \ell,\ell'\le M$ to give $\E\langle\lVert\bR-\langle\bR\rangle\rVert^2_{\mathrm{F}}\rangle$, and to prove concentration of the right-hand side.

Carrying out the differentiation on the left-hand side of inequality \eqref{RLineq} shows that
\begin{multline} \label{trR10expansion}
\frac{1}{2N}\frac{\partial}{\partial\eps}\Tr\E\langle\bR\rangle=\frac{1}{2}\Tr\E[\langle\bR\rangle\langle\cL\rangle-\langle\bR\cL\rangle]
\\=\frac{1}{2N^2}\sum_{i,j=1}^N\sum_{\ell,\ell'=1}^M\E\left[\frac{1}{2}X^*_{i\ell}\langle x_{i\ell}\rangle\langle x_{j\ell'}^2\rangle-X^*_{i\ell}X^*_{j\ell'}\langle x_{i\ell}\rangle\langle x_{j\ell'}\rangle-\frac{1}{2\sqrt{\eps}}X^*_{i\ell}\langle x_{il}\rangle\langle x_{j\ell'}\rangle\widetilde\bZ_{j\ell'}\right.
\\\left.-\frac{1}{2}X^*_{i\ell}\langle x_{i\ell}x^2_{j\ell'}\rangle+X^*_{i\ell}X^*_{j\ell'}\langle x_{i\ell}x_{j\ell'}\rangle+\frac{1}{2\sqrt{\eps}}X^*_{i\ell}\langle x_{il}x_{j\ell'}\rangle\widetilde\bZ_{j\ell'}\right].
\end{multline}
Gaussian integration by parts with respect to $\widetilde\bZ_{j\ell'}$ yields
\begin{multline*}
\frac{1}{2\sqrt{\eps}}\E\left[X^*_{i\ell}\langle x_{il}x_{j\ell'}\rangle\widetilde\bZ_{j\ell'}-X^*_{i\ell}\langle x_{il}\rangle\langle x_{j\ell'}\rangle\widetilde\bZ_{j\ell'}\right]
\\=\E\left[X^*_{i\ell}\langle x_{i\ell}\rangle\langle x_{j\ell'}\rangle^2-X^*_{i\ell}\langle x_{i\ell}x_{j\ell'}\rangle\langle x_{j\ell'}\rangle+\frac{1}{2}X^*_{i\ell}\langle x_{i\ell}x_{j\ell'}^2\rangle-\frac{1}{2}X^*_{i\ell}\langle x_{i\ell}\rangle\langle x_{j\ell'}^2\rangle\right],
\end{multline*}
so that, after an application of the Nishimori identity, equation \eqref{trR10expansion} reduces to
\begin{equation*}
\frac{1}{2N}\frac{\partial}{\partial\eps}\Tr\E\langle\bR\rangle=\frac{1}{2N^2}\sum_{i,j=1}^N\sum_{\ell,\ell'=1}^M\E\left[(\langle x_{i\ell}x_{j\ell'}\rangle-\langle x_{i\ell}\rangle\langle x_{j\ell'}\rangle)^2\right].
\end{equation*}
Substituting this into inequality \eqref{RLineq} then gives
\begin{equation} \label{ineq75}
\frac{1}{2N^2}\sum_{i,j=1}^N\sum_{\ell,\ell'=1}^M\E\left[(\langle x_{i\ell}x_{j\ell'}\rangle-\langle x_{i\ell}\rangle\langle x_{j\ell'}\rangle)^2\right] \leq \E\langle(\cL-\langle\cL\rangle)^2\rangle .
\end{equation}
Our hypothesis on $\pP_{\mathrm{out}}( \,\cdot \mid\bX_0 \bX_0^\intercal )$ being fully factorized implies that the term below is unchanged upon exchanging indices $k\leftrightarrow k'$, so we see for any $1\le k\le M$ that
\begin{align*}
&\frac{1}{2N^2}\sum_{i,j=1}^N\E\left[(\langle x_{ik}x_{jk}\rangle-\langle x_{ik}\rangle\langle x_{jk}\rangle)^2\right]
\\&=\frac{1}{2MN^2}\sum_{i,j=1}^N\sum_{\ell=1}^M\E\left[(\langle x_{i\ell}x_{j\ell}\rangle-\langle x_{i\ell}\rangle\langle x_{j\ell}\rangle)^2\right]
\\&\le\frac{1}{2MN^2}\sum_{i,j=1}^N\sum_{\ell,\ell'=1}^M\E\left[(\langle x_{i\ell}x_{j\ell'}\rangle-\langle x_{i\ell}\rangle\langle x_{j\ell'}\rangle)^2\right]
\\&\le\frac{1}{M}\E\langle(\cL-\langle\cL\rangle)^2\rangle;
\end{align*}
the first inequality is due to the trivial fact that we are adding expectations of squares, while the second inequality is an application of inequality \eqref{ineq75}. Combining this with the Cauchy--Schwarz inequality shows that
\begin{align}
&\frac{1}{s_N}\int_{s_N}^{2s_N}\E\langle(\bR_{\ell\ell'}-\langle\bR_{\ell\ell'}\rangle)^2\rangle\,\de\eps \nonumber
\\&\leq \sqrt{\frac{1}{N^2}\sum_{i,j=1}^N\E\left[(X^*_{i\ell'}X^*_{j\ell'})^2\right]}\sqrt{\frac{1}{N^2s_N}\sum_{i,j=1}^N\int_{s_N}^{2s_N}\E\left[(\langle x_{i\ell}x_{j\ell}\rangle-\langle x_{i\ell}\rangle\langle x_{j\ell}\rangle)^2\right]\,\de\eps} \nonumber
\\ &\leq C\sqrt{\frac{2}{Ms_N}\int_{s_N}^{2s_N}\E\langle(\cL-\langle \cL\rangle)^2\rangle\,\de\eps}. \label{RLconcentration}
\end{align}

Moving on to the thermal concentration of $\cL$, first note that
\begin{align*}
\E\langle(\cL-\langle\cL\rangle)^2\rangle&=\frac{M}{N}\frac{\partial^2f_N}{\partial\eps^2}+\frac{1}{N}\E\left\langle\frac{\partial\cL}{\partial\eps}\right\rangle
\\&\leq \frac{M}{N}\frac{\partial^2f_N}{\partial\eps^2}+\frac{CM}{4\eps N},
\end{align*}
where we have written
\begin{equation*}
f_N:=\frac{1}{NM}\E\ln Z_{N,\eps}''
\end{equation*}
and used the bound \eqref{dLdebounds}. Integrating this quantity over $\eps$ then shows that
\begin{align*}
\frac{1}{s_N}\int_{s_N}^{2s_N}\E\langle(\cL-\langle\cL\rangle)^2\rangle\,\de\eps&\le \frac{M}{Ns_N}\left(\frac{\partial f_N}{\partial\eps}\Big\vert_{\eps=2s_N}-\frac{\partial f_N}{\partial\eps}\Big\vert_{\eps=s_N}\right)+\frac{CM\ln2}{4Ns_N}
\\&=\frac{1}{Ns_N}\left(\E\langle\cL\rangle\vert_{\eps=s_N}-\E\langle\cL\rangle\vert_{\eps=2s_N}\right)+\frac{CM\ln2}{4Ns_N}
\\&=\frac{1}{2Ns_N}\left(\Tr\E\langle\bR\rangle\vert_{\eps=2s_N}-\Tr\E\langle\bR\rangle\vert_{\eps=s_N}\right)+\frac{CM\ln2}{4Ns_N}
\\&\leq \frac{CM}{Ns_N}.
\end{align*}
Here, the first line follows from the fundamental theorem of calculus, the second from comparing the definitions of $\cL$ and $f_N$, the third from equation \eqref{RLrelation}, and the fourth from the fact that the bounds \eqref{Rbound} hold uniformly over $\eps$.

Substituting the above bound into inequality \eqref{RLconcentration} shows that
\begin{equation*}
\frac{1}{s_N}\int_{s_N}^{2s_N}\E\langle(\bR_{\ell\ell'}-\langle\bR_{\ell\ell'}\rangle)^2\rangle\,\de\eps\le\frac{C}{\sqrt{Ns_N}}.
\end{equation*}
Summing this expression over $1\le\ell,\ell'\le M$ completes the proof.
\end{proof}

Note the particularly simple form of side information used here compared to, e.g., \cite{LelargeMiolane,overlapconcentration,adaptiveinterpolation,reeves}, wherein the side channel is Bernoulli distributed or is a matrix Gaussian channel like ours, but with perturbation parameter $\eps$ of matrix form rather than scalar. Our choice of perturbation is effective due to the fact that the cavity method only requires concentration with respect to the Gibbs average $\langle\,\cdot\,\rangle_{N,\eps}''$ instead of the stronger statements proved in the aforementioned references that include concentration with respect to the quenched Gibbs average $\E\langle\,\cdot\,\rangle_{N,\eps}''$ -- the latter corresponds to analogs of Theorem \ref{thrm4} involving the quantity $\E\langle\lVert\bR_{10}-\E\langle\bR_{10}\rangle_{N,\eps}''\rVert^2_{\mathrm{F}}\rangle_{N,\eps}''$. Another key point of our simple perturbation is that it maintains the exchangeability of the variables $(\bx_{ij})$ under the quenched Gibbs average: no statistical symmetries present in the original model are broken by the perturbation apart from the trivial sign ambiguity $\bX\to-\bX$. These two points allow us to obtain concentration for higher regimes of rank than would be possible using more complicated perturbations. As a concrete consequence, Theorem \ref{thrm4} gives thermal concentration of the overlap matrix for $M=\mathrm{o}(N^{1/4})$. Yet, there is room for improvement in order to reach $M=\mathrm{o}(N)$, the regime for which we believe Theorem \ref{thrm1} to hold.

For later convenience, we present a corollary of Theorem \ref{thrm4} regarding concentration of the overlap matrix $\bR_{12}$ \eqref{overlap12}.
\begin{corollary}[Thermal concentration of $\bR_{12}$] \label{Cor2}
Working in the setting of Theorem \ref{thrm4} and recalling that $\bR_{12} = \frac{1}{N}\bX_1^\intercal \bX_2$, there exists a finite positive constant $C$ independent of $M,N$ and depending only on properties of $\pP_X$ such that
\begin{align}
&\frac{1}{s_N}\int_{s_N}^{2s_N}\E\langle\lVert\bR_{12}-\langle\bR_{12}\rangle_{N,\eps}''\rVert_{\mathrm{F}}^2\rangle_{N,\eps}''\,\de\eps\le\Gamma(N,M), \label{R12concentration}
\\&\frac{1}{s_N}\int_{s_N}^{2s_N}\E\langle\lVert\bR_{10}-\langle\bR_{12}\rangle_{N,\eps}''\rVert_{\mathrm{F}}^2\rangle_{N,\eps}''\,\de\eps\le\Gamma(N,M). \label{R10R12concentration}
\end{align}
\end{corollary}
\begin{proof}
We reuse the notation introduced in the proof of Theorem \ref{thrm4}, with the Gibbs average $\langle\,\cdot\,\rangle$ now being taken over $\bX_1,\bX_2$, and further write $\widetilde{\bR}$ for $\bR_{12}$. In analogy with the bound \eqref{RLconcentration}, we first compute through the Cauchy--Schwarz inequality, Jensen's inequality, and the Nishimori identity that
\begin{align*}
&\frac{1}{s_N}\int_{s_N}^{2s_N}\E[\langle\bR_{\ell\ell'}\rangle^2-\langle\widetilde{\bR}_{\ell\ell'}\rangle^2]\,\de\eps
\\&\leq \sqrt{\frac{1}{N^2s_N}\sum_{i,j=1}^N\int_{s_N}^{2s_N}\E[\langle x_{i\ell'}\rangle^2\langle x_{j\ell'}\rangle^2]\,\de\eps}
\\&\qquad\times\sqrt{\frac{1}{N^2s_N}\sum_{i,j=1}^N\int_{s_N}^{2s_N}\E\left[(\langle x_{i\ell}x_{j\ell}\rangle-\langle x_{i\ell}\rangle\langle x_{j\ell}\rangle)^2\right]\,\de\eps}
\\&\leq\sqrt{\frac{1}{N^2}\sum_{i,j=1}^N\sqrt{\E (X^{*}_{i\ell'})^4\,\E (X^{*}_{j\ell'})^4}}
\\&\qquad\times\sqrt{\frac{1}{N^2s_N}\sum_{i,j=1}^N\int_{s_N}^{2s_N}\E\left[(\langle x_{i\ell}x_{j\ell}\rangle-\langle x_{i\ell}\rangle\langle x_{j\ell}\rangle)^2\right]\,\de\eps}
\\ &\leq C\sqrt{\frac{2}{Ms_N}\int_{s_N}^{2s_N}\E\langle(\cL-\langle \cL\rangle)^2\rangle\,\de\eps}.
\end{align*}
Thus, we may conclude as in the proof of Theorem \ref{thrm4} that
\begin{equation} \label{Fnormdiffconc}
\frac{1}{s_N}\int_{s_N}^{2s_N}\E[\lVert\langle\bR\rangle\rVert^2_{\mathrm{F}}-\lVert\langle\widetilde{\bR}\rangle\rVert_{\mathrm{F}}^2]\,\de\eps\le\Gamma(N,M).
\end{equation}
Next, we see that
\begin{align*}
\E\langle\lVert\widetilde{\bR}-\langle\widetilde{\bR}\rangle\rVert_{\mathrm{F}}^2\rangle&=\E\langle\widetilde{\bR}^\intercal\widetilde{\bR}\rangle-\E\langle\widetilde{\bR}^\intercal\rangle\langle\widetilde{\bR}\rangle
\\&=\E\langle\bR^\intercal\bR\rangle-\E\langle\widetilde{\bR}^\intercal\rangle\langle\widetilde{\bR}\rangle
\\&=\E\langle\lVert\bR-\langle\bR\rangle\rVert_{\mathrm{F}}^2\rangle+\E[\lVert\langle\bR\rangle\rVert_{\mathrm{F}}^2-\lVert\langle\widetilde{\bR}\rangle\rVert_{\mathrm{F}}^2],
\end{align*}
where the second line follows from using the Nishimori identity on the first term and the final line follows from adding and subtracting $\E\lVert\langle\bR\rangle\rVert^2_{\mathrm{F}}$. Integrating this over $\eps$ and using Theorem \ref{thrm4} together with inequality \eqref{Fnormdiffconc} yields inequality \eqref{R12concentration}. To obtain inequality \eqref{R10R12concentration} from \eqref{R12concentration}, we need only observe that due to the Nishimori identity,
\begin{equation*}
\E\langle\lVert\bR-\langle\widetilde{\bR}\rangle\rVert^2_{\mathrm{F}}\rangle=\E\langle\lVert\widetilde{\bR}-\langle\widetilde{\bR}\rangle\rVert_{\mathrm{F}}^2\rangle.
\end{equation*}
\end{proof}

\subsection{Negligibility of the side information} \label{s4.3}

In the previous subsection we showed that upon introduction of the perturbation Hamiltonian $H_{N,\eps}(\bX)$ \eqref{Hpert}, the overlap matrix $\bR_{10}$ \eqref{overlap} can be shown to concentrate with respect to the general Gibbs average $\langle\,\cdot\,\rangle_{N,\eps}''$ for $M=\mathrm{o}(N^{1/4})$. Our interest lies in the particular Gibbs average $\langle\,\cdot\,\rangle_{N,\eps}'$ specified by equation \eqref{Gibbspert}, which relates via the cavity method described in Section \ref{s5} to the perturbed Hamiltonian $\widetilde{H}_{N,\eps}(\bX)$ \eqref{Htilde} and perturbed free entropy
\begin{equation*}
\widetilde{F}_N(\lambda):=\frac{1}{NM}\E_{\bZ,\widetilde{\bZ},\bX_0,\eps}\ln\int_{\R^{N\times M}}e^{\widetilde{H}_{N,\eps}(\bX)}\,\de\!\pP_{X,N,M}(\bX),
\end{equation*}
where the average over $\eps$ is with respect to the uniform measure on $[s_N,2s_N]$. We have thus far claimed that upper bounding $\limsup_{N\to\infty}\widetilde{F}_N(\lambda)$ by $\sup_{q\in[0,\rho]}F_1^{\mathrm{RS}}(q,\lambda)$ is sufficient for proving Proposition \ref{prop3} since said quantity is itself an upper bound on the limit supremum of the free entropy $F_N(\lambda)$ of the spiked Wigner model \eqref{spikedwignermodel}. We now present a more precise finite $N$ result that implies this claim.

\begin{lemma}[Negligibility of perturbation] \label{Lemma9}
Let $(s_N)_{N\ge1}$ be a sequence of numbers in $(0,1)$ that decrease slowly to zero as $N\to\infty$ and let $\eps=\eps_N\in[s_N,2s_N]$. Recall the definition \eqref{ZNM} of the partition function $Z_{N,M}$ and that the perturbed partition function is given by
\begin{equation*}
\widetilde{Z}_{N,M}=\int_{\R^{N\times M}}e^{\widetilde{H}_{N,\eps}(\bX)}\,\de\!\pP_{X,N,M}(\bX)
\end{equation*}
with perturbed Hamiltonian $\widetilde{H}_{N,\eps}$ specified by equations \eqref{Htilde} and \eqref{Hpert}. Then, for each $N\ge1$, we have that
\begin{equation} \label{ZvsZpert}
\frac{1}{NM}\E_{\bZ,\bX_0}\ln Z_{N,M}\le \frac{1}{NM}\E_{\bZ,\widetilde{\bZ},\bX_0}\ln\widetilde{Z}_{N,M}.
\end{equation}
\end{lemma}
\begin{proof}
We employ an interpolation argument akin to that of \S\ref{s4.1}. Hence, define an interpolating Hamiltonian
\begin{align*}
H_t(\bX)&:=\frac12\Tr\Big( \sqrt{\frac{\lambda }{N}}  \bZ \bX \bX^\intercal  + \frac{\lambda }{N}  \bX_0 \bX_0^\intercal \bX \bX^\intercal  - \frac{\lambda }{2N}   \bX \bX^\intercal \bX \bX^\intercal \Big)
\\& \;\quad+\Tr\left(\sqrt{\eps t}\widetilde{\bZ}\bX^\intercal+\eps t\bX_0\bX^\intercal-\frac{\eps t}{2}\bX\bX^\intercal\right),
\end{align*}
with $\bZ,\bX_0$ as in equation \eqref{spikedwignermodel} and $\widetilde{\bZ}$ as in equation \eqref{Hpert}. The associated interpolating free entropy is
\begin{equation*}
\varphi(t):=\frac{1}{NM}\E_{\bZ,\widetilde{\bZ},\bX_0}\ln\int_{\R^{N\times M}}e^{H_t(\bX)}\,\de\!\pP_{X,N,M}(\bX).
\end{equation*}
Writing
\begin{equation*}
\langle\,\cdot\,\rangle_t:=\frac{1}{\widehat{Z}_{N,t}}\int_{\R^{N\times M}}\cdot\,\,e^{H_t(\bX)}\,\de\!\pP_{X,N,M},\qquad \widehat{Z}_{N,t}:=\int_{\R^{N\times M}} e^{H_t(\bX)}\,\de\!\pP_{X,N,M},
\end{equation*}
we see that
\begin{align*}
\varphi'(t)&=\frac{1}{NM}\E_{\bZ,\widetilde{\bZ},\bX_0}\left\langle\frac{\de}{\de t}H_t(\bX)\right\rangle_t
\\&=\frac{1}{NM}\Tr\E_{\bZ,\widetilde{\bZ},\bX_0}\left\langle\frac{1}{2}\sqrt{\frac{\eps}{t}}\widetilde{\bZ}\bX^\intercal+\eps\bX_0\bX^\intercal-\frac{\eps}{2}\bX\bX^\intercal\right\rangle_t
\\&=\frac{\eps}{2M}\Tr\E_{\bZ,\widetilde{\bZ},\bX_0}\langle\bR_{10}\rangle_t
\\&\ge0,
\end{align*}
with the last two lines following from Gaussian integration by parts with respect to $\widetilde{\bZ}$, the Nishimori identity, and the positive semidefiniteness of
\begin{equation*}
\E_{\bZ,\widetilde{\bZ},\bX_0}\langle\bR_{10}\rangle_t=\frac{1}{N}\E_{\bZ,\widetilde{\bZ},\bX_0}\langle\bX^\intercal\rangle_t\langle\bX\rangle_t.
\end{equation*}
By the definition of $\varphi(t)$, we then have that
\begin{align*}
&\frac{1}{NM}\E_{\bZ,\widetilde{\bZ},\bX_0}\ln\widetilde{Z}_{N,M}-\frac{1}{NM}\E_{\bZ,\bX_0}\ln Z_{N,M}
\\&=\varphi(1)-\varphi(0)
\\&=\int_0^1\varphi'(t)\,\de t
\\&\ge0,
\end{align*}
as desired.
\end{proof}

Since the left-hand side of equation \eqref{ZvsZpert} is independent of $\eps$, it is immediate from integrating both sides over $\eps\in[s_N,2s_N]$ and then taking the limit supremum that
\begin{equation} \label{Lemma9cor}
\limsup_{N\to\infty}F_N(\lambda)\le \limsup_{N\to\infty}\widetilde{F}_N(\lambda).
\end{equation}

\section{Treatment of growing rank: The multiscale cavity method} \label{s5}

The main methodological novelty begins now. In this section, we prove that the limit supremum of the free entropy $\widetilde{F}_N(\lambda)$ of the perturbed model corresponding to the Hamiltonian $\widetilde{H}_{N,\eps}(\bX)$ \eqref{Htilde} is bounded above by the supremum of the rank-one replica symmetric potential $F_1^{\mathrm{RS}}(q,\lambda)$ \eqref{F1RSpot}. By Lemma \ref{Lemma9} above, this is sufficient for proving Proposition \ref{prop3}, which combines with Proposition \ref{prop2} to produce Theorem \ref{thrm1}, as described in \S\ref{s2.3} and detailed in \S\ref{s5.4}.

For convenience, we repeat definitions of the involved Hamiltonians and cavity fields here, taking care to highlight now-relevant dependencies on SNR and various sources of noise. First, we recall that we take $(s_N)_{N\ge1}$ to be a sequence of numbers in $(0,1)$ that decrease slowly to zero as $N\to\infty$ and that for each $N\ge1$, we have $\eps_N\in[s_N,2s_N]$. Then, the Hamiltonian of the size $N$ perturbed model corresponding to the inference problem
\begin{equation*}
\begin{cases}\bY=\sqrt{\frac{\lambda}{N}}\bX_0\bX_0^\intercal+\bZ,\\ \bY^{(\eps_N)}=\sqrt{\eps_N}\bX_0+\widetilde{\bZ},\end{cases}
\end{equation*}
with $\bX_0\sim\pP_{X,N,M}$, $\bZ\in\R^{N\times N}$ Wigner, and $\widetilde{\bZ}\in\R^{N\times M}$ Ginibre is
\begin{equation} \label{Htildevar}
\widetilde{H}_{N,\eps_N}(\bX;\bZ,\widetilde{\bZ},\lambda)=H_N(\bX;\bZ,\lambda)+H_{N,\eps_N}(\bX;\widetilde{\bZ}),
\end{equation}
where the Hamiltonians corresponding to each channel are respectively
\begin{align}
H_N(\bX;\bZ,\lambda)&=\frac{1}{2} \Tr\Big( \sqrt{\frac{ \lambda }{N}}  \bZ \bX \bX^\intercal  + \frac{\lambda}{N}  \bX_0 \bX_0^\intercal \bX \bX^\intercal- \frac{\lambda}{2N}   \bX \bX^\intercal \bX \bX^\intercal  \Big),\nonumber
\\ H_{N,\eps_N}(\bX;\widetilde{\bZ})&=\Tr\left(\sqrt{\eps_N}\widetilde{\bZ}\bX^\intercal+\eps_N\bX_0\bX^\intercal-\frac{\eps_N}{2}\bX\bX^\intercal\right). \label{HNeps5}
\end{align}
The related perturbed free entropy and partition function are then
\begin{align}
\widetilde{F}_N(\lambda)&=\frac{1}{NM}\E\ln\widetilde{Z}_{N,M}, \label{Fnmtilde}
\\ \widetilde{Z}_{N,M}&=\int_{\R^{N\times M}}e^{\widetilde{H}_{N,\eps_N}(\bX;\bZ,\widetilde{\bZ},\lambda)}\,\de\!\pP_{X,N,M}(\bX), \label{Znmtilde}
\end{align}
with the expectation here and henceforth being taken over $\bZ,\widetilde{\bZ},\bX_0$, and uniformly over $\eps_N\in[s_N,2s_N]$. Letting $\bet\in\R^M$, $\bet_0\sim\pP_{X,M}$, $\bxi\sim\mathcal{N}(0,I_N)$, $\widetilde{\bxi}\sim\mathcal{N}(0,I_M)$, and $\xi\sim\mathcal{N}(0,1)$ account for the $(N+1)\textsuperscript{th}$ coordinate, we have that the size $N+1$ equivalent of the Hamiltonian $\widetilde{H}_{N,\eps_N}(\bX;\bZ,\widetilde{\bZ},\lambda)$ is given by
\begin{align}
\widetilde{H}_{N+1,\eps_{N+1}}(\bX,\bet;\bZ,\widetilde{\bZ},\bxi,\widetilde{\bxi},\xi,\lambda)&=H_N'(\bX;\bZ,\lambda)+H_{N,\eps_{N+1}}(\bX;\widetilde{\bZ})+\bet z(\bX;\bxi,\lambda)\nonumber
\\&\qquad+\bet s(\bX;\lambda)\bet^\intercal +m(\bet;\xi,\lambda)+r_{\eps_{N+1}}(\bet;\widetilde{\bxi}), \label{HtildeN+1var}
\end{align}
with cavity fields
\begin{align}
H_N'(\bX;\bZ,\lambda)&=\frac{1}{2}\Tr\Bigg( \sqrt{\frac{ \lambda }{N+1}}  \bZ \bX \bX^\intercal  + \frac{\lambda}{N+1}  \bX_0 \bX_0^\intercal \bX \bX^\intercal\nonumber
\\&\hspace{4.1cm} -\frac{\lambda}{2(N+1)}   \bX \bX^\intercal \bX \bX^\intercal  \Bigg), \label{HN'5}
\\ z(\bX;\bxi,\lambda)&=\sqrt{\frac{\lambda}{N+1}}\bX^\intercal\bxi^\intercal+\frac{\lambda}{N+1}\bX^\intercal\bX_0\bet_0^\intercal, \label{zcavity5}
\\ s(\bX;\lambda)&=-\frac{\lambda}{2(N+1)}\bX^\intercal\bX, \label{scavity5}
\\ m(\bet;\xi,\lambda)&=\frac{1}{2}\left(\sqrt{\frac{\lambda}{N+1}}\xi+\frac{\lambda}{N+1}\bet_0\bet_0^\intercal-\frac{\lambda}{2(N+1)}\bet\bet^\intercal\right)\bet\bet^\intercal, \label{mcavity5}
\\ r_{\eps_{N+1}}(\bet;\widetilde{\bxi})&=\sqrt{\eps_{N+1}}\widetilde{\bxi}\bet^\intercal+\eps_{N+1}\bet_0\bet^\intercal-\frac{\eps_{N+1}}{2}\bet\bet^\intercal. \label{rcavity5}
\end{align}
The perturbed free entropy $\widetilde{F}_{N+1}(\lambda)$ and partition function $\widetilde{Z}_{N+1,M}$ are defined as their size $N$ analogs according to equations \eqref{Fnmtilde} and \eqref{Znmtilde}, but with the expectation $\E[\,\cdot\,]$ now being taken over $\eps_{N+1}\in[s_{N+1},2s_{N+1}]$ instead of $\eps_N$ and over the additional variables $\bxi,\widetilde{\bxi},\xi,\bet_0$.

In \S\ref{s5.1}, we prove Theorem \ref{thrm3}, which shows that $\limsup_{N\to\infty}\widetilde{F}_N(\lambda)$ is bounded above by a convex combination of the limit suprema of $\Delta_N/M$ and $\Delta_M/N$, where
\begin{align}
\Delta_N&=\E\ln\widetilde{Z}_{N+1,M_{N+1}}-\E\ln\widetilde{Z}_{N,M_{N+1}}, \label{DeltaN5}
\\ \Delta_M&=\E\ln\widetilde{Z}_{N,M_{N+1}}-\E\ln\widetilde{Z}_{N,M_N}. \label{DeltaM5}
\end{align}
As $\Delta_N$ corresponds to a cavity in $N$ of fixed rank $M_{N+1}$ and vice versa for $\Delta_M$, this enables us to decompose the problem of dealing with two growing indices $M,N$ into two separate single-index cavity computations. It turns out that said convex combination has the desired upper bound $\sup_{q\in[0,\rho]}F_1^{\mathrm{RS}}(q,\lambda)$ due to both of the involved limit suprema being bounded above by this quantity. This is shown in \S\ref{s5.2} and \S\ref{s5.3}. In the former, we give a detailed presentation of the standard fixed-rank Aizenman--Sims--Starr scheme \cite{ASS}, \cite[\S4.3 \& \S6.2.5]{LelargeMiolane}, while in the latter, we prove the heuristic (valid for $N$ such that $M_{N+1}=M_N+1$)
\begin{align*}
\frac{\Delta_M}{N}&=\frac{1}{N}\E\ln\widetilde{Z}_{N,M_{N+1}}-\frac{1}{N}\E\ln\widetilde{Z}_{N,M_N}
\\&\approx(M_N+1)\sup_{q\in[0,\rho]}F_1^{\mathrm{RS}}(q,\lambda)-M_N\sup_{q\in[0,\rho]}F_1^{\mathrm{RS}}(q,\lambda)=\sup_{q\in[0,\rho]}F_1^{\mathrm{RS}}(q,\lambda).
\end{align*}

\subsection{Free entropy upper bound: Multiscale Aizenman--Sims--Starr scheme} \label{s5.1}

In this subsection, we allow the matrix size $N$ and rank $M$ to grow generically with each other. Our primary goal is to show how the standard Ces\`aro sum extends to the multiscale setting to give
\begin{equation*}
 \limsup_{N\to\infty}\widetilde F_N(\lambda)\leq \alpha\limsup_{N\to\infty}\frac{\Delta_N}{M}+(1-\alpha)\limsup_{N\to\infty}\frac{\Delta_M}{N}
 \end{equation*}
for some $0\le \alpha\le 1$.
 
\begin{proof}[\textbf{Proof of Theorem \ref{thrm3}}]
We focus on proving the upper bound \eqref{DeltaN+M}, since the analogous lower bound \eqref{DeltaN+M_inf} will follow from obvious modifications. Let $\epsilon>0$ and for $1\le n\le N$, define $m_n=M_n$ and its generalized inverse
\begin{equation*}
n_m=\inf\{n\mid m_n\ge m\}
\end{equation*}
in order to denote the $m$ index as a function of $n$ and vice versa; these dependencies are the same as between $M_N$ and $N$. By the definition of the limit supremum, there exists $T=T(\epsilon)$ independent of $N$ such that for all $n\ge T$,
\begin{equation} \label{DeltaNsupbound}
\frac{\widetilde{\Delta}_N(n)}{m_n}\le\limsup_{N\to\infty}\frac{\Delta_N}{M}+\epsilon
\end{equation}
and (considering only $n\ge T$ for which $m_{n+1}\ne m_n$)
\begin{equation} \label{DeltaMsupbound}
\frac{\widetilde{\Delta}_M(n)}{n}\le\limsup_{N\to\infty}\frac{\Delta_M}{N}+\epsilon .
\end{equation}

With $T$ in hand, we may now proceed as in the standard proof of the Aizenman--Sims-Starr scheme. Thus, express the log partition function as the telescoping sum
\begin{align}
\widetilde{F}_N(\lambda) &= \frac{1}{NM} \sum_{n=T}^{N-1} \Big(\E\ln\widetilde{Z}_{n+1,m_{n + 1}} - \E\ln\widetilde{Z}_{n,m_n}\Big)+ \mathrm{o}_N(T) \nonumber
\\&=\frac{1}{NM}\Big(\sum_{n=T}^{N-1}\widetilde{\Delta}_N(n) +\sum_{n=T}^{N-1}\widetilde{\Delta}_M(n)  \Big) + \mathrm{o}_N(T), \label{DeltaDecomposition}
\end{align}
where $\widetilde{Z}_{0,0}\equiv1$ by convention, we write
\begin{align*}
 \widetilde{\Delta}_N(n) &:=\E\ln\widetilde{Z}_{n + 1,m_{n + 1}}-\E\ln\widetilde{Z}_{n,m_{n + 1}},
 \\ \widetilde{\Delta}_M(n) &:=\E\ln\widetilde{Z}_{n,m_{n + 1}}-\E\ln\widetilde{Z}_{n,m_{n}},
 \end{align*}
and the error term is the truncated sum
\begin{equation*}
\mathrm{o}_N(T)=\frac{1}{NM}\sum_{n=0}^{T-1}\left(\widetilde{\Delta}_N(n)+\widetilde{\Delta}_M(n)\right),
\end{equation*}
which vanishes as $N\to\infty$: $T$ is independent of $N$ and for each $0\le n\le T-1$, $\widetilde{\Delta}_N(n)$ and $\widetilde{\Delta}_M(n)$ are finite $(nm_n)$-dimensional integrals that are also independent of $N$.

Notice that $\widetilde{\Delta}_M(n) = 0$ if $m_{n + 1} = m_n$, so keeping only the terms for which the $m_n$ coordinate increases, we read from the bound \eqref{DeltaMsupbound} that
\begin{equation*}
\sum_{n=T}^{N-1}\widetilde{\Delta}_M(n)\le\sum_{m=M_T}^{M_{N-1}}n_m\left(\limsup_{N\to\infty}\frac{\Delta_M}{N}+\epsilon\right).
\end{equation*}
Substituting this and the analogous statement deriving from the bound \eqref{DeltaNsupbound} into equation \eqref{DeltaDecomposition} gives the upper bound
\begin{multline}
    \widetilde{F}_{N}(\lambda) \leq \Big( \frac{1}{NM}\sum_{n=T}^{N - 1} m_n \Big) \Big( \limsup_{N\to\infty}\frac{\Delta_N}{M}  + \epsilon \Big)
    \\\quad +  \Big( \frac{1}{NM}\sum_{m=M_T}^{M_{N - 1}} n_m \Big) \Big( \limsup_{N\to\infty}\frac{\Delta_M}{N} + \epsilon \Big) + \mathrm{o}_N(T) \label{MScavupbd}.
\end{multline}
To obtain agreement with the bound \eqref{DeltaN+M}, we let
\begin{equation*}
\alpha=\limsup_{N\to\infty}\frac{1}{NM}\sum_{n=T}^{N-1} m_n
\end{equation*}
so that the proof of inequality \eqref{DeltaN+M} is completed upon showing that $0\le\alpha\le1$ and that
\begin{equation} \label{1-alpha}
\limsup_{N\to\infty}\frac{1}{NM}\sum_{m=M_T}^{M_{N - 1}} n_m=1-\alpha.
\end{equation}
As $1\le m_n\le M$ for each $1\le n\le N$, it is immediate that $1\le\sum_{n=T}^{N-1} m_n\le NM$, so $0\le\alpha\le1$. To prove the second claim, we first recall Young's inequality for increasing functions \cite{younginequality},
\begin{equation*}
ab\le\int_0^a f(x)\,\mathrm{d}x+\int_0^bf^{-1}(x)\,\mathrm{d}x,\quad a,b>0,
\end{equation*}
where $f$ is differentiable, strictly increasing, and is such that $f(0)=0$. When $b=f(a)$, this inequality is an equality that we may discretize and adopt to our setting to see that
\begin{equation*}
\sum_{n=T}^{N-1}m_n+\sum_{m=m_T}^{M_{N-1}}n_m = NM_{N-1}-T(M_T-1).
\end{equation*}
Normalizing by $NM$ and taking the limit supremum then shows that
\begin{align*}
\alpha+\limsup_{N\to\infty}\frac{1}{NM}\sum_{m=m_T}^{M_{N-1}}n_m &= \limsup_{N\to\infty}\frac{1}{NM}\Big(\sum_{n=T}^{N-1}m_n+\sum_{m=m_T}^{M_{N-1}}n_m\Big)
\\&=\limsup_{N\to\infty}\frac{NM_{N-1}-T(M_T-1)}{NM},
\end{align*}
which equals one, hence equation \eqref{1-alpha} is verified. Taking the limit supremum of both sides of inequality \eqref{MScavupbd} and then taking $\epsilon \to 0$ completes the proof of the upper bound. The lower bound follows from a similar argument to that just presented with 
\begin{equation*}
\beta =\liminf_{N\to\infty}\frac{1}{NM}\sum_{n=T}^{N-1} m_n
\end{equation*}
in place of $\alpha$.

We now show that $\alpha,\beta$ satisfies the relation \eqref{eqalpha} in the case when $M=\lfloor f(N)\rfloor$. Without loss of generality, we modify $f$ so that $f(N)\in\N$ to simplify notation. Noting that $m_n=\lfloor f(n)\rfloor$, using L'H\^opital's rule to simplify the integral upper bound of the normalized sum shows that
\begin{equation*}
\alpha = \limsup_{N\to\infty}\frac{1}{NM}\sum_{n=T}^{N-1}m_n\le\limsup_{N\to\infty}\int_T^N\frac{f(x)}{Nf(N)}\,\mathrm{d}x\le\limsup_{x\to\infty}\frac{f(x)}{f(x)+xf'(x)}.
\end{equation*}
Similarly, consideration of the integral lower bound of the same sum gives
\begin{equation*}
\beta = \liminf_{N\to\infty}\frac{1}{NM}\sum_{n=T}^{N-1}m_n\ge\liminf_{N\to\infty}\int_{T-1}^{N-1}\frac{f(x)}{Nf(N)}\,\mathrm{d}x\ge\liminf_{x\to\infty}\frac{f(x-1)}{f(x)+xf'(x)},
\end{equation*}
concluding the proof.
\end{proof}

As mentioned before, it can be observed that the above proof applies to generic sequences with two growing indices and that no properties of the log partition function are used. If $\widetilde{F}_N(\lambda)$ remains finite in the thermodynamic limit, then $\Delta_{N}$ is an order $M$ quantity and $\Delta_{M}$ is an order $N$ quantity, so the normalized differences in the spin and rank coordinates are both of constant order and both have non-negligible contribution to the limiting value when $\alpha\ne0,1$. In other words, even if $M_N\ll N$ so that adding a column occurs much less frequently than rows when the system size grows, the free entropy difference due to column addition is much larger than the one due to row addition because columns are of length $N$, while rows are of length $M$. These two effects perfectly compensate each other so that both row and column addition contributes the same order in the multiscale Aizenman--Sims--Starr scheme.

As Theorem \ref{thrm3} shows that the effects of incrementing $N$ or $M$ can be isolated to the $\Delta_N$ and $\Delta_M$ cavities, we now move on to the task of showing that the limit suprema of both of these are bounded above by the right-hand side of inequality \eqref{FrenUpperBound}.

\subsection{The standard row cavity bound} \label{s5.2}

As $\Delta_N$ \eqref{DeltaN5} is a difference of two terms that correspond to the same rank $M_{N+1}$, it suffices to bound $\limsup_{N\to\infty}\Delta_N/M$ using the standard fixed-rank arguments of \cite{ASS,LelargeMiolane} while tracking any $M$-dependencies that could become problematic in the sublinear-rank regime. Thus, we begin by showing that $\Delta_N/M$ is bounded above by a simpler analog that is obtained by showing that in the large $N$ limit, 
\begin{itemize}
\item the cavity fields $m(\bet)$ \eqref{mcavity5}, $r_{\eps_{N+1}}(\bet)$ \eqref{rcavity5}, $H_{N,\eps_{N+1}}(\bX)-H_{N,\eps_N}(\bX)$ \eqref{HNeps5} become negligible;
\item the cavity fields $z(\bX)$ \eqref{zcavity5} and $s(\bX)$ \eqref{scavity5} can be replaced by
\begin{align*}
\widehat{z}(\bX;\bxi,\lambda)&:=\sqrt{\frac{\lambda}{N}}\bX^\intercal\bxi^\intercal+\lambda\bR_{10}\bet_0^\intercal,
\\ \widehat{s}(\bX;\lambda)&:=-\frac{\lambda}{2}\bR_{11},
\end{align*}
where we recall that $\bR_{10}=\frac{1}{N}\bX^\intercal\bX_0$ \eqref{overlap} and write $\bR_{11}=\frac{1}{N}\bX_1^\intercal\bX_1=\frac{1}{N}\bX^\intercal\bX$ in accordance with equation \eqref{overlap12};
\item and the Hamiltonian $\widetilde{H}_{N,\eps_N}(\bX)$ is statistically equivalent to
\begin{equation*}
H_N(\bX)+H_{N,\eps_N}(\bX)+\widehat{y}(\bX),
\end{equation*}
where for $\widehat{\bZ}$ an independent copy of $\bZ$,
\begin{equation} \label{yhat}
\widehat{y}(\bX;\bX_0,\widehat{\bZ},\lambda)=\frac{1}{2}\Tr\Bigg(\frac{\sqrt{\lambda}}{N}\widehat{\bZ}\bX\bX^\intercal+\lambda\bR_{10}^\intercal\bR_{10}-\frac{\lambda}{2}\bR_{11}^\intercal\bR_{11}\Bigg).
\end{equation}
\end{itemize}
Then, we use the thermal concentration results of Theorem \ref{thrm4} and Corollary~\ref{Cor2} to show that said upper bound can be simplified even further to be written in terms of the replica symmetric potential $F_M^{\mathrm{RS}}(\bQ,\lambda)$ \eqref{RSpot}.

\begin{lemma}[Standard Aizenman--Sims--Starr Scheme] \label{Lemma10}
Assume hypothesis \ref{H2}, let $\pP_X$ be a centered distribution with bounded fourth moment, let $\bet_0\sim\pP_{X,M}$, let $\widetilde{Z}_{N,M}$ be specified by equation \eqref{Znmtilde}, and let $\widehat{z}(\bX)$, $\widehat{s}(\bX)$, and $\widehat{y}(\bX)$ be as above. Then,
\begin{multline} \label{Lemma10_1}
\frac{\E_{\bxi,\widetilde{\bxi},\xi,\bet_0}\E\ln\widetilde{Z}_{N+1,M}-\E\ln\widetilde{Z}_{N,M}}{M}
\\ \le\frac{1}{M}\E_{\bxi,\bet_0}\E\ln\left\langle\int_{\R^M}e^{\bet\widehat{z}(\bX;\bxi,\lambda)+\bet\widehat{s}(\bX;\lambda)\bet^\intercal}\,\de\!\pP_{X,M}(\bet)\right\rangle_{N,\eps_N}'
\\-\frac{1}{M}\E_{\widehat{\bZ}}\E\ln\left\langle e^{\widehat{y}(\bX;\bX_0,\widehat{\bZ},\lambda)}\right\rangle_{N,\eps_N}'+\Gamma'(N,M),
\end{multline}
where the unspecified expectation $\E[\,\cdot\,]$ is as in equation \eqref{Fnmtilde}, the Gibbs average $\langle\,\cdot\,\rangle_{N,\eps_N}'$ is as specified by equation \eqref{Gibbspert}, and we define, with $C$ some finite positive constant independent of $M,N$,
\begin{equation*}
\Gamma'(N,M):=C\left(\frac{M-N}{N+1}+1+s_N\right),
\end{equation*}
which is $\mathrm{o}_N(1)$ for any $M=\mathrm{o}(N)$.
\end{lemma}
\begin{proof}
We break the proof into two parts. First, we quantify the difference
\begin{multline} \label{ZN+1Mdiff}
\frac{\E_{\bxi,\widetilde{\bxi},\xi,\bet_0}\E\ln\widetilde{Z}_{N+1,M}}{M}
\\-\frac{1}{M}\E_{\bxi,\bet_0}\E\ln\left[Z_{N,\eps_N}'\left\langle\int_{\R^M}e^{\bet\widehat{z}(\bX;\bxi,\lambda)+\bet\widehat{s}(\bX;\lambda)\bet^\intercal}\,\de\!\pP_{X,M}(\bet)\right\rangle_{N,\eps_N}'\right],
\end{multline}
where we recall the definition \eqref{ZNeps'} of the partition function $Z_{N,\eps_N}'$. To this end, we introduce the interpolating Hamiltonian
\begin{multline*}
H_{1,t}(\bX,\bet):=\widetilde{H}_{N+1,ts_{N+1}\widehat{\eps}}(\bX,\bet;\bZ,\widetilde{\bZ},\bxi,\widetilde{\bxi},\xi,t\lambda)+H_N'(\bX;\bU,(1-t)\lambda)
\\+H_{N,(1-t)s_N\widehat{\eps}}(\bX;\widetilde{\bU})+\bet\widehat{z}(\bX;\bmu,(1-t)\lambda)+\bet\widehat{s}(\bX;(1-t)\lambda)\bet^\intercal,
\end{multline*}
with $\bU,\widetilde{\bU},\bmu$ being independent copies of $\bZ,\widetilde{\bZ},\bxi$ and $\widehat{\eps}$ a perturbation parameter distributed uniformly on $[1,2]$ so that $\eps_N=s_N\widehat{\eps}$ in distribution. Letting $\E_1[\,\cdot\,]$ denote expectation with respect to the signal variables $\bX_0,\bet_0$, the perturbation parameter $\widehat{\eps}$, and the noise variables $\bZ,\widetilde{\bZ},\bU,\widetilde{\bU},\bxi,\widetilde{\bxi},\xi,\bmu$, the corresponding free entropy
\begin{equation*}
\varphi_1(t):=\frac{1}{M}\E_1\ln\int_{\R^{(N+1)\times M}}e^{H_{1,t}(\bX,\bet)}\,\de\!\pP_{X,N+1,M}(\bX,\bet)
\end{equation*}
is such that $\varphi_1(1)$ and $\varphi_1(0)$ respectively equal the first and second terms of the difference \eqref{ZN+1Mdiff}. As usual, we proceed by computing the derivative of this free entropy and see upon substituting in equation \eqref{HtildeN+1var} and writing $\langle\,\cdot\,\rangle_{1,t}$ for the Gibbs average with respect to the Hamiltonian $H_{1,t}(\bX,\bet)$ that
\begin{align*}
\varphi_1'(t)&=\frac{1}{M}\E_1\left\langle\frac{\de}{\de t}H_{1,t}(\bX,\bet)\right\rangle_{1,t}
\\&=\frac{1}{M}\E_1\left\langle\frac{\de}{\de t}\left[H_N'(\bX;\bZ,t\lambda)+H_N'(\bX;\bU,(1-t)\lambda)\right]\right\rangle_{1,t}
\\&\quad+\frac{1}{M}\E_1\left\langle\frac{\de}{\de t}\left[H_{N,ts_{N+1}\widehat{\eps}}(\bX;\widetilde{\bZ})+H_{N,(1-t)s_N\widehat{\eps}}(\bX;\widetilde{\bU})\right]\right\rangle_{1,t}
\\&\quad+\frac{1}{M}\E_1\left\langle\frac{\de}{\de t}\bet \left[z(\bX;\bxi,t\lambda)+\widehat{z}(\bX;\bmu,(1-t)\lambda)\right]\right\rangle_{1,t}
\\&\quad+\frac{1}{M}\E_1\left\langle\frac{\de}{\de t}\bet\left[s(\bX;t\lambda)+\widehat{s}(\bX;(1-t)\lambda)\right]\bet^\intercal\right\rangle_{1,t}
\\&\quad+\frac{1}{M}\E_1\left\langle\frac{\de}{\de t}m(\bet;\xi,t\lambda)\right\rangle_{1,t}+\frac{1}{M}\E_1\left\langle\frac{\de}{\de t}r_{ts_{N+1}\widehat{\eps}}(\bet;\widetilde{\bxi})\right\rangle_{1,t}.
\end{align*}
By Gaussian integration by parts and the Nishimori identity, we see that
\begin{equation*}
\frac{1}{M}\E_1\left\langle\frac{\de}{\de t}\left[H_N'(\bX;\bZ,t\lambda)+H_N'(\bX;\bU,(1-t)\lambda)\right]\right\rangle_{1,t}=0,
\end{equation*}
\begin{multline*}
\frac{1}{M}\E_1\left\langle\frac{\de}{\de t}\left[H_{N,ts_{N+1}\widehat{\eps}}(\bX;\widetilde{\bZ})+H_{N,(1-t)s_N\widehat{\eps}}(\bX;\widetilde{\bU})\right]\right\rangle_{1,t}
\\=\frac{N(s_{N+1}-s_N)}{2M}\Tr\E_1\langle\widehat{\eps}\bR_{10}\rangle_{1,t},
\end{multline*}
\begin{multline*}
\frac{1}{M}\E_1\left\langle\frac{\de}{\de t}\bet \left[z(\bX;\bxi,t\lambda)+\widehat{z}(\bX;\bmu,(1-t)\lambda)\right]\right\rangle_{1,t}
\\=-\frac{\lambda}{2M(N+1)}\E_1\left\langle\bet\bR_{10}\bet_0^\intercal+\bet\bR_{11}\bet^\intercal\right\rangle_{1,t},
\end{multline*}
\begin{equation*}
\frac{1}{M}\E_1\left\langle\frac{\de}{\de t}\bet\left[s(\bX;t\lambda)+\widehat{s}(\bX;(1-t)\lambda)\right]\bet^\intercal\right\rangle_{1,t}=\frac{1}{2M(N+1)}\E_1\left\langle\bet\bR_{11}\bet^\intercal\right\rangle_{1,t},
\end{equation*}
\begin{equation*}
\frac{1}{M}\E_1\left\langle\frac{\de}{\de t}m(\bet;\xi,t\lambda)\right\rangle_{1,t}=\frac{\lambda}{4M(N+1)}\E_1\left\langle\bet_0\bet_0^\intercal\bet\bet^\intercal\right\rangle_{1,t},
\end{equation*}
\begin{equation*}
\frac{1}{M}\E_1\left\langle\frac{\de}{\de t}r_{ts_{N+1}\widehat{\eps}}(\bet;\widetilde{\bxi})\right\rangle_{1,t}=\frac{s_{N+1}}{2M}\E_1\left\langle\widehat{\eps}\bet_0\bet^\intercal\right\rangle_{1,t}.
\end{equation*}
Let $\E_1'[\,\cdot\,]$ denote the same expectation as $\E_1[\,\cdot\,]$ but without the average over $\widehat{\eps}$ so that $\E_1[\,\cdot\,]=\E_{\widehat{\eps}}\E_1'[\,\cdot\,]$. Due to the same reasoning as behind the series of inequalities \eqref{Rbound}, we have that
\begin{equation*}
0\le\Tr\E_1'\langle\bR_{10}\rangle_{1,t},\E_1\langle\bet\bR_{10}\bet_0^\intercal\rangle_{1,t},\frac{1}{M}\E_1\langle\bet_0\bet_0^\intercal\bet\bet^\intercal\rangle_{1,t},\E_1'\langle\bet_0\bet^\intercal\rangle_{1,t}\le CM
\end{equation*}
for some finite positive constant $C$ independent of $M,N$ and depending only on properties of $\pP_X$. Then, since $s_N$ is a decreasing sequence, we see that
\begin{equation*}
\frac{N(s_{N+1}-s_N)}{2M}\Tr\E_1\langle\widehat{\eps}\bR_{10}\rangle_{1,t} = \frac{N(s_{N+1}-s_N)}{2M}\E_{\widehat{\eps}}\left[\widehat{\eps}\Tr\E_1'\langle\bR_{10}\rangle_{1,t}\right] \le 0,
\end{equation*}
as we have the product of the expectation of a positive trace and the negative quantity $N(s_{N+1}-s_N)/(2M)$. Likewise, we have that the other terms in $\varphi_1'(t)$ can be bounded according to
\begin{multline*}
-\frac{\lambda}{2M(N+1)}\E_1\left\langle\bet\bR_{10}\bet_0^\intercal+\bet\bR_{11}\bet^\intercal\right\rangle_{1,t}+\frac{1}{2M(N+1)}\E_1\left\langle\bet\bR_{11}\bet^\intercal\right\rangle_{1,t}
\\= -\frac{\lambda}{2M(N+1)}\E_1\left\langle\bet\bR_{10}\bet_0^\intercal\right\rangle_{1,t}\le 0,
\end{multline*}
\begin{equation*}
\frac{\lambda}{4M(N+1)}\E_1\left\langle\bet_0\bet_0^\intercal\bet\bet^\intercal\right\rangle_{1,t} \le \frac{C\lambda M}{4(N+1)},
\end{equation*}
\begin{equation*}
\frac{s_{N+1}}{2M}\E_1\left\langle\widehat{\eps}\bet_0\bet^\intercal\right\rangle_{1,t}=\frac{s_{N+1}}{2M}\E_{\widehat{\eps}}\left[\widehat{\eps}\,\E_1'\langle\bet_0\bet^\intercal\rangle_{1,t}\right] \le \frac{Cs_{N+1}}{2}\E_{\widehat{\eps}}\,\widehat{\eps} \le 3Cs_N.
\end{equation*}
Combining all of this together shows that
\begin{equation*}
\varphi_1'(t)\le\frac{C\lambda M}{4(N+1)}+3Cs_N.
\end{equation*}
Since this is independent of $t$, we have that
\begin{equation} \label{Gprime1}
\varphi_1(1)-\varphi_1(0)=\int_0^1\varphi_1'(t)\,\de t \le \frac{C\lambda M}{4(N+1)}+3Cs_N,
\end{equation}
so the difference \eqref{ZN+1Mdiff} obeys the same upper bound.

We now move on to quantifying the difference
\begin{equation} \label{ZNMdiff}
\frac{1}{M}\E_{\widehat{\bZ}}\E\ln\left[Z_{N,\eps_N}'\left\langle e^{\widehat{y}(\bX;\bX_0,\widehat{\bZ},\lambda)}\right\rangle_{N,\eps_N}'\right]-\frac{1}{M}\E\ln\widetilde{Z}_{N,M}.
\end{equation}
The relevant interpolating Hamiltonian and free entropy are given by
\begin{align*}
H_{2,t}(\bX)&:=H_N'(\bX;\bZ,t\lambda)+H_{N,t\eps_N}(\bX;\widetilde{\bZ})+\widehat{y}(\bX;\bX_0,\widehat{\bZ},t\lambda)
\\&\quad\;+\widetilde{H}_{N,(1-t)\eps_N}(\bX;\bU,\widetilde{\bU},(1-t)\lambda),
\\ \varphi_2(t)&:=\frac{1}{M}\E_2\ln\int_{\R^{N\times M}}e^{H_{2,t}(\bX)}\,\de\!\pP_{X,N,M}(\bX),
\end{align*}
where $\bU,\widetilde{\bU}$ are independent copies of $\bZ,\widetilde{\bZ}$, as before, and $\E_2[\,\cdot\,]$ denotes an expectaton with respect to $\bX_0,\bZ,\widetilde{\bZ},\widehat{\bZ},\bU,\widetilde{\bU}$, and $\eps_N$. Observe that $\varphi_2(1)$ and $\varphi_2(0)$ respectively equal the first and second terms of the difference \eqref{ZNMdiff}. Writing $\langle\,\cdot\,\rangle_{2,t}$ for the Gibbs average with respect to $H_{2,t}(\bX)$ and using equation \eqref{Htildevar} shows that
\begin{align*}
\varphi_2'(t)&=\frac{1}{M}\E_2\left\langle\frac{\de}{\de t}H_{2,t}(\bX)\right\rangle_{2,t}
\\&=\frac{1}{M}\E_2\left\langle\frac{\de}{\de t}\left[H_N'(\bX;\bZ,t\lambda)+H_N(\bX;\bU,(1-t)\lambda)\right]\right\rangle_{2,t}
\\&\quad+\frac{1}{M}\E_2\left\langle\frac{\de}{\de t}\left[H_{N,t\eps_N}(\bX;\widetilde{\bZ})+H_{N,(1-t)\eps_N}(\bX;\widetilde{\bU})\right]\right\rangle_{2,t}
\\&\quad+\frac{1}{M}\E_2\left\langle\frac{\de}{\de t}\widehat{y}(\bX;\bX_0,\widehat{\bZ},t\lambda)\right\rangle_{2,t}.
\end{align*}
By Gaussian integration by parts and the Nishimori identity, we see that
\begin{multline*}
\frac{1}{M}\E_2\left\langle\frac{\de}{\de t}\left[H_N'(\bX;\bZ,t\lambda)+H_N(\bX;\bU,(1-t)\lambda)\right]\right\rangle_{2,t}
\\=-\frac{\lambda N}{4(N+1)M}\Tr\E_2\left\langle\bR_{10}^\intercal\bR_{10}\right\rangle_{2,t},
\end{multline*}
\begin{equation*}
\\ \frac{1}{M}\E_2\left\langle\frac{\de}{\de t}\left[H_{N,t\eps_N}(\bX;\widetilde{\bZ})+H_{N,(1-t)\eps_N}(\bX;\widetilde{\bU})\right]\right\rangle_{2,t}=0,
\end{equation*}
\begin{equation*}
\frac{1}{M}\E_2\left\langle\frac{\de}{\de t}\widehat{y}(\bX;\bX_0,\widehat{\bZ},t\lambda)\right\rangle_{2,t}=\frac{\lambda}{4M}\Tr\E_2\left\langle\bR_{10}^\intercal\bR_{10}\right\rangle_{2,t}.
\end{equation*}
Combining these equalities shows that
\begin{equation*}
\varphi_2'(t) = \frac{\lambda}{4M}\left(1-\frac{N}{N+1}\right)\Tr\E_2\left\langle\bR_{10}^\intercal\bR_{10}\right\rangle_{2,t}.
\end{equation*}
Extending the logic behind the bounds \eqref{Rbound} shows that
\begin{multline*}
0\le \Tr\E_2\left\langle\bR_{10}^\intercal\bR_{10}\right\rangle_{2,t} \le \frac{1}{N^2}\Tr\E_2\left[\bX_0^\intercal\bX_0\bX_0^\intercal\bX_0\right]
\\ =\frac{1}{N^2}\sum_{i,j=1}^N\sum_{\ell,\ell'=1}^M\E_2\left[\bx_{0,i\ell}\bx_{0,i\ell'}\bx_{0,j\ell'}\bx_{0,j\ell}\right] \le C\left(M+\tfrac{M^2}{N}\right),
\end{multline*}
where the last inequality follows from the fact that, since the entries of $\bX_0$ are i.i.d. with bounded fourth moments, $\E_2\left[\bx_{0,i\ell}\bx_{0,i\ell'}\bx_{0,j\ell'}\bx_{0,j\ell}\right]$ is given by some finite positive constant $C$ if either $i=j$ or $\ell=\ell'$ and is otherwise zero -- the summand is non-zero for $\mathrm{o}(MN^2+M^2N)$ combinations of $i,j,\ell,\ell'$. As $M\le N$, we hence see that
\begin{equation*}
\varphi_2'(t) \le \frac{C\lambda}{4}\left(1-\frac{N}{N+1}\right)
\end{equation*}
and integrating this over $t$ shows that the difference \eqref{ZNMdiff} also obeys this upper bound. Summing this with the upper bound \eqref{Gprime1} on the difference \eqref{ZN+1Mdiff} produces $\Gamma'(N,M)$. Thus, we are done upon noting that $\E_{\bxi,\bet_0}\E\ln Z_{N,\eps_N}'=\E_{\widehat{\bZ}}\E\ln Z_{N,\eps_N}'$.
\end{proof}

Lemma \ref{Lemma10} displays a similarity in structure to the replica symmetric potential $F_M^{\mathrm{RS}}(\bQ,\lambda)$ \eqref{RSpot}. Our goal now is to refine this observation. Writing
\begin{equation} \label{HMNRS}
\widehat{H}_{M,N}(\bA,\bB,\bx;\bz^{(N)},\lambda):=\sqrt{\frac{\lambda}{N}}\bx^\intercal\bA^\intercal\bz^{(N)}+\frac{\lambda}{N}\bx^\intercal\bA^\intercal\bB\bx_0-\frac{\lambda}{2N}\bx^\intercal\bA^\intercal\bA\bx
\end{equation}
for $\bA,\bB\in\R^{N\times M}$, $\bx\in\R^M$, $\bz^{(N)}\sim\mathcal{N}(0,I_N)$, and $\bx_0\sim\pP_{X,M}$, one observes upon replacing $\bxi^\intercal,\bet^\intercal,\bet_0^\intercal$ respectively by $\bz^{(N)},\bx,\bx_0$ that the first term in the right-hand side of inequality \eqref{Lemma10_1} is
\begin{equation*}
\frac{1}{M}\E_{\bz^{(N)},\bx_0}\E\ln\left\langle\int_{\R^M}e^{\widehat{H}_{M,N}(\bX,\bX_0,\bx;\bz^{(N)},\lambda)}\,\de\!\pP_{X,M}(\bx)\right\rangle_{N,\eps}'.
\end{equation*}
On the other hand, further replacing $\bX$ and $\bX_0$ by $\langle\bX\rangle_{N,\eps}'$ in the integrand produces
\begin{align}
&\frac{1}{M}\E_{\bz^{(N)},\bx_0}\E\ln\left\langle\int_{\R^M}e^{\widehat{H}_{M,N}(\langle\bX\rangle_{N,\eps}',\langle\bX\rangle_{N,\eps}',\bx;\bz^{(N)},\lambda)}\,\de\!\pP_{X,M}(\bx)\right\rangle_{N,\eps}' \nonumber
\\=&\frac{1}{M}\E_{\bz^{(N)},\bx_0}\E\ln\int_{\R^M}e^{\widehat{H}_{M,N}(\langle\bX\rangle_{N,\eps}',\langle\bX\rangle_{N,\eps}',\bx;\bz^{(N)},\lambda)}\,\de\!\pP_{X,M}(\bx) \label{HMNRS1}
\\=&\frac{1}{M}\E_{\bz^{(M)},\bx_0}\E\ln\int_{\R^M}e^{\widehat{H}_{M,M}(\sqrt{\bT},\sqrt{\bT},\bx;\bz^{(M)},\lambda N)}\,\de\!\pP_{X,M}(\bx), \label{HMMRS}
\end{align}
where we have defined
\begin{equation} \label{Tdef}
\bT:=\langle\bR_{12}\rangle_{N,\eps}'=\frac{1}{N}\langle\bX^\intercal\rangle_{N,\eps}'\langle\bX\rangle_{N,\eps}'=\frac{1}{N}\sum_{i=1}^N\langle\bx_i\rangle_{N,\eps}'\langle\bx_i^\intercal\rangle_{N,\eps}'.
\end{equation}
Here, equation \eqref{HMNRS1} is due to the fact that the integral is independent of $\bX$, which is what the Gibbs average $\langle\,\cdot\,\rangle_{N,\eps}'$ is over, while equation \eqref{HMMRS} follows from the fact that the Gaussian vectors $\langle\bX^\intercal\rangle_{N,\eps}'\bz^{(N)}$ and $\sqrt{N\bT}\bz^{(M)}$ are equal in distribution, since $\bz^{(N)},\bz^{(M)}$ are independent of $\langle\bX\rangle_{N,\eps}'$. 

Similarly, substituting $\bX$ and $\bX_0$ by $\langle\bX\rangle_{N,\eps}'$ in the exponent of the second term in the right-hand side of inequality~\eqref{Lemma10_1} yields the simplification
\begin{align*}
\frac{1}{M}\E_{\widehat{\bZ}}\E\ln\left\langle e^{\widehat{y}(\langle\bX\rangle_{N,\eps}';\langle\bX\rangle_{N,\eps}',\widehat{\bZ},\lambda)}\right\rangle_{N,\eps}'&=\frac{1}{M}\E_{\widehat{\bZ}}\E\left[\widehat{y}(\langle\bX\rangle_{N,\eps}';\langle\bX\rangle_{N,\eps}',\widehat{\bZ},\lambda)\right]
\\&=\frac{\lambda}{4M}\Tr\E\,\bT^2,
\end{align*}
where the first equality follows as above, while the second is due to the independence of $\widehat{\bZ}$ and $\langle\bX\rangle_{N,\eps}'$. Combining this with equation~\eqref{HMMRS} shows that replacing $\bX$ and $\bX_0$ by $\langle\bX\rangle_{N,\eps}'$ within the averages in the right-hand side of inequality \eqref{Lemma10_1} produces the averaged replica symmetric potential $\E F_M^{\mathrm{RS}}(\bT,\lambda)$. Connection with Theorem~\ref{thrm1} is made upon noting that $\bT$ is symmetric positive semidefinite, replacing the average with a supremum, and using Theorem \ref{thrm2} to establish that
\begin{equation*}
\E F_M^{\mathrm{RS}}(\bT,\lambda) \le \sup_{\bQ\in\cS_M}F_M^{\mathrm{RS}}(\bQ,\lambda)=\sup_{q\in[0,\rho]}F_1^{\mathrm{RS}}(q,\lambda).
\end{equation*}

The driving mechanism behind the above is the action of replacing $\bX$ and $\bX_0$ by $\langle\bX\rangle_{N,\eps}'$. This is allowed, with a negligible cost, due to the concentration results proved in \S\ref{s4.2}. We now make this precise through a slight rewriting of the arguments of \cite{LelargeMiolane}.

\begin{lemma}[First application of overlap concentration] \label{Lemma11}
Assume hypotheses \ref{H2} and \ref{H3}. Define
\begin{align*}
U_1&:=\left\langle\int_{\R^M}e^{\widehat{H}_{M,N}(\bX,\bX_0,\bx;\bz,\lambda)}\,\de\!\pP_{X,M}(\bx)\right\rangle_{N,\eps}',
\\ V_1&:=\int_{\R^M}e^{\widehat{H}_{M,N}(\langle\bX\rangle_{N,\eps}',\langle\bX\rangle_{N,\eps}',\bx;\bz,\lambda)}\,\de\!\pP_{X,M}(\bx)
\end{align*}
with $\bz\sim\mathcal{N}(0,I_N)$, $\widehat{H}_{M,N}(\bA,\bB,\bx;\bz,\lambda)$ as specified by equation \eqref{HMNRS}, and the Gibbs average $\langle\,\cdot\,\rangle_{N,\eps_N}'$ as in equation \eqref{Gibbspert}. Then, there is a positive constant $C$ independent of $M,N$ such that
\begin{equation*}
\left|\E'\E_{\bz}\ln U_1-\E'\E_{\bz}\ln V_1\right| \le C\sqrt{M}e^{CM^2}\Gamma(N,M)^{1/4},
\end{equation*}
where $\Gamma(N,M)=CM^2/\sqrt{N s_N}$ is the rate of concentration given in Theorem \ref{thrm4} and we write $\E'[\,\cdot\,]$ for the expectation over $\bZ,\widetilde{\bZ},\bX_0,\bx_0$, and uniformly over $\eps\in[s_N,2s_N]$.
\end{lemma}
\begin{proof}
First observe that since $|\ln(x)-\ln(y)|\le\max\{x^{-1},y^{-1}\}\,|x-y|$, we have by Jensen's inequality and the Cauchy--Schwarz inequality that
\begin{align}
&\left|\E'\E_{\bz}\ln U_1-\E'\E_{\bz}\ln V_1\right| \nonumber
\\&\le \E'\E_{\bz}\left|\ln U_1-\ln V_1\right| \nonumber
\\&\le\E'\sqrt{\E_{\bz}[U_1^{-2}+V_1^{-2}]}\sqrt{\E_{\bz}[(U_1-V_1)^2]} \nonumber
\\&\le\E'\sqrt{\E_{\bz}[U_1^{-2}+V_1^{-2}]}\sqrt{2|\E_{\bz}U_1^2-\E_{\bz}V_1^2|+2|\E_{\bz}[U_1V_1]-\E_{\bz}V_1^2|}. \label{Lem11_0}
\end{align}
Thus, we bound each of the terms $\E_{\bz}[U_1^{-2}+V_1^{-2}]$, $|\E_{\bz}U_1^2-\E_{\bz}V_1^2|$, $|\E_{\bz}[U_1V_1]-\E_{\bz}V_1^2|$ in a series of three steps and then simplify the above using Corollary \ref{Cor2} of Theorem \ref{thrm4}.

\textit{Step 1:} Applying Jensen's inequality shows that
\begin{equation*}
U_1\ge\int_{\R^M}e^{\langle\widehat{H}_{M,N}(\bX,\bX_0,\bx;\bz,\lambda)\rangle_{N,\eps}'}\,\de\!\pP_{X,M}(\bx).
\end{equation*}
Thus, using Jensen's inequality again shows that
\begin{equation*}
U_1^{-2}\le\int_{\R^M}e^{-2\langle\widehat{H}_{M,N}(\bX,\bX_0,\bx;\bz,\lambda)\rangle_{N,\eps}'}\,\de\!\pP_{X,M}(\bx).
\end{equation*}
Now, using the fact that $\E_ze^{az}=e^{a^2/2}$ for $z\sim\mathcal{N}(0,1)$ and $a$ independent of $z$ gives
\begin{equation*}
\E_{\bz}U_1^{-2}\le\int_{\R^M}e^{2\lambda\bx^\intercal\bT\bx-\lambda\langle2\bx^\intercal\bR_{10}\bx_0-\bx^\intercal\bR_{11}\bx\rangle_{N,\eps}'}\,\de\!\pP_{X,M}(\bx),
\end{equation*}
where we recall that $\bT=\frac{1}{N}\langle\bX^\intercal\rangle_{N,\eps}'\langle\bX\rangle_{N,\eps}'$ \eqref{Tdef}. Since $\pP_X$ has $D$-bounded support, the entries of $\bX,\bX_0$ lie in $[-D,D]$, so $|\bx^\intercal\bT\bx|,|\bx^\intercal\bR_{10}\bx_0|,|\bx^\intercal\bR_{11}\bx|\le M^2D^4$ and we obtain the upper bound $\E_{\bz}U_1^{-2}\le e^{5\lambda M^2D^4}$. The exact same argument shows that we also have $\E_{\bz}V_1^{-2}\le e^{5\lambda M^2D^4}$, hence
\begin{equation} \label{Lem11_1}
\E_{\bz}[U_1^{-2}+V_1^{-2}] \le 2e^{5\lambda M^2D^4}.
\end{equation}

\textit{Step 2:} Letting $\langle\,\cdot\,\rangle_{N,\eps}'$ now refer to a Gibbs average with respect to the replicas $\bX_1,\bX_2$ of $\bX$, we have that
\begin{align*}
\E_{\bz}U_1^2&=\E_{\bz}\left\langle\int_{\R^M\times\R^M}e^{\widehat{H}_{M,N}(\bX_1,\bX_0,\bx_1;\bz,\lambda)+\widehat{H}_{M,N}(\bX_2,\bX_0,\bx_2;\bz,\lambda)}\,\de\!\pP_{X,M}^{\otimes 2}(\bx_1,\bx_2)\right\rangle_{N,\eps}'
\\&=\left\langle\int_{\R^M\times\R^M}e^{\lambda(\bx_1^\intercal\bR_{10}\bx_0+\bx_2^\intercal\bR_{20}\bx_0+\bx_1^\intercal\bR_{12}\bx_2)}\,\de\!\pP_{X,M}^{\otimes 2}(\bx_1,\bx_2)\right\rangle_{N,\eps}',
\end{align*}
where we have used the fact that
\begin{equation*}
\E_{\bz}\exp\left(\sqrt{\frac{\lambda}{N}}\left(\bX_1\bx_1+\bX_2\bx_2\right)^\intercal\bz\right)=\exp\left(\frac{\lambda}{2N}\lVert\bX_1\bx_1+\bX_2\bx_2\rVert_{\mathrm{F}}^2\right).
\end{equation*}
Defining $F_1:([-D^2,D^2]^{M\times M})^3\to\R$ by
\begin{equation*}
F_1(\bA,\bB,\bC):=\int_{\R^M\times\R^M}e^{\lambda(\bx_1^\intercal\bA\bx_0+\bx_2^\intercal\bB\bx_0+\bx_1^\intercal\bC\bx_2)}\,\de\!\pP_{X,M}^{\otimes 2}(\bx_1,\bx_2),
\end{equation*}
we observe that
\begin{equation*}
\E_{\bz}U_1^2=\langle F_1(\bR_{10},\bR_{20},\bR_{12})\rangle_{N,\eps}'.
\end{equation*}
Likewise, we have that
\begin{equation*}
\E_{\bz}V_1^2= F_1(\bT,\bT,\bT).
\end{equation*}
As $\pP_X$ has $D$-bounded support, we compute for all $\bA,\bB,\bC\in[-D^2,D^2]^{M\times M}$ that
\begin{equation*}
\lVert\nabla F_1(\bA,\bB,\bC)\rVert_{\mathrm{F}}\le\sqrt{3}\lambda M D^2e^{3\lambda M^2D^4}.
\end{equation*}
Thus, $F_1(\bA,\bB,\bC)$ is $(C_1 Me^{C_2M^2})$-Lipschitz for some constants $C_1,C_2$ independent of $M,N$. Hence, using Jensen's inequality and our Lipschitz bound shows that
\begin{align}
|\E_{\bz}U_1^2-\E_{\bz}V_1^2|&=|\langle F_1(\bR_{10},\bR_{20},\bR_{12})\rangle_{N,\eps}'-F_1(\bT,\bT,\bT)| \nonumber
\\&\le\langle|F_1(\bR_{10},\bR_{20},\bR_{12})-F_1(\bT,\bT,\bT)|\rangle_{N,\eps}' \nonumber
\\&\le C_1 Me^{C_2M^2}\left\langle\sqrt{\lVert\bR_{10}-\bT\rVert_{\mathrm{F}}^2+\lVert\bR_{20}-\bT\rVert_{\mathrm{F}}^2+\lVert\bR_{12}-\bT\rVert_{\mathrm{F}}^2}\right\rangle_{N,\eps}'. \label{Lem11_2}
\end{align}

\textit{Step 3:} From the same reasoning as in the above step, we see that 
\begin{align*}
&\E_{\bz}[U_1V_1]
\\&=\E_{\bz}\left\langle\int_{\R^M\times\R^M}e^{\widehat{H}_{M,N}(\bX,\bX_0,\bx_1;\bz,\lambda)+\widehat{H}_{M,N}(\langle\bX\rangle_{N,\eps}',\langle\bX\rangle_{N,\eps}',\bx_2;\bz,\lambda)}\,\de\!\pP_{X,M}^{\otimes 2}(\bx_1,\bx_2)\right\rangle_{N,\eps}'
\\&=\langle F_1(\bR_{10},\bT,\tfrac{1}{N}\bX^\intercal\langle\bX\rangle_{N,\eps}')\rangle_{N,\eps}'.
\end{align*}
Hence, using Jensen's inequality and our Lipschitz bound on $F_1(\bA,\bB,\bC)$ yet again shows that
\begin{align}
|\E_{\bz}[U_1V_1]-\E_{\bz}V_1^2|&=|\langle F_1(\bR_{10},\bT,\tfrac{1}{N}\bX^\intercal\langle\bX\rangle_{N,\eps}')\rangle_{N,\eps}'-F_1(\bT,\bT,\bT)| \nonumber
\\&\le\langle|F_1(\bR_{10},\bT,\tfrac{1}{N}\bX^\intercal\langle\bX\rangle_{N,\eps}')-F_1(\bT,\bT,\bT)|\rangle_{N,\eps}' \nonumber
\\&\le C_1Me^{C_2M^2}\left\langle\sqrt{\lVert\bR_{10}-\bT\rVert_{\mathrm{F}}^2+\lVert\tfrac{1}{N}\bX^\intercal\langle\bX\rangle_{N,\eps}'-\bT\rVert_{\mathrm{F}}^2}\right\rangle_{N,\eps}'. \label{Lem11_3}
\end{align}

\textit{Step 4:} Applying Jensen's inequality and Corollary \ref{Cor2} to inequality \eqref{Lem11_2} shows that
\begin{align}
\E'|\E_{\bz}U_1^2-\E_{\bz}V_1^2|&\le C_1Me^{C_2M^2}\sqrt{\E'\langle\lVert\bR_{10}-\bT\rVert_{\mathrm{F}}^2+\lVert\bR_{20}-\bT\rVert_{\mathrm{F}}^2+\lVert\bR_{12}-\bT\rVert_{\mathrm{F}}^2\rangle_{N,\eps}'} \nonumber
\\&\le C_1Me^{C_2M^2}\sqrt{3\Gamma(N,M)}, \label{Lem11_4.1}
\end{align}
where we have used the Nishimori identity to replace $\bR_{20}$ by $\bR_{10}$. Doing the same to inequality \eqref{Lem11_3} shows that
\begin{align}
\E'|\E_{\bz}[U_1V_1]-\E_{\bz}V_1^2|&\le C_1Me^{C_2M^2}\sqrt{\Gamma(N,M)+\E'\langle\lVert\tfrac{1}{N}\bX^\intercal\langle\bX\rangle_{N,\eps}'-\bT\rVert_{\mathrm{F}}^2\rangle_{N,\eps}'} \nonumber
\\&\le C_1Me^{C_2M^2}\sqrt{2\Gamma(N,M)}, \label{Lem11_4.2}
\end{align}
where in the last line we have used the fact that, if we let $\langle\,\cdot\,\rangle$ denote the Gibbs average $\langle\,\cdot\,\rangle_{N,\eps}'$ with respect to only $\bX_2$,
\begin{align*}
\E'\langle\lVert\tfrac{1}{N}\bX^\intercal\langle\bX\rangle_{N,\eps}'-\bT\rVert_{\mathrm{F}}^2\rangle_{N,\eps}'&=\E'\langle\lVert\langle\tfrac{1}{N}\bX_1^\intercal\bX_2-\bT\rangle\rVert_{\mathrm{F}}^2\rangle_{N,\eps}'
\\&\le \E'\langle\langle\lVert\bR_{12}-\bT\rVert_{\mathrm{F}}^2\rangle\rangle_{N,\eps}'
\\&=\E'\langle\lVert\bR_{12}-\bT\rVert_{\mathrm{F}}^2\rangle_{N,\eps}'
\\&\le \Gamma(N,M).
\end{align*}

Returning to inequality \eqref{Lem11_0}, inserting inequality \eqref{Lem11_1}, using Jensen's inequality, and then finally using inequalities \eqref{Lem11_4.1}, \eqref{Lem11_4.2}, we see that there is a finite positive constant $C$ independent of $M,N$ such that
\begin{align*}
&\left|\E'\E_{\bz}\ln U_1-\E'\E_{\bz}\ln V_1\right|
\\&\le Ce^{CM^2}\E'\sqrt{2|\E_{\bz}U_1^2-\E_{\bz}V_1^2|+2|\E_{\bz}[U_1V_1]-\E_{\bz}V_1^2|}
\\&\le Ce^{CM^2}\sqrt{\E'|\E_{\bz}U_1^2-\E_{\bz}V_1^2|+\E'|\E_{\bz}[U_1V_1]-\E_{\bz}V_1^2|}
\\&\le C\sqrt{M}e^{CM^2}\Gamma(N,M)^{1/4},
\end{align*}
as desired.
\end{proof}

\begin{lemma}[Second application of overlap concentration] \label{Lemma12}
Assume hypotheses \ref{H2} and \ref{H3}. Define
\begin{align*}
U_2&:=\left\langle e^{\widehat{y}(\bX;\bX_0,\widehat{\bZ},\lambda)}\right\rangle_{N,\eps}',
\\ V_2&:= e^{\widehat{y}(\langle\bX\rangle_{N,\eps}';\langle\bX\rangle_{N,\eps}',\widehat{\bZ},\lambda)}
\end{align*}
with $\widehat{y}(\bX;\bX_0,\widehat{\bZ},\lambda)$ as specified by equation \eqref{yhat} and the Gibbs average $\langle\,\cdot\,\rangle_{N,\eps_N}'$ as in equation \eqref{Gibbspert}. Then, there is a positive constant $C$ independent of $M,N$ such that
\begin{equation*}
\left|\E'\E_{\widehat{\bZ}}\ln U_2-\E'\E_{\widehat{\bZ}}\ln V_2\right| \le C\sqrt{M}e^{CM^2}\Gamma(N,M)^{1/4},
\end{equation*}
where $\Gamma(N,M)=CM^2/\sqrt{N s_N}$ is the rate of concentration given in Theorem \ref{thrm4} and $\E'[\,\cdot\,]$ retains its meaning from Lemma \ref{Lemma11}.
\end{lemma}
\begin{proof}
The structure of the proof is identical to that of Lemma \ref{Lemma11}. Indeed, as before, we have that
\begin{multline} \label{Lem12_0}
|\E'\E_{\widehat{\bZ}}\ln U_2-\E'\E_{\widehat{\bZ}}\ln V_2|
\\ \le\E'\sqrt{\E_{\widehat{\bZ}}[U_2^{-2}+V_2^{-2}]}\sqrt{2|\E_{\widehat{\bZ}}U_2^2-\E_{\widehat{\bZ}}V_2^2|+2|\E_{\widehat{\bZ}}[U_2V_2]-\E_{\widehat{\bZ}}V_2^2|}
\end{multline}
and we proceed by bounding each of the terms $\E_{\widehat{\bZ}}[U_2^{-2}+V_2^{-2}]$, $|\E_{\widehat{\bZ}}U_2^2-\E_{\widehat{\bZ}}V_2^2|$, $|\E_{\widehat{\bZ}}[U_2V_2]-\E_{\widehat{\bZ}}V_2^2|$ in three steps analogous to those of the proof of Lemma \ref{Lemma11}, before applying Corollary \ref{Cor2}.

\textit{Step 1:} Using Jensen's inequality shows that
\begin{equation*}
U_2^{-2}\le e^{-2\langle\widehat{y}(\bX;\bX_0,\widehat{\bZ},\lambda)\rangle_{N,\eps}'},\qquad V_2^{-2}\le e^{-2\widehat{y}(\langle\bX\rangle_{N,\eps}';\langle\bX\rangle_{N,\eps}',\widehat{\bZ},\lambda)}.
\end{equation*}
Then, since
\begin{equation} \label{Gaussianint}
\E_{\widehat{\bZ}}e^{\Tr(\widehat{\bZ}\bA)}=e^{\Tr\bA^2}
\end{equation}
for symmetric $\bA\in\R^{N\times N}$ independent of $\widehat{\bZ}$, we have that
\begin{align*}
\E_{\widehat{\bZ}}U_2^{-2}&\le e^{\lambda\Tr\left(\langle\bR_{11}^\intercal\rangle_{N,\eps}'\langle\bR_{11}\rangle_{N,\eps}'-\langle\bR_{10}^\intercal\bR_{10}\rangle_{N,\eps}'+\tfrac{1}{2}\langle\bR_{11}^\intercal\bR_{11}\rangle_{N,\eps}'\right)},
\\ \E_{\widehat{\bZ}}V_2^{-2}&\le e^{\tfrac{\lambda}{2}\Tr\bT^2},
\end{align*}
where we recall that $\bT=\frac{1}{N}\langle\bX^\intercal\rangle_{N,\eps}'\langle\bX\rangle_{N,\eps}'$ \eqref{Tdef}. As the entries of $\bX,\bX_0$ lie in $[-D,D]$, we have that
\begin{equation*}
|\Tr(\langle\bR_{11}^\intercal\rangle_{N,\eps}'\langle\bR_{11}\rangle_{N,\eps}')|, |\Tr(\bR_{10}^\intercal\bR_{10})|,|\Tr(\bR_{11}^\intercal\bR_{11})|,|\Tr\bT^2|\le M^2D^4
\end{equation*}
and thus
\begin{equation} \label{Lem12_1}
\E_{\widehat{\bZ}}[U_2^{-2}+V_2^{-2}]\le 2e^{\frac{5}{2}\lambda M^2D^4}.
\end{equation}

\textit{Step 2:} Letting $\langle\,\cdot\,\rangle_{N,\eps}'$ denote the Gibbs average with respect to the replicas $\bX_1,\bX_2$ of $\bX$, we have
\begin{align*}
\E_{\widehat{\bZ}}U_2^2&=\E_{\widehat{\bZ}}\left\langle e^{\widehat{y}(\bX_1;\bX_0,\widehat{\bZ},\lambda)+\widehat{y}(\bX_2;\bX_0,\widehat{\bZ},\lambda)}\right\rangle_{N,\eps}'
\\&=\langle F_2(\bR_{10},\bR_{20},\bR_{12})\rangle_{N,\eps}',
\end{align*}
where we have again used identity \eqref{Gaussianint} and defined $F_2:([-D,D]^{M\times M})^3\to\R$ by
\begin{equation*}
F_2(\bA,\bB,\bC):=e^{\frac{\lambda}{2}\Tr\left(\bA^\intercal\bA+\bB^\intercal\bB+\bC^\intercal\bC\right)}.
\end{equation*}
Note also that
\begin{align*}
\E_{\widehat{\bZ}}V_2^2&=\E_{\widehat{\bZ}}e^{2\widehat{y}(\langle\bX\rangle_{N,\eps}';\langle\bX\rangle_{N,\eps}',\widehat{\bZ},\lambda)}
\\&=e^{\frac{3\lambda}{2}\Tr\bT^\intercal\bT}
\\&=F_2(\bT,\bT,\bT).
\end{align*}
Now, on the domain of $F_2$, we have that
\begin{equation*}
\lVert\nabla F_2(\bA,\bB,\bC)\rVert_{\mathrm{F}} \le \frac{\sqrt{3}}{2}\lambda MD^2e^{\frac{3}{2}\lambda M^2D^4}.
\end{equation*}
Thus, $F_2(\bA,\bB,\bC)$ is $(C_1Me^{C_2M^2})$-Lipschitz for some constants $C_1,C_2$ independent of $M,N$. Hence, using Jensen's inequality shows that
\begin{align}
|\E_{\widehat{\bZ}}U_2^2-\E_{\widehat{\bZ}}V_2^2|&=|\langle F_2(\bR_{10},\bR_{20},\bR_{12})\rangle_{N,\eps}'-F_2(\bT,\bT,\bT)| \nonumber
\\&\le \langle |F_2(\bR_{10},\bR_{20},\bR_{12})-F_2(\bT,\bT,\bT)|\rangle_{N,\eps}' \nonumber
\\&\le C_1Me^{C_2M^2}\left\langle\sqrt{\lVert\bR_{10}-\bT\rVert_{\mathrm{F}}^2+\lVert\bR_{20}-\bT\rVert_{\mathrm{F}}^2+\lVert\bR_{12}-\bT\rVert_{\mathrm{F}}^2}\right\rangle_{N,\eps}'. \label{Lem12_2}
\end{align}

\textit{Step 3:} Similarly to above, we see that
\begin{align*}
\E_{\widehat{\bZ}}[U_2V_2]&=\E_{\widehat{\bZ}}\left\langle e^{\widehat{y}(\bX;\bX_0,\widehat{\bZ},\lambda)+\widehat{y}(\langle\bX\rangle_{N,\eps}';\langle\bX\rangle_{N,\eps}',\widehat{\bZ},\lambda)}\right\rangle_{N,\eps}'
\\&=\langle F_2(\bR_{10},\bT,\tfrac{1}{N}\bX^\intercal\langle\bX\rangle_{N,\eps}')\rangle_{N,\eps}'.
\end{align*}
Hence, using Jensen's inequality and inserting our Lipschitz constant shows that
\begin{align}
|\E_{\widehat{\bZ}}[U_2V_2]-\E_{\widehat{\bZ}}V_2^2|&=|\langle F_2(\bR_{10},\bT,\tfrac{1}{N}\bX^\intercal\langle\bX\rangle_{N,\eps}')\rangle_{N,\eps}'-F_2(\bT,\bT,\bT)| \nonumber
\\&\le \langle |F_2(\bR_{10},\bT,\tfrac{1}{N}\bX^\intercal\langle\bX\rangle_{N,\eps}')-F_2(\bT,\bT,\bT)|\rangle_{N,\eps}' \nonumber
\\&\le C_1Me^{C_2M^2}\left\langle\sqrt{\lVert\bR_{10}-\bT\rVert_{\mathrm{F}}^2+\lVert\tfrac{1}{N}\bX^\intercal\langle\bX\rangle_{N,\eps}'-\bT\rVert_{\mathrm{F}}^2}\right\rangle_{N,\eps}'. \label{Lem12_3}
\end{align}

\textit{Step 4:} As inequalities \eqref{Lem12_2}, \eqref{Lem12_3} are equivalent to inequalities \eqref{Lem11_2}, \eqref{Lem11_3}, we may read off from Step 4 of the proof of Lemma \ref{Lemma11} that
\begin{align*}
\E'|\E_{\widehat{\bZ}}U_2^2-\E_{\widehat{\bZ}}V_2^2|&\le C_1Me^{C_2M^2}\sqrt{3\Gamma(N,M)}, \label{Lem12_4.1}
\\ \E'|\E_{\widehat{\bZ}}[U_2V_2]-\E_{\widehat{\bZ}}V_2^2|&\le C_1Me^{C_2M^2}\sqrt{2\Gamma(N,M)}. 
\end{align*}
Therefore, inserting inequality \eqref{Lem12_1} into inequality \eqref{Lem12_0}, using Jensen's inequality, and then finally using the above two inequalities, we see that there is a finite positive constant $C$ independent of $M,N$ such that
\begin{align*}
&\left|\E'\E_{\widehat{\bZ}}\ln U_2-\E'\E_{\widehat{\bZ}}\ln V_2\right|
\\&\le Ce^{CM^2}\E'\sqrt{2|\E_{\widehat{\bZ}}U_2^2-\E_{\widehat{\bZ}}V_2^2|+2|\E_{\widehat{\bZ}}[U_2V_2]-\E_{\widehat{\bZ}}V_2^2|}
\\&\le Ce^{CM^2}\sqrt{\E'|\E_{\widehat{\bZ}}U_2^2-\E_{\widehat{\bZ}}V_2^2|+\E'|\E_{\widehat{\bZ}}[U_2V_2]-\E_{\widehat{\bZ}}V_2^2|}
\\&\le C\sqrt{M}e^{CM^2}\Gamma(N,M)^{1/4}.
\end{align*}
\end{proof}

To conclude this subsection, it remains now to combine lemmas \ref{Lemma11} and \ref{Lemma12} with Lemma \ref{Lemma10}, use the equality in distribution of $\langle\bX^\intercal\rangle_{N,\eps}'\bz^{(N)}$ and $\sqrt{N\bT}\bz^{(M)}$ for $\bz^{(N)}\sim\mathcal{N}(0,I_N)$ and $\bz^{(M)}\sim\mathcal{N}(0,I_M)$ discussed earlier, replace the expectation over $\bT$ with a supremum, and finally apply Theorem \ref{thrm2}.

\begin{proposition}[$\Delta_N/M$ cavity upper bound] \label{prop4}
Assume the hypotheses of Theorem \ref{thrm1}, let $M=M_N$, and recall \eqref{DeltaN5} that
\begin{equation*}
\Delta_N=\E\ln\widetilde{Z}_{N+1,M_{N+1}}-\E\ln\widetilde{Z}_{N,M_{N+1}},
\end{equation*}
where $\E[\,\cdot\,]$ denotes an expectation with respect to $\bZ,\widetilde{\bZ},\bX_0,\bxi,\widetilde{\bxi},\xi,\bet_0,\eps_N,\eps_{N+1}$. There exists a finite positive constant $C$ independent of $M,N$ such that enforcing the constraint $s_N\ge\min\{0.9,N^{-1/9}(Ce^{CM^2})^{8/9}\}$, with $s_N$ still vanishing at large $N$, ensures that
\begin{equation} \label{DeltaNMbound}
\limsup_{N\to\infty}\frac{\Delta_N}{M}\le \sup_{q\in[0,\rho]}F_1^{\mathrm{RS}}(q,\lambda).
\end{equation}
\end{proposition}
\begin{proof}
First, we set $M=M_{N+1}$ instead of $M=M_N$, recalling that Lemma \ref{Lemma10} holds for any $M$. Rewriting said lemma in the notation of lemmas \ref{Lemma11} and \ref{Lemma12} then yields
\begin{equation*}
\frac{\Delta_N}{M}\le\frac{\E\ln U_1-\E_{\widehat{\bZ}}\E\ln U_2}{M}+\Gamma'(N,M),
\end{equation*}
where we have replaced $\bxi^\intercal,\bet^\intercal,\bet_0^\intercal$ by $\bz,\bx,\bx_0$. By lemmas \ref{Lemma11} and \ref{Lemma12}, we then have
\begin{equation} \label{prop4_1}
\frac{\Delta_N}{M}\le\frac{\E\ln V_1-\E_{\widehat{\bZ}}\E\ln V_2}{M}+\frac{C}{\sqrt{M}}e^{CM^2}\Gamma(N,M)^{1/4}+\Gamma'(N,M)
\end{equation}
for some finite positive constant $C$ independent of $M,N$. Now, writing $\bz'\sim\mathcal{N}(0,I_M)$, we observe that $\langle\bX^\intercal\rangle_{N,\eps_N}'\bz$ and $\sqrt{N\bT}\bz'$ are equal in distribution, so
\begin{equation*}
\frac{1}{M}\E\ln V_1=\frac{1}{M}\E_{\bz'}\E\ln\int_{\R^M}e^{\sqrt{\lambda}\bx^\intercal\sqrt{\bT}\bz'+\bx^\intercal\bT\bx_0-\frac{\lambda}{2}\bx^\intercal\bT\bx}\,\de\!\pP_{X,M}(\bx).
\end{equation*}
On the other hand, using the independence of $\widehat{\bZ}$ and $\langle\bX\rangle_{N,\eps_N}'$ shows that
\begin{equation*}
\frac{1}{M}\E_{\widehat{\bZ}}\E\ln V_2=\frac{1}{2M}\Tr\E_{\widehat{\bZ}}\E\left[\sqrt{\lambda}\widehat{\bZ}\bT+\frac{\lambda}{2}\bT^2\right]=\frac{\lambda}{4M}\Tr\E\bT^2.
\end{equation*}
Inserting these two expressions back into inequality \eqref{prop4_1} gives
\begin{align*}
\frac{\Delta_N}{M}&\le\frac{1}{M}\E_{\bz'}\E\ln\int_{\R^M}e^{\sqrt{\lambda}\bx^\intercal\sqrt{\bT}\bz'+\bx^\intercal\bT\bx_0-\frac{\lambda}{2}\bx^\intercal\bT\bx}\,\de\!\pP_{X,M}(\bx)-\frac{\lambda}{4M}\Tr\E\bT^2
\\&\quad+\frac{C}{\sqrt{M}}e^{CM^2}\Gamma(N,M)^{1/4}+\Gamma'(N,M).
\end{align*}
As this integral depends on all of the involved random variables through only $\bT,\bz',\bx_0$, we recognize from equation \eqref{RSpot} that the above simplifies to
\begin{equation*}
\frac{\Delta_N}{M}\le \E_{\bT}F_M^{\mathrm{RS}}(\bT,\lambda)+\frac{C}{\sqrt{M}}e^{CM^2}\Gamma(N,M)^{1/4}+\Gamma'(N,M).
\end{equation*}
Noting that $\bT$ is symmetric positive semidefinite, we circumvent the intractability of the expectation over $\bT$ by replacing it with a supremum to obtain
\begin{align} 
\frac{\Delta_N}{M}&\le \sup_{\bT\in\cS_M}F_M^{\mathrm{RS}}(\bT,\lambda)+\frac{C}{\sqrt{M}}e^{CM^2}\Gamma(N,M)^{1/4}+\Gamma'(N,M) \nonumber
\\&= \sup_{q\in[0,\rho]}F_1^{\mathrm{RS}}(q,\lambda)+\frac{C}{\sqrt{M}}e^{CM^2}\Gamma(N,M)^{1/4}+\Gamma'(N,M), \label{prop4_2}
\end{align}
with the second line following from using Theorem \ref{thrm2} to replace the $M$-dependent supremum $\sup_{\bT\in\cS_M}F_M^{\mathrm{RS}}(\bT,\lambda)$ by its $M$-independent analog. To make connection with inequality \eqref{DeltaNMbound}, we multiply both sides by $M_{N+1}/M_N$ to obtain
\begin{multline*}
\frac{\Delta_N}{M_N}\le \frac{M_{N+1}}{M_N}\sup_{q\in[0,\rho]}F_1^{\mathrm{RS}}(q,\lambda)
\\+\frac{C\sqrt{M_{N+1}}}{M_N}e^{CM_{N+1}^2}\Gamma(N,M_{N+1})^{1/4}+\frac{M_{N+1}}{M_N}\Gamma'(N,M_{N+1}).
\end{multline*}
By hypothesis, either $M_{N+1}=M_N$ or $M_{N+1}=M_N+1$. In either case, the result follows due to our choice of $s_N$ inducing the vanishing of all error terms when taking the limit supremum; the second term in the right-hand side vanishes because $M_N=\mathrm{o}(\sqrt{\ln N})$ and $s_N\ge N^{-1/9}(Ce^{CM_N^2})^{8/9}$, while the final term vanishes because $s_N\to0$.
\end{proof}

\subsection{The novel column cavity bound} \label{s5.3}
In this subsection, we formalize the intuition contained within the approximation \eqref{DeltaMapprox} to prove the counterpart to Proposition \ref{prop4} bounding the limit supremum of $\Delta_M/N$. This is a novel term that only comes into play when examining multiscale models and arises here due to our working in the sublinear rank regime $M=\mathrm{o}(\sqrt{\ln N})$. As per the approximation \eqref{DeltaMapprox}, our approach is to bootstrap the finite $M,N$ cavity bounds through a second application of the cavity method.

The strategy then is to fix $N$ and $M=M_N$ at a value such that $M_{N+1}=M_N+1$ and bound each of the terms in \eqref{DeltaM5}
\begin{equation*}
\frac{\Delta_M}{N}=\frac{1}{N}\E\ln\widetilde{Z}_{N,M+1}-\frac{1}{N}\E\ln\widetilde{Z}_{N,M}
\end{equation*}
separately. The upper bound on the second term, being negative, can be obtained simply from the free entropy lower bound of Proposition \ref{prop2} -- note that this proposition holds for any $M,N\ge 1$. To obtain an upper bound on $\frac{1}{N}\E\ln\widetilde{Z}_{N,M+1}$, we proceed by introducing the auxiliary difference
\begin{equation} \label{DeltaNn}
\widehat{\Delta}_N(n):=\E\ln\widetilde{Z}_{n+1,M+1}-\E\ln\widetilde{Z}_{n,M+1}
\end{equation}
so that we have the telescoping series
\begin{equation*}
\frac{1}{N}\E\ln\widetilde{Z}_{N,M+1}=\frac{1}{N}\sum_{n=0}^{N-1}\widehat{\Delta}_N(n).
\end{equation*}
Then, we truncate this sum at a truncation parameter $n=T$ and show that the truncated sum can be bounded by a standard application of the cavity method, while the remainder is bounded with $T$-dependent error due to Proposition \ref{prop4}. Finally, we take the limit supremum, choosing $T$ such that all error terms vanish at large $N$.

\begin{lemma}[Bounding the truncated sum of $\widehat{\Delta}_N(n)$] \label{Lemma13}
Let $N,M,T\ge1$, assume hypotheses \ref{H2} and \ref{H3}, and let $\widehat{\Delta}_N(n)$ be defined by equation \eqref{DeltaNn}. Then, there exists some positive constant $C$ independent of $M,N,T$ such that
\begin{equation*}
\sum_{n=0}^{T-1}\widehat{\Delta}_N(n) \le CTM^2.
\end{equation*}
\end{lemma}
\begin{proof}
First, note that
\begin{equation} \label{Lem13_0}
\sum_{n=0}^{T-1}\widehat{\Delta}_N(n)\le T\sup_{n\le T}\widehat{\Delta}_N(n).
\end{equation}
Then, replacing $N$ by $n$ and $M$ by $M+1$ in Lemma \ref{Lemma10} shows that for $C$ a generic finite positive constant independent of $N,M,T$ here on out,
\begin{multline} \label{Lem13_1}
\widehat{\Delta}_N(n)\le \E_{\bxi',\bet_0}\E\ln\left\langle\int_{\R^{M+1}}e^{\bet\widehat{z}(\bX';\bxi',\lambda)+\bet\widehat{s}(\bX';\lambda)\bet^\intercal}\,\de\!\pP_{X,M+1}(\bet)\right\rangle_{n,\eps_n}'
\\-\E_{\widehat{\bZ}'}\E\ln\langle e^{\widehat{y}(\bX';\bX_0',\widehat{\bZ}',\lambda)}\rangle_{n,\eps_n}'+(M+1)\Gamma'(n,M+1),
\end{multline}
where we define $\bX'\in\R^{n\times (M+1)}$ to be the integration variable of the Gibbs average $\langle\,\cdot\,\rangle_{n,\eps_n}'$, let $\bX_0'\in\R^{n\times (M+1)}$ have i.i.d.~entries drawn from $\pP_X$, let $\bet_0\sim\pP_{X,M+1}$, let $\widehat{\bZ}'\in\R^{n\times n}$ be a standard Wigner matrix, and let $\bxi'\sim\mathcal{N}(0,I_n)$.

Subsituting $\widehat{z}(\bX')$ \eqref{zcavity5} and $\widehat{s}(\bX')$ \eqref{scavity5} into the first term of the right-hand side of inequality \eqref{Lem13_1}, using Jensen's inequality, and then using the identity $\E_ze^{az}=e^{a^2/2}$ for $z\sim\mathcal{N}(0,1)$ and $a$ independent of $z$ shows that
\begin{align*}
&\E_{\bxi',\bet_0}\E\ln\left\langle\int_{\R^{M+1}}e^{\bet\widehat{z}(\bX';\bxi',\lambda)+\bet\widehat{s}(\bX';\lambda)\bet^\intercal}\,\de\!\pP_{X,M+1}(\bet)\right\rangle_{n,\eps_n}'
\\&= \E_{\bxi',\bet_0}\E\ln\left\langle\int_{\R^{M+1}}e^{\sqrt{\frac{\lambda}{n}}\bet\bX'^\intercal\bxi'^\intercal+\frac{\lambda}{n}\bet\bX'^\intercal\bX_0'\bet_0-\frac{\lambda}{2n}\bet\bX'^\intercal\bX'\bet^\intercal}\,\de\!\pP_{X,M+1}(\bet)\right\rangle_{n,\eps_n}'
\\&\le \ln\E_{\bet_0}\E\left\langle\int_{\R^{M+1}}\E_{\bxi'}e^{\sqrt{\frac{\lambda}{n}}\bet\bX'^\intercal\bxi'^\intercal+\frac{\lambda}{n}\bet\bX'^\intercal\bX_0'\bet_0-\frac{\lambda}{2n}\bet\bX'^\intercal\bX'\bet^\intercal}\,\de\!\pP_{X,M+1}(\bet)\right\rangle_{n,\eps_n}'
\\&= \ln\E_{\bet_0}\E\left\langle\int_{\R^{M+1}}e^{\frac{\lambda}{n}\bet\bX'^\intercal\bX_0'\bet_0}\,\de\!\pP_{X,M+1}(\bet)\right\rangle_{n,\eps_n}'.
\end{align*}
As $\pP_X$ has $D$-bounded support, we have that $|\bet\bX'^\intercal\bX_0'\bet_0|\le n(M+1)^2D^4$ and therefore,
\begin{align}
&\E_{\bxi',\bet_0}\E\ln\left\langle\int_{\R^{M+1}}e^{\bet\widehat{z}(\bX';\bxi',\lambda)+\bet\widehat{s}(\bX';\lambda)\bet^\intercal}\,\de\!\pP_{X,M+1}(\bet)\right\rangle_{n,\eps_n}' \nonumber
\\&\le \ln\E_{\bet_0}\E\left\langle\int_{\R^{M+1}}e^{\frac{\lambda}{n}\bet\bX'^\intercal\bX_0'\bet_0}\,\de\!\pP_{X,M+1}(\bet)\right\rangle_{n,\eps_n}' \nonumber
\\&\le \ln\E_{\bet_0}\E\left\langle\int_{\R^{M+1}}e^{\lambda (M+1)^2D^4}\,\de\!\pP_{X,M+1}(\bet)\right\rangle_{n,\eps_n}' \nonumber
\\&= \lambda (M+1)^2D^4. \label{Lem13_2}
\end{align}

Likewise, substituting $\widehat{y}(\bX')$ \eqref{yhat} into the second term of the right-hand side of inequality \eqref{Lem13_1}, using Jensen's inequality, using the independence of $\widehat{\bZ}'$ and $\bX'$, and then using the $D$-boundedness of $\supp\pP_X$ shows that
\begin{align}
&-\E_{\widehat{\bZ}'}\E\ln\langle e^{\widehat{y}(\bX';\bX_0',\widehat{\bZ}',\lambda)}\rangle_{n,\eps_n}' \nonumber
\\&=-\E_{\widehat{\bZ}'}\E\ln\left\langle e^{\frac{1}{2}\Tr\left(\frac{\sqrt{\lambda}}{n}\widehat{\bZ}'\bX'\bX'^\intercal+\frac{\lambda}{n^2}\bX_0'^\intercal\bX'\bX'^\intercal\bX_0'-\frac{\lambda}{2n^2}\bX'^\intercal\bX'\bX'^\intercal\bX'\right)}\right\rangle_{n,\eps_n}' \nonumber
\\&\le -\E_{\widehat{\bZ}'}\E\left\langle \frac{1}{2}\Tr\left(\frac{\sqrt{\lambda}}{n}\widehat{\bZ}'\bX'\bX'^\intercal+\frac{\lambda}{n^2}\bX_0'^\intercal\bX'\bX'^\intercal\bX_0'-\frac{\lambda}{2n^2}\bX'^\intercal\bX'\bX'^\intercal\bX'\right)\right\rangle_{n,\eps_n}' \nonumber
\\&=-\E\left\langle \frac{1}{2}\Tr\left(\frac{\lambda}{n^2}\bX_0'^\intercal\bX'\bX'^\intercal\bX_0'-\frac{\lambda}{2n^2}\bX'^\intercal\bX'\bX'^\intercal\bX'\right)\right\rangle_{n,\eps_n}' \nonumber
\\&\le \frac{3\lambda (M+1)^2D^4}{4}. \label{Lem13_3}
\end{align}

Combining inequalities \eqref{Lem13_2} and \eqref{Lem13_3} with the fact that the remaining error term $(M+1)\Gamma'(n,M+1)$ in inequality \eqref{Lem13_1} is at most quadratic in $M$ for all choices of $n$, said inequality simplifies to give
\begin{equation*}
\widehat{\Delta}_N(n)\le CM^2.
\end{equation*}
Substituting this into inequality \eqref{Lem13_0} completes the proof.
\end{proof}

The next proposition allows us to deduce the $\Delta_M$ cavity bound from the $\Delta_N$ cavity bound at a slight loss of the rate, provided that we have a rank-one equivalence formula. 
\begin{proposition}[$\Delta_M/N$ cavity upper bound] \label{prop5}
Assume hypotheses \ref{H2}--\ref{H5} and recall \eqref{DeltaM5} that for $N$ and $M=M_N$ such that $M_{N+1}=M+1$,
\begin{equation*}
\Delta_M=\E\ln\widetilde{Z}_{N,M+1}-\E\ln\widetilde{Z}_{N,M},
\end{equation*}
where $\E[\,\cdot\,]$ denotes an expectation with respect to $\bZ,\widetilde{\bZ},\bX_0,\eps_N$. Furthermore, suppose that
\begin{equation}\label{eq:generic_bound}
\frac{\Delta_N}{M+1} \leq \sup_{q\in[0,\rho]}F_1^{\mathrm{RS}}(q,\lambda) + \Phi(N,M)
\end{equation}
for some error function $\Phi(N,M) : \N^2 \to \R_+$ decreasing in $N$. Then, for any $T \leq N$, there exists a finite positive constant $C$ independent of $M,N,T$ such that
\begin{equation} \label{DeltaMerror}
\frac{\Delta_M}{N}\le \sup_{q\in[0,\rho]}F_1^{\mathrm{RS}}(q,\lambda) + (M + 1)\Phi(T,M) + \frac{CTM^2}{N}.
\end{equation}
\end{proposition}
\begin{proof}
By Lemma \ref{Lemma9} and Proposition \ref{prop2}, we have for each $N,M\ge1$ that
\begin{equation*}
\frac{1}{N}\E\ln\widetilde{Z}_{N,M}\ge MF_N(\lambda)\ge M\sup_{q\in[0,\rho]}F_1^{\mathrm{RS}}(q,\lambda).
\end{equation*}
Hence, it is immediate that
\begin{equation} \label{prop5_1}
\frac{\Delta_M}{N}\le\frac{1}{N}\E\ln\widetilde{Z}_{N,M+1}-M\sup_{q\in[0,\rho]}F_1^{\mathrm{RS}}(q,\lambda).
\end{equation}
Recalling that $\widehat{\Delta}_N(n)=\E\ln\widetilde{Z}_{n+1,M+1}-\E\ln\widetilde{Z}_{n,M+1}$ \eqref{DeltaNn} and letting $1\le T\le N-1$, we have the telescoping series
\begin{align}
\frac{1}{N}\E\ln\widetilde{Z}_{N,M+1}&=\frac{1}{N}\sum_{n=0}^{N-1}\widehat{\Delta}_N(n) \nonumber
\\&=\frac{1}{N}\sum_{n=0}^{T-1}\widehat{\Delta}_N(n)+\frac{1}{N}\sum_{n=T}^{N-1}\widehat{\Delta}_N(n) \nonumber
\\&\le \left(1-\frac{T}{N}\right)\sup_{n\ge T}\widehat{\Delta}_N(n)+\frac{CTM^2}{N}, \label{prop5_2}
\end{align}
where we have used Lemma \ref{Lemma13}. Replacing $N$ by $n$ in inequality \eqref{eq:generic_bound} shows that for all $n\ge T$,
\begin{equation*}
\widehat{\Delta}_N(n)\le (M+1)\sup_{q\in[0,\rho]}F_1^{\mathrm{RS}}(q,\lambda)
 + (M + 1) \Phi(T,M).
\end{equation*}
As $T\le N-1$, substituting this into inequality \eqref{prop5_2} produces
\begin{equation*}
\frac{1}{N}\E\ln\widetilde{Z}_{N,M+1}\le (M+1)\sup_{q\in[0,\rho]}F_1^{\mathrm{RS}}(q,\lambda)+ (M + 1) \Phi(T,M) +\frac{CTM^2}{N},
\end{equation*}
which further combines with inequality \eqref{prop5_1} to yield
\begin{equation*}
\frac{\Delta_M}{N}\le\sup_{q\in[0,\rho]}F_1^{\mathrm{RS}}(q,\lambda)+ (M + 1) \Phi(T,M) +\frac{CTM^2}{N},
\end{equation*}
as required.
\end{proof}

A key property of this upper bound is that it is sharp when one is able to take $T$ such that both $T = \mathrm{o}(N M_N^{-2})$ and $M_N \Phi( T ,M_N) \to 0$ at large $N$. Indeed, a sharp bound is made possible upon taking $M_N=\mathrm{o}(\sqrt{\ln N})$ and setting
\begin{equation} \label{Phichoice}
\Phi(N,M) = \frac{C}{\sqrt{M+1}}e^{C(M+1)^2}\Gamma(N,M+1)^{1/4}+\Gamma'(N,M+1),
\end{equation}
as prescribed by inequality \eqref{prop4_2} (recalling that $M$ therein stands for $M_{N+1}$).
\begin{corollary} \label{Cor3}
Assume the hypotheses of Theorem \ref{thrm1} and recall \eqref{DeltaM5} that for $N$ and $M=M_N$ such that $M_{N+1}\ne M$,
\begin{equation*}
\Delta_M=\E\ln\widetilde{Z}_{N,M+1}-\E\ln\widetilde{Z}_{N,M},
\end{equation*}
where $\E[\,\cdot\,]$ denotes an expectation with respect to $\bZ,\widetilde{\bZ},\bX_0,\eps_N$. There exists a finite positive constant $C$ independent of $M,N$ such that choosing $s_N$ to decay to zero while satisfying $s_N\ge\min\{0.9,(C^2\ln N/(8C+9))^4N^{-1/(8C+9)}\}$ ensures that
\begin{equation*}
\limsup_{N\to\infty}\frac{\Delta_M}{N}\le \sup_{q\in[0,\rho]}F_1^{\mathrm{RS}}(q,\lambda).
\end{equation*}
\end{corollary}
\begin{proof}
Inequality \eqref{prop4_2} stems from the same hypotheses of Proposition \ref{prop4} as those being assumed here, so we read from said inequality that
\begin{equation*}
\frac{\Delta_N}{M+1}\le \sup_{q\in[0,\rho]}F_1^{\mathrm{RS}}(q,\lambda)+\frac{C}{\sqrt{M+1}}e^{C(M+1)^2}\Gamma(N,M+1)^{1/4}+\Gamma'(N,M+1).
\end{equation*}
Thus, we may take $\Phi(N,M)$ within Proposition \ref{prop5} to be as in equation \eqref{Phichoice} so that 
\begin{multline} \label{prop5_3}
	\frac{\Delta_M}{N}\le\sup_{q\in[0,\rho]}F_1^{\mathrm{RS}}(q,\lambda)
	\\+ C\sqrt{M+1}e^{C(M+1)^2}\Gamma(T,M+1)^{1/4}+\Gamma'(T,M+1) +\frac{CTM^2}{N}
\end{multline}
for any $T\le N$.
		
It remains now to choose $T$ such that the error terms in the above vanish in the large $N$ limit. Hence, let $\epsilon=1/(8C+10)$ and set $T=N^{1-\epsilon}$. Since $M=\mathrm{o}(\sqrt{\ln N})$ by hypothesis \ref{H1}, we have that for $N$ large enough, $M<M+1\le\sqrt{\ln N^{\epsilon}}$ and thus
\begin{align*}
C\sqrt{M+1}e^{C(M+1)^2}\Gamma(T,M+1)^{1/4}&=e^{C(M+1)^2}\frac{C(M+1)}{T^{1/8}s_T^{1/8}}
\\&\le C\sqrt{\epsilon}\sqrt{\ln N}N^{(C+1/8)\epsilon-1/8}s_T^{-1/8},
\\ (M+1)\Gamma'(T,M+1)&=C\left(\frac{M+1-T}{T+1}+1+s_T\right)
\\&\le C\left(\frac{\sqrt{\epsilon}\sqrt{\ln N}-T}{T+1}+1+s_T\right),
\\\frac{CTM^2}{N}&\le C\epsilon N^{-\epsilon}\ln N.
\end{align*}
By our choice of $s_N$, we have that $s_T\ge(C^2\tfrac{\epsilon}{1-\epsilon}\ln T)^4T^{\epsilon/(\epsilon-1)}$ in the thermodynamic limit. Thus, inserting $T=N^{1-\epsilon}$ into the first error term above shows that it is bounded above by $N^{-\epsilon}$, hence vanishes at large $N$. As $s_T\to0$ and $\sqrt{\ln N}\ll T=N^{1-\epsilon}$, the second error term also vanishes at large $N$. Likewise for the final error term. Therefore, taking the limit supremum in inequality \eqref{prop5_3} produces the stated result.
\end{proof}
	
We conclude this subsection with a demonstration of how the arguments underlying Proposition \ref{prop5} can be extended to the case where $\Phi(N,M)$ vanishes in the $M=\mathrm{o}(N^{\gamma/2})$ regime for some $0<\gamma\le1$ (in line with the rate of concentration in Theorem \ref{thrm4}), highlighting the fact that the bottleneck posed by hypothesis \ref{H1} is purely due to the method of proof used in \S\ref{s5.2}.
\begin{lemma}
If there exists $0 < \gamma \le 1$ such that
\begin{equation*}
\Phi(N,M) = \frac{M^2}{N^\gamma},
\end{equation*}
then for $M = \mathrm{o}(N^{\gamma/(2\gamma+3)})$, there exists $T$ a sequence in $N$ such that
\begin{equation*}
\max\bigg\{ M\Phi(T,M), \frac{TM^2}{N} \bigg\} \to 0. 
\end{equation*}
\end{lemma}
\begin{proof}
We show that $M = \mathrm{o}(N^{\gamma/(2\gamma+3)})$ is the optimal rate for which there exists $T$ such that the error terms in question vanish at large $N$. We set $M = N^x$ and $T = N^y$ for some yet to be determined constants $0 \le x,y \le 1$. We have that
\begin{equation*}
M\Phi(T,M) = \frac{M^3}{N^\gamma} \to 0  \iff x < \frac{\gamma y}{3}
\end{equation*}
and
\begin{equation*}
\frac{TM^2}{N} \to 0 \iff x < \frac{1 - y}{2}.
\end{equation*}
The curves bounding these regions intersect precisely when 
\begin{equation*}
y = \frac{3}{2\gamma+3}.
\end{equation*}
Therefore, for any $\epsilon > 0$ sufficiently small, both error terms limit to zero upon setting $x = \gamma/(2\gamma+3) - \epsilon$ and $y = 3/(2\gamma+3) - \epsilon$. 
\end{proof}
\begin{remark}
Notice that even upon assuming that a rate of $M = \mathrm{o}(N^{\frac{\gamma}{2}})$ is sufficient for the error term $\Phi(N,M)$ in the $\Delta_N$ cavity bound \eqref{eq:generic_bound} to vanish at large $N$, one requires the slower rate $M = \mathrm{o}(N^{\gamma/(2\gamma+3)})$ for the equivalent error term of the accompanying $\Delta_M$ cavity bound \eqref{DeltaMerror} to also vanish as $N\to\infty$.
\end{remark}

\subsection{Proof of Proposition \ref{prop3} and Theorem \ref{thrm1}} \label{s5.4}
We now conclude this paper by giving a brief demonstration of how the above results constitute a proof of Theorem \ref{thrm1}. The majority of the argument is in fact a proof of Proposition \ref{prop3}, which combines with Proposition \ref{prop2}, proven in \S\ref{s4.1}, to give Theorem \ref{thrm1}.

\begin{proof}[\textbf{Proof of Theorem \ref{thrm1}}]
Throughout this paper, $C$ has stood for a generic finite positive constant independent of $N,M$, and the truncation parameter $T$ of Lemma \ref{Lemma13}. Noting that each inequality of this paper involving $C$ remains true upon arbitrarily increasing $C$, we now take $C$ to be large enough such that every such inequality is true for the same choice of $C$. Then, setting
\begin{equation*}
s_N=\min\left\{0.9,\max\left\{N^{-1/9}(Ce^{CM^2})^{8/9},\left(C^2\ln N/(8C+9)\right)^4N^{-1/(8C+9)}\right\}\right\}
\end{equation*}
in order to satisfy the hypotheses of Proposition \ref{prop4} and Corollary \ref{Cor3}, we have by said proposition and corollary that
\begin{align*}
\limsup_{N\to\infty}\frac{\Delta_N}{M}&\le\sup_{q\in[0,\rho]}F_1^{\mathrm{RS}}(q,\lambda),
\\ \limsup_{N\to\infty}\frac{\Delta_M}{N}&\le\sup_{q\in[0,\rho]}F_1^{\mathrm{RS}}(q,\lambda).
\end{align*}
Thus, by Theorem \ref{thrm3}, there exists some $0\le\alpha\le1$ such that
\begin{align*}
\limsup_{N\to\infty}\widetilde F_N(\lambda)&\leq \alpha\limsup_{N\to\infty}\frac{\Delta_N}{M}+(1-\alpha)\limsup_{N\to\infty}\frac{\Delta_M}{N}
\\&\leq\sup_{q\in[0,\rho]}F_1^{\mathrm{RS}}(q,\lambda).
\end{align*}

Proposition~\ref{prop3} hence follows from inequality \eqref{Lemma9cor}, which we recall is a straightforward consequence of Lemma \ref{Lemma9}. Combining Proposition \ref{prop3} with the result of taking the limit infimum in Proposition \ref{prop2} finally shows that
\begin{equation*}
\limsup_{N\to\infty}F_N(\lambda)\le\sup_{q\in[0,\rho]}F_1^{\mathrm{RS}}(q,\lambda)\le \liminf_{N\to\infty}F_N(\lambda).
\end{equation*}
This is equivalent to Theorem \ref{thrm1}.
\end{proof}

\begin{appendices}

\section{The replica symmetric potential in terms of mutual information}\label{appA}
Consider the vector Gaussian channel
\begin{equation} \label{vectorchannel}
\by=\sqrt{\lambda\bQ}\bx_0+\bz
\end{equation}
with $\bz\in\mathbb{R}^M$ a standard Gaussian vector and $\bx_0\sim\pP_{X,M}$. Inspection of the replica symmetric potential \eqref{RSpot} reveals that it is closely related to the free entropy of this vector channel. Indeed, the intuition behind the replica method is that the size $N$, rank $M$ spiked Wigner model \eqref{spikedwignermodel} is, up to the regularization term $-\tfrac{\lambda}{4M}\Tr\bQ^2$, statistically equivalent to the mean field model corresponding to $N$ replicas of the channel \eqref{vectorchannel} --- the quadratic interactions within the $\bX_0\bX_0^\intercal$ term of \eqref{spikedwignermodel} are accounted for by the order parameter $\bQ$, which is eventually seen to play the role of the averaged overlap matrix $\E_{\bz,\bx_0}\langle\bx\bx_0^\intercal\rangle_{\rm RS}$ (see Lemma \ref{Lemma2}).

To rewrite the replica symmetric potential in the form \eqref{RSpotMI}, we study the vector channel \eqref{vectorchannel} from an information theoretic viewpoint. Thus, begin by recalling that the \textit{differential entropy} \cite{coverthomas} of a random vector $\bX\sim p_X(\bX)$ is defined as
\begin{equation*}
H(\bX):=-\int_{\R^M}\ln p_X(\bX)\,\de p_X(\bX).
\end{equation*}
Then, the mutual information between $\bx_0$ and the output $\by$ of the channel \eqref{vectorchannel} is
\begin{equation} \label{MIentropy}
I(\bx_0;\by)=H(\by)-H(\by\mid\bx_0)=H(\sqrt{\lambda\bQ}\bx_0+\bz)-H(\bz);
\end{equation}
it is a standard result that $H(\bz)=M\ln(2\pi e)/2$. By Bayes' rule, the first term in the right-hand side of \eqref{MIentropy} reads
\begin{align*}
H(\sqrt{\lambda\bQ}\bx_0+\bz)&=H(\by)
\\&=-\E_{\bz,\bx_0}\ln\int_{\R^M}\pP_{Y|X}(\by(\bx_0,\bz)\mid\bx)\,\de\!\pP_{X,M}(\bx).
\end{align*}
Conditional on $\bx_0$, the data $\by$ is Gaussian distributed with mean $\sqrt{\lambda\bQ}\bx_0$, so the above expression for the entropy reduces further to
\begin{align*}
H(\sqrt{\lambda\bQ}\bx_0+\bz)&=-\E_{\bz,\bx_0}\ln\int_{\R^M}\exp\left(-\tfrac{1}{2}\|\sqrt{\lambda\bQ}\bx_0+\bz-\sqrt{\lambda\bQ}\bx\|_{\rm F}^2\right)\,\de\!\pP_{X,M}(\bx)
\\&\qquad+\frac{M}{2}\ln(2\pi)
\\&=-\E_{\bz,\bx_0}\ln\int_{\R^M}\exp\left(\sqrt{\lambda}\bx^\intercal\sqrt{\bQ}\bz+\lambda\bx_0^\intercal\bQ\bx-\tfrac{\lambda}{2}\bx^\intercal\bQ\bx\right)\,\de\!\pP_{X,M}(\bx)
\\&\qquad+\frac{1}{2}\E_{\bz,\bx_0}\|\sqrt{\lambda\bQ}\bx_0+\bz\|_{\rm F}^2+\frac{M}{2}\ln(2\pi).
\end{align*}
Since $\bx_0$ and $\bz$ are independent and $\bz$ has identity covariance, we have that
\begin{equation*}
\E_{\bz,\bx_0}\|\sqrt{\lambda\bQ}\bx_0+\bz\|_{\rm F}^2=\lambda\E_{\bz,\bx_0}\left[\bx_0^\intercal\bQ\bx_0\right]+M.
\end{equation*}
Moreover, we have as hypothesis that our prior distribution factorizes as $\pP_{X,M}=\pP_X^{\otimes M}$ according to \eqref{iidprior} and $\pP_X$ has variance $\rho$, so
\begin{equation*}
\E_{\bz,\bx_0}\left[\bx_0^\intercal\bQ\bx_0\right]=\lambda\Tr\bQ\E\left[\bx_0\bx_0^\intercal\right]=\lambda\rho\Tr\bQ.
\end{equation*}
Thus, returning to equation \eqref{MIentropy}, we find that
\begin{equation*}
I(\bx_0;\by)=-\E_{\bz,\bx_0}\ln\int_{\R^M}e^{\sqrt{\lambda}\bx^\intercal\sqrt{\bQ}\bz+\lambda\bx_0^\intercal\bQ\bx-\tfrac{\lambda}{2}\bx^\intercal\bQ\bx}\,\de\!\pP_{X,M}(\bx)+\frac{\lambda\rho}{2}\Tr\bQ.
\end{equation*}
Comparing this expression to the definition \eqref{RSpot} of $F_M^{\rm RS}(\bQ,\lambda)$, and completing the square to write $\Tr\bQ^2-2\rho\Tr\bQ=\|\bQ-\rho I_M\|_{\rm F}^2-M\rho^2$ finally shows that
\begin{align*}
	F^{\rm RS}_M(\bQ,\lambda)&=-\frac1M I(\bx_0;\sqrt{\lambda\bQ}\bx_0+\bz)+\frac{\lambda\rho }{2M} \Tr \bQ - \frac{\lambda}{4M}  \Tr\bQ^2\\
	&=-\frac1M I(\bx_0;\sqrt{\lambda\bQ}\bx_0+\bz) - \frac{\lambda}{4M}  \|\bQ-\rho I_M\|_{\rm F}^2 +\frac{\lambda\rho^2 }{4}.
\end{align*}

\end{appendices}

\subsection*{Acknowledgments}
J.~Barbier and A.~A.~Rahman were funded by the European Union (ERC, CHORAL, project number 101039794). J.~Ko was funded by the European Union (ERC, LDRAM, project number 884584) and the Natural Sciences and Engineering Research Council of Canada (RGPIN-2020-04597). Views and opinions expressed are however those of the authors only and do not necessarily reflect those of the European Union or the European Research Council. Neither the European Union nor the granting authority can be held responsible for them. The authors acknowledge crucial discussions with Francesco Camilli and Koki Okajima which led to the joint development of the multiscale cavity method. J.~Barbier wishes to thank Nicolas Macris for the uncountable discussions concerning this problem and many others over the years.

\bibliographystyle{unsrtnat}
\bibliography{refs}

\end{document}